\documentclass{article}

\RequirePackage{algorithm}
\RequirePackage{algpseudocode}

\usepackage{graphicx}

\usepackage{amsmath}
\usepackage{amssymb}
\usepackage{mathtools}
\usepackage{amsthm}

\usepackage[capitalize,noabbrev]{cleveref}

\usepackage{dsfont}
\usepackage{tikz}
\usetikzlibrary{spy}
\usepackage[titletoc,toc,page]{appendix}
\usepackage{pgfplots, pgfplotstable}
\pgfplotsset{compat=1.18}
\usepackage{relsize}
\usepackage{subcaption}
\usepackage{stmaryrd}

\PassOptionsToPackage{numbers, compress}{natbib}

\newcommand{\bbE}[0]{\mathbb{E}}
\newcommand{\bbR}[0]{\mathbb{R}}
\newcommand{\OO}{\mathcal{O}}  %
\newcommand{\R}{\mathbb{R}}  %
\newcommand{\setX}{\mathcal{X}}

\DeclarePairedDelimiter{\abs}{\lvert}{\rvert}%
\DeclarePairedDelimiterX{\norm}[1]{\lVert}{\rVert}{#1}

\DeclareMathOperator{\vect}{vec}
\DeclareMathOperator{\unvec}{unvec}
\DeclareMathOperator{\diag}{diag}
\DeclareMathOperator{\KL}{KL}
\DeclareMathOperator{\uKL}{uKL}
\DeclareMathOperator{\argmin}{argmin}

\theoremstyle{plain}
\newtheorem{theorem}{Theorem}[section]
\newtheorem{proposition}[theorem]{Proposition}
\newtheorem{lemma}[theorem]{Lemma}

\theoremstyle{definition}

\newtheorem{assumption}[theorem]{Assumption}
\theoremstyle{remark}
\newtheorem{remark}[theorem]{Remark}

\usepackage[final]{neurips_2025}

\usepackage[utf8]{inputenc} %
\usepackage[T1]{fontenc}    %
\usepackage{url}            %
\usepackage{amsfonts}       %
\usepackage{nicefrac}       %
\usepackage{xcolor}         %

\title{Least squares variational inference}

\author{%
    Yvann Le Fay\textsuperscript{\rm 1}, ~Nicolas Chopin\textsuperscript{\rm 1,\rm} \thanks{Corresponding author}, ~Simon Barthelmé\textsuperscript{\rm 2} \\
    \textsuperscript{1} ENSAE, CREST, IP Paris \\
    \textsuperscript{2} GIPSA-Lab, CNRS \\
    \texttt{\{yvann.lefay,nicolas.chopin\}@ensae.fr} \\
    \texttt{simon.barthelme@gipsa-lab.fr}
}

\begin{document}

    \maketitle

    \begin{abstract}
        Variational inference seeks the best approximation of a target distribution within a chosen family, where "best" means minimising Kullback-Leibler divergence.
        When the approximation family is exponential, the optimal approximation satisfies a fixed-point equation.
        We introduce LSVI (Least Squares Variational Inference), a gradient-free, Monte Carlo-based scheme for the fixed-point recursion, where each iteration boils down to performing ordinary least squares regression on tempered log-target evaluations under the variational approximation.
        We show that LSVI is equivalent to biased stochastic natural gradient descent and use this to derive convergence rates with respect to the numbers of samples and iterations.
        When the approximation family is Gaussian, LSVI involves inverting the Fisher information matrix, whose size grows quadratically with dimension $d$.
        We exploit the regression formulation to eliminate the need for this inversion, yielding $\OO(d^3)$ complexity in the full-covariance case and $\OO(d)$ in the mean-field case.
        Finally, we numerically demonstrate LSVI’s performance on various tasks,
        including logistic regression, discrete variable selection, and Bayesian
        synthetic likelihood, showing results competitive with state-of-the-art methods, even when gradients are unavailable.
    \end{abstract}

    \addtocontents{toc}{\protect\setcounter{tocdepth}{0}}

    \section{Introduction}\label{sec:Motivation}

    This paper focuses on parametric variational inference (VI,~\citep{Jordan1999,variational_review,8588399}).
    Given an (unnormalised) target density $\pi$, we aim at finding the distribution that
    minimises the (reverse) Kullback-Leibler divergence:
    \begin{equation}
        \label{eq:KLmin} \arg\min_{q \in \mathcal{Q}} \KL\left(q \mid \bar{\pi}
        \right)   \coloneq \int q \log\left( q / \bar{\pi}\right)
    \end{equation}
    where $\mathcal{Q}$ is a user-chosen parametric family (e.g., Gaussians), and
    $\bar{\pi}=\pi/\int \pi$. This approach has become a de facto standard in
    probabilistic machine learning in recent years and is implemented in various
    software packages, such as STAN, NumPyro, PyMC3, and Blackjax~\citep{JSSv076i01,phan2019composableeffectsflexibleaccelerated,
        Salvatier2016,cabezas2024blackjaxcomposablebayesianinference}. The minimisation is
    typically carried out through gradient-based procedures using automatic differentiation,
    either stochastic gradient descent (SGD,~\citep[][]{pmlr-v33-ranganath14})---often
    applied after reparameterising the target distribution~\citep{advi,10.5555/3044805.3045112}---or its faster alternative
    natural gradient descent (NGD,~\citep[][]{Amari1998NaturalGW, pmlr-v54-khan17a, 10.23919/ISITA.2018.8664326, JMLR:v24:22-0291}).
    This is convenient for users, as they only have to provide the function $f\coloneq \log\pi$ to the software.

    These procedures use different gradient estimators;
    some require $\log \pi$ to be amenable to automatic differentiation, which is the case
    when using a reparameterisation, while others only require gradient estimators of
    expectations under the variational distribution via
    the log-derivative trick~\citep{Williams1992}.
    The gradient estimator for expectations usually suffers from
    high variance, and practical implementations rely on the
    reparameterisation trick, which is not possible in several important cases,
    for instance when $\pi$ is a discrete distribution, or when $\pi$
    is intractable or non-differentiable (as in likelihood-free inference).
    Additionally, convergence of SGD is sometimes slow and/or tedious to assess,
    and requires careful step sizes tuning~\citep{WelandaweetalFramework} while
    a naive implementation of NGD requires costly matrix inversions.

    \subsection{Outline and contributions}
    We introduce practical algorithms for VI within exponential families when gradients of $\log \pi$
    are unavailable. These algorithms involve taking biased stochastic
    gradient descent steps, but we show both theoretical convergence and good
    performance in non-toy problems. In Section~\ref{sec:lsvi}, we derive an exact,
    but intractable, iteration we call LSVI, that boils down to performing successive
    least squares (OLS) regression.
    We highlight connections to NGD and discuss its convergence properties.
    In Section~\ref{sec:stochastic}, we
    introduce a stochastic variant that updates the OLS estimate using multiple
    draws from the current approximation.
    Importantly, under standard smoothness
    and relative convexity assumptions on the objective, and bounded-moment
    assumptions on the variational family, we establish convergence guarantees
    and rates with respect to the numbers of draws and iterations, conditioned
    on high-probability events. In addition, we provide an adaptive method to
    calibrate step sizes by controlling the linear regression residuals.
    Section~\ref{sec:gaussian} focuses on the Gaussian variational family; we
    propose a reparametrisation of the linear regression such that the
    OLS procedure requires no inversion of the Fisher
    information matrix (FIM).
    These schemes tailored to Gaussian distributions are cost-efficient: our
    methods scale linearly with $d$ in the mean-field case, and in the
    full-covariance case, the cost matches the cost of computing $d\times d$
    matrix products, i.e., $\OO(d^3)$. In
    Section~\ref{sec:experiments}, we extensively illustrate the performance of
    our methods compared to other inference procedures, including
    gradient-based VI and exact Bayesian inference procedures. Limitations are
    discussed in Section~\ref{sec:limitations}. We provide a Python package
    supporting GPU parallelisation via JAX to replicate the experiments:
    \url{https://github.com/ylefay/LSVI}.

    \section{Exact LSVI}\label{sec:Fixed point iteration}
    \label{sec:lsvi}

    Let $\pi:\setX \to \R$ be some unnormalised target density, with $\setX \subset \R^d$.
    It will be convenient to work with an exponential family
    $\mathcal{Q}$ of \emph{unnormalised} densities:
    \[ q_\eta(x) \coloneq \exp\{ \eta^\top s(x)\},
    \qquad \eta \in \mathcal{V} \coloneq \{\eta: Z(q_\eta) < \infty \}
    \]
    where $\eta\in \mathcal{V}$ is the natural parameter associated to $q_\eta\in \mathcal{Q}$,
    $Z(q) \coloneq \int_{\setX} q$ denotes the partition function,
    and $s:\setX\to \bbR^m$ the extended statistic function defined as
    \[
        s(x) = \begin{pmatrix}
                   1 \\
                   \bar{s}(x)
        \end{pmatrix},
        \quad \bar{s}:\setX \to \bbR^{m-1}.
    \]
    In words, we include an intercept in $s$ to make the family closed under multiplication
    by a positive scalar.
    For $\eta=\left(\eta^{(0)}, \bar{\eta}^\top\right)^\top$,
    where $\eta^{(0)}$ denotes the first component of $\eta$, let $\bar{q}_{\bar{\eta}}$ be
    the normalised version of $q_\eta$ (which therefore depends only on
    $\bar{\eta}$): $\bar{q}_{\bar{\eta}} = q_\eta / Z_\eta$, using the short-hand
    $Z_\eta \coloneq Z(q_\eta)$. Notation $\bbE_\eta[\cdot]$ means a properly normalised expectation, i.e.
    $\bbE_\eta[h] = \int_{\setX} q_\eta h / Z_{\eta}$.
    Likewise, we replace the standard Kullback-Leibler objective with a divergence for
    unnormalised densities~\citep{Minka2005DivergenceMA}, which is defined by
    \begin{equation}
        \label{eq:def_ukl}
        \uKL(q\mid \pi ) \coloneq
        \int q\log\left(\frac{q}{\pi}\right) + Z(\pi) - Z(q),
    \end{equation}
    for any density $q$ absolutely continuous with respect to $\pi$.
    In addition, we assume the variational family $\mathcal{Q}$ is
    minimal and regular, which is a standard assumption in VI~\citep{pmlr-v54-khan17a,pmlr-v80-khan18a,pmlr-v97-lin19b,wu2024understandingstochasticnaturalgradient},
    and is met by any standard exponential families (e.g., Gaussian, Beta, Poisson, Bernoulli, etc.,~\citep[][Table 3.1]{wainwright2008graphical}). These assumptions
    ensure $\eta\in \mathcal{V}\mapsto q_{\eta}$ is injective and the log-partition function is differentiable everywhere~\citep[][Prop. 3.1]{wainwright2008graphical}.

    \begin{assumption}[Minimality and regularity of $\mathcal{Q}$]
        \label{ass:minimality_and_regularity}
        The components of $s$ are linearly independent (minimality), and the set of natural parameters $\mathcal{V}$ is open (regularity).
    \end{assumption}

    The next proposition shows that the critical points of the uKL
    divergence are also critical points of the KL divergence, and vice versa.
    In words, nothing is lost by considering the uKL instead of KL.

    \begin{proposition}
        \label{prop:min_ukl_kl}
        Let $\eta = (\eta^{(0)}, \bar{\eta}^\top)^\top \in \mathcal{V}$, if $\nabla_{\eta} \uKL(q_{\eta}\mid \pi)=0$ then $\nabla_{\bar{\eta}} \KL(\bar{q}_{\bar{\eta}}\mid \bar{\pi}) = 0$,
        and the reciprocal holds:  $\nabla_{\bar{\eta}} \KL(\bar{q}_{\bar{\eta}}\mid \bar{\pi}) = 0$ and $\partial_{\eta^{(0)}}\uKL(q_{\eta}\mid \pi) = 0$, then $\nabla_{\eta} \uKL(q_{\eta}\mid \pi) = 0$.
    \end{proposition}

    The first-order condition of the uKL minimisation problem is given by the following proposition.
    \begin{proposition}
        \label{prop:first_order_optimality}
        Let $f = \log \pi$ be the (unnormalised) log target density.
        Let $\eta = (\eta^{(0)}, \bar{\eta}^{\top})^{\top}\in \mathcal{V}$, $\nabla_{\eta} \uKL(q_{\eta}\mid \pi) = 0$ if and only if $\left\{\bbE_{\eta}[ss^{\top}]\right\} \eta = \bbE_{\eta}[fs]$.
        Furthermore, if $\nabla_{\eta}\uKL(q_{\eta}\mid \pi) = 0$, then $\eta^{(0)} = -\KL(\bar{q}_{\bar{\eta}}\mid \bar{\pi})
        + \log \left(Z(\pi) / \int_{\setX} \exp\left( \bar{\eta}^\top
                                                   \bar{s}\right)\right).$
    \end{proposition}

    \subsection{The exact LSVI scheme}
    \label{subsec:exact_lsvi}
    The first-order optimality condition is equivalent to the fixed point equation: $\eta
    = \phi(\eta)$ with
    \begin{equation}
        \label{eq:def_phi}
        \phi(\eta) \coloneq F_{\eta}^{-1} z_\eta,
        \quad F_\eta \coloneq \bbE_\eta[s s^\top],
        \quad z_\eta \coloneq \bbE_\eta[f s],
    \end{equation}
    and $F_{\eta}$ is the Fisher information matrix (FIM) associated to $q_{\eta}$.
    \citet{10.1214/13-BA858} remark that $\phi$ is the ordinary least squares regressor~(OLS; ~\citep{Penrose_1956}) of $f(X)$ with respect to $s(X)$ when $X\sim q_{\eta}$:
    \begin{equation}
        \label{eq:LSVI_as_OLS}
        \phi(\eta) = \argmin_{\beta\in \bbR^m}
        \bbE_{\eta } \left[\left\{\beta^{\top}s(X) - f(X)\right\}^2\right].
    \end{equation}

    A nice property of $\phi$ when $\pi$ is in the variational family with  $\pi=q_{\eta^{\star}}$,
    is that for any $\eta\in \mathcal{V}$, $\phi(\eta) = \eta^{\star}$, i.e., $\phi$ exactly recovers $\pi$.
    However, in general $\phi(\eta)$ may not be in $\mathcal{V}$, and naively performing a fixed-point
    scheme can lead to unstable variational approximations or, worse, result in non-normalisable densities (i.e.,
    $\phi(\eta)\notin \mathcal{V}$). To address this, we consider a relaxation of the fixed-point scheme obtained
    via a momentum fixed-point iteration~\citep{Bauschke20172}:
    \begin{equation}
        \label{eq:krasnoselskii-mann} \eta_{t+1} \coloneq \varepsilon_t \phi(\eta_t)
        + (1-\varepsilon_t) \eta_t, \qquad t\geq 0
    \end{equation}
    where $\varepsilon_t>0$ is such that $\eta_{t+1}$ is in $\mathcal{V}$. Such
    an $\varepsilon_t$ necessarily exists because $\mathcal{V}$ is open (Assumption~\ref{ass:minimality_and_regularity}).
    Since iteration~\eqref{eq:krasnoselskii-mann} assumes that one has
    access to expectations under the variational family (which in general is not
    the case), we refer to~\eqref{eq:krasnoselskii-mann} as the \emph{exact}
    Least Squares Variational Inference (LSVI) iteration.
    This relaxation has a natural interpretation in this specific
    context: $\eta_{t+1}$ in~\eqref{eq:krasnoselskii-mann} is the solution of the
    least squares objective~\eqref{eq:LSVI_as_OLS}
    when $\pi=\exp f$ is replaced by the tempered (annealed) density
    $q_{\eta_t}^{1-\varepsilon_t}\pi^{\varepsilon_t}$.

    \subsection{LSVI as natural gradient descent (NGD) and mirror descent (MD)}
    \label{subsec:ngd_md}
    This subsection summarises a well-established connection between NGD and MD
    in the variational inference literature~\citep{pmlr-v54-khan17a,
        pmlr-v80-khan18a, wu2024understandingstochasticnaturalgradient} but generalised to the \emph{unnormalised} KL divergence.

    Let us define the (unnormalised) moment parameter mapping $\omega:\eta\in\mathcal{V}\mapsto \nabla_\eta Z_\eta = \int s(x) q_{\eta}(x)$,
    and let $\mathcal{W}=\omega(\mathcal{V})$ be the set of moment parameters.
    We denote by $\eta :
    \mathcal{W}\to\mathcal{V}$ the inverse mapping of $\omega :\mathcal{V}\to\mathcal{W}$, whose existence is guaranteed under Assumption~\ref{ass:minimality_and_regularity}~\citep[][Ch. 3]{wainwright2008graphical}.
    Define $l$ as the unnormalised KL divergence~\eqref{eq:def_ukl}. When expressed in
    natural parameters, we write $l:\eta\in\mathcal{V}\mapsto \uKL(q_\eta\mid \pi)$,
    when expressed in moment parameters, we write $l:\omega\in\mathcal{W}\mapsto
    \uKL(q_\omega\mid\pi)$, and similarly for expectations: $\bbE_\omega \coloneq \bbE_{\eta(\omega)}$.

    The following proposition states that LSVI
    iteration~\eqref{eq:krasnoselskii-mann} is a NGD
    iteration on the uKL divergence in the natural space of parameters, and
    equivalently a MD in the moment space~\citep[][Ch. 3]{NemirovskyYudin1983,doi:10.1137/16M1099546}.

    \begin{proposition}[LSVI is NGD which is equivalent to MD,\protect{~\cite[Lemma 1]{wu2024understandingstochasticnaturalgradient}}]
        \label{prop:gradient_descent}
        Under Assumption~\ref{ass:minimality_and_regularity} and
        provided the sequence $(\eta_t)$ defined by~\eqref{eq:krasnoselskii-mann}
        is in $\mathcal{V}$, $(\eta_t)$ satisfies the dynamic (NGD),
        \begin{equation}
            \label{eq:first_iteration_gd}
            \eta_{t+1} = \eta_{t} - \varepsilon_t F_{\eta_{t}}^{-1} \nabla_{\eta} l(\eta_{t})/Z_{\eta_{t}},
        \end{equation}
        or equivalently,
        \begin{equation}
            \label{eq:second_iteration_gd}
            \eta_{t+1} = \eta_{t} - \varepsilon_t \nabla_{\omega} l(\omega(\eta_{t})).
        \end{equation}
        Furthermore, let $\omega_0 \in \mathcal{W}$ and define for $t\geq 0$ (MD),
        \begin{equation}
            \label{eq:mirror_descent}
            \omega_{t+1} \coloneq \argmin_{\omega\in\mathcal{W}} \left\{  \nabla^{\top}_{\omega}l(\omega_t) \omega + \varepsilon_t^{-1}D_{Z^*}(\omega, \omega_t)\right\},
        \end{equation}
        where $D_{Z^{\star}}$ is the Bregman divergence~\citep{BREGMAN1967200} with respect to $Z^*$ the Legendre transform of $Z$: $Z^*(\omega)=\sup_{\eta\in\mathcal{V}}\{\eta^{\top}\omega-Z_\eta\}$.
        Then, the sequence $(\eta_t)$ defined by~\eqref{eq:krasnoselskii-mann} with $\eta_0 = \eta(\omega_0)$ satisfies for all $t\geq 0$, $\eta_t = \eta(\omega_t)$.
        In words, LSVI performs a natural gradient step in the space of natural parameters, which corresponds to a mirror descent step in the dual (moment) space.
    \end{proposition}

    Proposition~\ref{prop:gradient_descent} allows us to leverage known convergence results for
    MD under standard smoothness and convexity assumptions on the uKL objective~\citep[in the VI literature, see,][]{pmlr-v54-khan17a, 10.5555/3020948.3020982, NIPS2015_3214a6d8}.
    \begin{assumption}
        \label{ass:convexsmooth}
        The $\uKL$ objective $l:\omega\in \mathcal{W}\mapsto \uKL(q_{\omega}\mid \pi)$ is $L$-smooth, $\mu$-strongly convex relative to $D_{Z^*}$.
    \end{assumption}

    Under Assumption~\ref{ass:convexsmooth}, MD is known to converge with rate
    $O(1/k)$ for sufficiently small and linearly decreasing step sizes: $\varepsilon_t=(L+\alpha t)^{-1}$ for $0\leq \alpha <\mu$,~\citep[see, e.g.,][Theorem 4.5 and Lemma 4.8, Theorem 4]{Hanzely2021, aubin2024}.
    The non-strongly convex case ($\mu=0$) exhibits $O(1/\sqrt{k})$ convergence rate for a specific choice of step sizes~\citep[see, e.g.,][Corollary 4.6]{Hanzely2021}.
    In practice, it is not trivial to set the $\varepsilon$ to obtain a $O(1/k)$ rate, as the relative strong convexity parameter $\mu$, if it exists, might be unknown and eventually very small.
    \begin{remark}
        The strong-convexity/convexity assumptions are rarely verified in practice, however, such assumptions are standard for analysing convergence of
        optimisation algorithms (including NGD and MD), to ensure a unique minimiser and tractable rates~\citep[see, e.g.,][Ch. 5]{Bach2024}.
        While non-conjugate VI objectives may not be globally convex~\citep{wu2024understandingstochasticnaturalgradient} (but it holds when the variational family contains the target), \textit{local} convexity near optima often suffices for local convergence to hold.
    \end{remark}
    See~\citep{pmlr-v119-domke20a,NEURIPS2023_d0bcff64} for provable smoothness guarantees on the KL objective.

    \section{Practical algorithms and their analysis}
    \label{sec:stochastic}

    The exact LSVI mapping $\phi$ assumed that one has access to expectations
    under the variational family. In practice, exact computation of those
    expectations is intractable for a general target log-density $f$. In this
    section, we introduce a practical algorithm in which these expectations are
    estimated via Monte Carlo, and we study the impact of the Monte Carlo error
    on the convergence guarantees.

    \subsection{Generic LSVI}

    Our first algorithm comes down to replacing the two expectations in~\eqref{eq:def_phi} with
    Monte Carlo estimates:
    \begin{equation}
        \label{eq:two_mc_estimates}
        \hat{F}_{\eta}  \coloneq \frac{1}{N}\sum_{i=1}^{N} s(X_i)s(X_i)^{\top},\quad
        \hat{z}_{\eta}  \coloneq \frac{1}{N}\sum_{i=1}^{N} f(X_i)s(X_i),
    \end{equation}
    where $X_1, \dots, X_N\overset{\textup{i.i.d.}}{\sim} q_{\eta}$.
    The counterpart to the exact iteration~\eqref{eq:krasnoselskii-mann} is then
    \begin{equation}
        \label{eq:km_lsvi_stochastic}
        \hat{\eta}_{t+1} \coloneq \varepsilon_t \hat{F}_{\hat{\eta}_{t}}^{-1}\hat{z}_{\hat{\eta}_{t}} + (1-\varepsilon_t)\hat{\eta}_t,
    \end{equation}
    with $\hat{\eta}_0 = \eta_0$.
    At any iteration $t\geq 1$, the step size $\varepsilon_t$ can, in all generality, depend on the current state of the algorithm via a function \texttt{stepsize}.
    We discuss one possible choice in Section~\ref{sec:choice_eps}.
    This leads naturally to generic LSVI~\cref{alg:lsvi_krasnoselskii},
    whose one-iteration cost is $\OO(m^3+m^2N)$.

    \begin{algorithm}[t]
        \caption{Generic LSVI (any family $\mathcal{Q}$)}
        \label{alg:lsvi_krasnoselskii}
        \begin{algorithmic}
            \Require $\eta_0 \in \mathcal{V}$, $N \geq 1$
            \State $\hat{\eta}_0 \gets \eta_0$
            \While{not converged}
                \State $X_1, \dots, X_N \sim q_{\hat{\eta}_t}$
                \State $\hat{F}_{\hat{\eta}_t} \gets \frac{1}{N} \sum_{i=1}^{N} s(X_i) s^\top(X_i)$
                \State $\hat{z}_{\hat{\eta}_t} \gets \frac{1}{N} \sum_{i=1}^{N} s(X_i) f(X_i)$
                \State $\hat{\eta}'_{t+1} \gets \hat{F}^{-1}_{\hat{\eta}_t} \hat{z}_{\hat{\eta}_t}$ \algorithmiccomment{ordinary least squares estimator (OLS)}
                \State $\varepsilon_t \gets \texttt{stepsize}(\hat{F}_{\hat{\eta}_t}, \hat{z}_{\hat{\eta}_t}, \hat{\eta}'_{t+1}, \hat{\eta}_t, X)$
                \State $\hat{\eta}_{t+1} \gets \varepsilon_t \hat{\eta}'_{t+1} + (1 - \varepsilon_t)\hat{\eta}_t$
            \EndWhile
        \end{algorithmic}
    \end{algorithm}

    Iteration~\eqref{eq:km_lsvi_stochastic} replaces the exact computation
    of $F_\eta^{-1}$ with a Monte Carlo estimate $\hat{F}_{\eta}^{-1}$. This
    approximation introduces a bias in the estimation of the inverse FIM, and
    consequently, in the estimation of the natural gradient involved in~\eqref{eq:first_iteration_gd}. Further analysis
    of the statistical properties of the sequence $(\hat{\eta}_t)$, in
    particular, its convergence toward a neighbourhood of the optimum, requires a
    careful control of the bias.
    We assume that the iterates remain in a compact set.
    \begin{assumption}[]
        \label{ass:bounded_iterates}
        There exists a compact convex set $\mathcal{W}_0\subset \mathcal{W}$ such that the iterates $\hat{\omega}_t=\omega(\hat{\eta}_t)$ defined by~\eqref{eq:km_lsvi_stochastic}, and the exact mirror-descent iterates~\eqref{eq:mirror_descent} starting from them, remain in $\mathcal{W}_0$.
    \end{assumption}
    When $s$ admits uniformly bounded fourth-order moment, and the spectrum of $F_\eta$ is bounded away from zero, the bias conditioned to a high-probability event can be controlled.
    Both properties are a consequence of Assumptions~\ref{ass:minimality_and_regularity} and~\ref{ass:bounded_iterates}, since i) $\mathcal{W}_0$ is a compact set, and ii) $s(X)$ admits moments of all orders, for $X\sim q_{\omega}$ and for all $\omega\in \mathcal{W}$, which are continuous in $\omega$, and iii) $F_\omega$ is invertible (minimality).
    That is, the sufficient statistic $s$ admits uniformly bounded fourth-order moments:
    \begin{equation}
        \label{eq:fourth_order_moment}
        \mu_4 \coloneq \sup_{\omega\in \mathcal{W}_0}\max_{1\leq i\leq m}\left(\bbE_{\omega}\left[\abs{s(X)_i}^4\right]\right)^{1/4} < \infty,\quad    \nu \coloneq \sup_{\omega\in \mathcal{W}_0} \sup_{\norm{u}=1, u\in\bbR^m} \left(\bbE_{\omega}\left[\abs{u^{\top} s(X)}^4\right]\right)^{1/4}<\infty,
    \end{equation}
    and the smallest spectral value $r \coloneq \inf_{\omega\in \mathcal{W}_0} \norm{F^{-1}_{\omega}}^{-1}$ is strictly positive.
    While $\mathcal{W}$ is generally not a compact set, Assumption~\ref{ass:bounded_iterates} should not be considered
    as a limiting assumption in practice, and can be lifted, see~\citep{NEURIPS2020_bd4d08cd,pmlr-v162-scaman22a,batardière2024importancesamplingbasedgradientmethod}.
    We further assume that $f$ admits uniformly bounded fourth-order moment as this is required to control the norm and the variance of $\hat{z}_\omega$.

    \begin{assumption}[]
        \label{ass:moment_two_f}  $\sup_{\omega\in \mathcal{W}_0} \bbE_{\omega}[f^4] < \infty$, and we let $m_2 \coloneq \sup_{\omega\in \mathcal{W}_0} \bbE_{\omega}[f^2]^{1/2}$.
    \end{assumption}
    We derive the convergence in expectation to the minimum of the KL loss conditioned on the event that the estimated FIMs are well-conditioned.

    \begin{theorem}[Explicit convergence rates for LSVI]
        \label{th:main_theorem}
        Assume~\ref{ass:minimality_and_regularity},~\ref{ass:convexsmooth},~\ref{ass:bounded_iterates}, and~\ref{ass:moment_two_f}.
        Let $k\geq 0$, and let $\hat{\eta}_0, \hat{\eta}_1,\ldots, \hat{\eta}_k$ be given by~\eqref{eq:km_lsvi_stochastic}, with $\hat{\omega}_t=\omega(\hat{\eta}_t)$ for $0\leq t\leq k$, and step sizes $\varepsilon_t\leq 1/L$.
        Let $\mathcal{A}_k = \cap_{t=0}^{k}\mathcal{A}(\hat{\omega}_t)$ with $\mathcal{A}(\omega) = [\norm{F_\omega-\hat{F}_\omega}\leq \tfrac{1}{2}\norm{F^{-1}_{\omega}}^{-1}]$.
        Further assume that at each iteration $t\geq 0$, the quantities $\hat{F}_{\hat{\eta}_t}$ and $\hat{z}_{\hat{\eta}_t}$  are computed using two independent sets of samples.
        Let $c_t = c_{t-1}\varepsilon_{t-1}^{-1}(\varepsilon_t^{-1}-\mu)^{-1}$ for $t\geq 1$, $c_0 = 1$, $C_k = \sum_{t=1}^{k}c_{t-1}$. Let $\bar{\omega}_k =
        \frac{1}{C_k}\sum_{t=1}^{k} c_{t-1} \hat{\omega}_{t}$ be the weighted average
        of the iterates, and let $\omega^*$ be the minimiser of $l$.
        \begin{enumerate}
            \item Fix $\delta\in (0, 1)$, provided $\sqrt{N}\geq C_0 r^{-1}(k+1)^{1/2}\delta^{-1/2} m^{1/2}(\sqrt{\log(m)} \mu_4 \nu + m\log(m)\mu_4^2)$ for some constant $C_0>0$, $\mathcal{A}_k$ happens with probability at least $1-\delta$.
            \item We have
            \begin{equation}
                \label{eq:eq_theorem}
                \bbE[(l(\bar{\omega}_k)-l(\omega^*))\mathds{1}_{\mathcal{A}_k}] \leq \frac{(\varepsilon_0^{-1}-\mu)\uKL(q_{\omega^*}\mid q_{\omega_0})}{C_k}
                +\OO\left(\frac{1}{N}\right)\sum_{t=0}^{k-1} \frac{c_t\varepsilon_t}{C_k} + \OO\left(\frac{1}{N}\right),
            \end{equation}
            where the big-$\OO$ terms are independent of $k$.
            \item Let $\varepsilon_t^{-1} = L + \alpha t$ for some $\alpha > 0$.
            The RHS of~\eqref{eq:eq_theorem} has asymptotic convergence rates that depend on $\alpha$ compared to the strong-convexity parameter $\mu$.
            When $\alpha >\mu$, the sequence $(c_t)$ is strictly decreasing, and the rate is $\OO\left(k^{-\mu/\alpha}\right)+ \OO\left(N^{-1}\right)$.
            When $\alpha=\mu$, the sequence $(c_t)$ is constant, and the rate is $\OO(k^{-1}) + \OO\left(\log(k)k^{-1}N^{-1}\right)+ \OO\left(N^{-1}\right)$.
            When $\alpha <\mu$, the sequence $(c_t)$ is strictly increasing, and the rate is $\OO\left(k^{-\mu/\alpha}\right)+\OO\left(k^{-1}N^{-1}\right) + \OO\left(N^{-1}\right)$.
        \end{enumerate}
    \end{theorem}
    \begin{remark}
        Our proof follows a similar strategy to that of~\citet{Hanzely2021}, extending their mirror descent lemma to biased estimates.
        We control both the bias and the variance of the FIM estimate, conditionally on the event that the estimated FIM is well-conditioned ($\mathcal{A}_k$).
        We show this occurs with high probability when $N$ is sufficiently large, using concentration inequalities for positive-definite matrices~\citep{10.1093/imaiai/ias001}.
    \end{remark}
    \begin{remark}
        The convergence guarantees can be decomposed in three terms.
        The first term is due to initialisation and vanishes as $k\to \infty$, the third term is the Monte Carlo bias and vanishes as $N\to \infty$, and the second is a
        cross term and vanishes whenever $k\to\infty$ or $N\to\infty$.
    \end{remark}
    \begin{remark}
        The OLS estimate to the regression problem uses a single set of samples to compute both $\hat{F}_{\hat{\eta}_t}$ and $\hat{z}_{\hat{\eta}_t}$, contrary to the estimate introduced in the previous theorem.
        Additionally, for many exponential families, closed-form expressions for $F_{\eta}$ are known.
        Since the OLS is optimal with respect to the variance, it exhibits lower variance compared to others estimates.
        Importantly, it is inefficient to use two distinct set of samples or to replace the estimated FIM with the exact FIM~\citep[][]{10.1214/13-BA858}.
    \end{remark}

    \subsection{The choice of the $\varepsilon_t$'s}
    \label{sec:choice_eps}
    Setting $\varepsilon$ to a small enough and linearly decreasing sequence of step sizes ensures convergence of the sequence~\eqref{eq:km_lsvi_stochastic}
    to a neighbourhood of a local minimizer $\eta^\star$~\citep{Hanzely2021, aubin2024}, see~\cref{th:main_theorem}.
    However, the smoothness and strong-convexity parameters $(L,\mu)$ of the KL objective, if they exist, are rarely known in practice~\citep{pmlr-v119-domke20a, NEURIPS2023_d0bcff64}.
    For these reasons, choosing the $\varepsilon_t$ can be a tedious task, as in any
    stochastic optimisation scheme~\citep{WelandaweetalFramework}: step sizes
    that are too large lead to unstable behaviours while too small step sizes lead to slow convergence.

    Let $\eta\in \mathcal{V}$ and $\eta^{\star} = \phi(\eta)$ be the OLS, consider the following linear regression objective,
    \begin{equation}
        \label{eq:regression_heuristic2}
        f(X_i) = \eta^{\star\top}s(X_i) + v_i,\quad X_1, \dots, X_N\overset{\textup{i.i.d.}}{\sim} q_{\eta},
    \end{equation}
    where $v_i$ is the residual of the regression.
    Then~\eqref{eq:regression_heuristic2} implies that for any $\varepsilon \in (0, 1]$
    \begin{equation}
        \label{eq:ols_epsilon_ss}
        \varepsilon f(X_i) + (1-\varepsilon)\eta^{\top}s(X_i) =
        (\varepsilon \phi(\eta) + (1-\varepsilon)\eta)^{\top}s(X_i) + \varepsilon v_i.
    \end{equation}
    The previous equation~\eqref{eq:ols_epsilon_ss} shows that descending toward the direction of the OLS with step size $\varepsilon$ multiplies
    the variance of the residuals $v_1, \ldots, v_N$, $v^2$ by $\varepsilon^2$.
    Let $u^2$ be some upper bound on the variance of the residuals, and let $\varepsilon\leq u/v$, then the residuals have variance less than $u^2$.
    This remark, combined with a backtracking procedure to ensure that the iterates remain in the set of natural parameters, yields
    an adaptive schedule for choosing the step sizes (Algorithm~\ref{alg:recipe_backracking_variance_control} in Appendix~\ref{appendix:algorithm_eps}),
    which we have found to be robust against noisy iterates and slow descents.

    \section{Gaussian families}\label{sec:gaussian}

    The two most commonly-used families $\mathcal{Q}$ in variational inference are
    the full-covariance Gaussian family ($N_d(\mu, \Sigma)$ with arbitrary $\mu$
    and $\Sigma\succ 0$), and the mean-field Gaussian family ($\Sigma$ is diagonal).
    A single iteration of LSVI requires inverting the Fisher information matrix (FIM) $F$, which is too expensive to be practical in high-dimension; i.e., $\OO(m^3)$, with $m=\OO(d)$ (resp. $m=\OO(d^2)$)
    in the mean-field (resp. full-covariance) case.

    Attempts to lessen the computational complexity of inference procedures over Gaussian distributions
    either rely on access to cheap gradient estimates in the space of moments~\citep{10.23919/ISITA.2018.8664326,pmlr-v80-khan18a,wu2024understandingstochasticnaturalgradient,
        10.1093/jrsssb/qkaf001}, on single draw updates making the FIM estimate cheap to invert but noisy~\citep{10.23919/ISITA.2018.8664326,Godichon-Baggioni26092024},
    or on restrictive assumptions on the target density~\citep{NIPS1997_e816c635,opper2009thevariational}.
    We derive closed-form formulae for the natural gradient descent iteration
    whose cost, in the full-covariance case, essentially amounts to the cost of
    computing $d\times d$ matrix products, that is $\OO(d^3)$. In the mean-field case, the cost is $\OO(d)$.

    \paragraph{Full-covariance Gaussian family}
    Let $\mathcal{Q}$ be the family of (unnormalised) Gaussian densities of dimension $d$.
    The sufficient statistic is $s(x)\coloneq (1, x^{\top}, (\vect{(xx^{\top})})^{\top})^{\top}\in \bbR^{m}$ with $m=d(d+1)+1$, where $\vect{(xx^\top)}$ denotes the vector obtained by vertically stacking the columns of $xx^{\top}$, and we denote by $\unvec$ the inverse operation.
    Consider a natural parameter $\eta = (\eta^{(0)}, \eta^{(1),\top},
    \eta^{(2),\top})^{\top}\in \mathcal{V}$ with $\eta^{(0)}\in \bbR$, $\eta^{(1)}\in
    \bbR^d$ and $\eta^{(2)}\in \bbR^{d^2}$, then it defines a unique Gaussian distribution with mean and covariance matrix given by
    \begin{equation}
    (\mu, \Sigma)
        =\left(-\frac{1}{2}\unvec(\eta^{(2)})^{-1}\eta^{(1)},
        -\frac{1}{2}\unvec(\eta^{(2)})^{-1}\right).
    \end{equation}
    We reparameterise the linear regression of $f(X)$ with respect to $s(X)$, where $X\sim N(\mu,\Sigma)$, into a regression
    of $f(\mu+C Z)$ with respect to $t(Z)$, where $Z\sim N(0, I_d)$, and $C=\textup{Chol}(\Sigma)$ is the Cholesky of $\Sigma$, and
    \begin{equation}
        \label{eq:statistic_full_gaussian_case}
        t(z) \coloneq
        \bigg(
        1,
        z^{\top},
        \frac{z_1^2-1}{\sqrt{2}},
        z_1 z_2,
        \dots,
        z_1 z_d,
        \frac{z_2^2-1}{\sqrt{2}},
        z_2 z_3,
        \dots,
        \frac{z_d^2-1}{\sqrt{2}}
        \bigg)^{\top},
    \end{equation}
    and
    \begin{equation}
        \label{eq:ols_fullcov}
        \gamma \coloneq \argmin_{\gamma\in\bbR^{1+d+d(d+1)/2}}\bbE\left[\{\gamma^{\top}t(Z)-f(\mu+C Z)\}^2\right].
    \end{equation}
    In brief, $t$ is a one-to-one transformation such that the output vector has un-correlated components: $\bbE[t(Z)t^{\top}(Z)]=I$.
    That makes possible the estimation of $\gamma$ without inverting the FIM.
    The explicit mapping from $\gamma$ to $\eta$ depending on $(\mu, \Sigma)$ is given by the next theorem.

    \begin{theorem}
    [LSVI mapping $\phi$ for full-covariance Gaussian distributions]
        \label{th:full_cov_gaussian}
        Let $\eta\in \mathcal{V}$ defines a Gaussian distribution $X\sim \mathcal{N}(\mu, \Sigma)$, and let $C=\textup{Chol}(\Sigma)$ be the Cholesky of $\Sigma$.
        Then, $\beta\coloneq \phi(\eta)$ is defined recursively from bottom to top by
        \begin{equation}
            \label{eq:mappingfullcov}
            \beta = \begin{pmatrix}
                        \beta^{(0)}\\ \beta^{(1)} \\ \beta^{(2)}
            \end{pmatrix}=
            \begin{pmatrix}
                \gamma^{(0)} - \sum_{i=1}^{d}\Gamma_{i,i} - \beta^{(1),\top}\mu - \beta^{(2),\top}\vect{\mu\mu^{\top}} \\
                C^{-\top}\gamma^{(1)}-2\unvec{\beta^{(2)}}\mu                                         \\
                \vect{\left(C^{-\top}\Gamma C^{-1}\right)}
            \end{pmatrix},
        \end{equation}

        where $\gamma = \bbE [t(Z)f(\mu+CZ)]$ has subcomponents $\gamma=(\gamma^{(0)}, \gamma^{(1), \top},\gamma^{(2), \top})^{\top}$, $\gamma^{(0)}\in \bbR$, $\gamma^{(1)}\in\bbR^d$,$\gamma^{(2)}\in \bbR^{d(d+1)/2}$, and where $\Gamma$ is the symmetric matrix given componentwise by $\Gamma_{i,i} = \gamma^{(2)}_{1+1/2(2d+2-i)(i-1)}/\sqrt{2}$, $\Gamma_{i,i+k}=\gamma^{(2)}_{ 1+1/2(2d+2-i)(i-1)+k}/2$ for $1\leq i\leq d$ and $1\leq k\leq d-i$.
        In addition, if $f$ has second-order derivatives such that $\norm{\bbE_{X}[\nabla f]}< \infty$ and $0\prec
        -\bbE_{X}\left[\nabla^2 f\right]$, then $\phi(\eta)$ defines a Gaussian distribution with mean and covariance given by
        \begin{align}
        (\mu', \Sigma')
            = \left(\mu - \bbE\left[ \nabla^2 f(X)\right]^{-1}
                  \bbE \left[\nabla f(X)\right],  - \bbE\left[ \nabla^2 f(X)\right]^{-1}\right),\quad X\sim N(\mu,\Sigma).
        \end{align}
    \end{theorem}

    Theorem~\ref{th:full_cov_gaussian} gives
    the regressor with respect to $s(X)$ of $f(X)$, as a function of $(\mu,C)$
    and $\gamma$. Furthermore, all the involved operations have cost dominated by the computation of $C$,
    which is the same as computing products of $d\times d$ matrices, $O(d^3)$.

    \paragraph{Mean-field Gaussian family}

    The
    family of mean-field Gaussian distributions is treated similarly to the previous one by
    removing the cross-terms $z_i z_j$ in the sufficient statistic. The total cost of the OLS computation is $\OO(d)$.
    See Appendix~\ref{appendix:mf} for the explicit regression procedure.

    \subsection{Stochastic schemes tailored to Gaussian distributions}
    We now take advantage of the reparametrisation tricks previously introduced
    to derive tailored implementations of LSVI for
    Gaussian variational families, with optimal one-iteration cost in $d$.

    An unbiased estimate of the OLS~\eqref{eq:ols_fullcov} is given by
    \begin{equation}
        \hat{\gamma} = N^{-1}\sum_{i=1}^{N}
        t(Z_i)f(\mu+C Z_i),\qquad Z_1,\dots, Z_N
        \overset{\text{i.i.d.}}{\sim} N(0, I_d).
    \end{equation}
    We define $\hat{\eta}$ as the estimate obtained by plugging $\hat{\gamma}$ into~\eqref{eq:mappingfullcov} of Theorem~\ref{th:full_cov_gaussian}.
    The mean-field case is treated in a similar manner.
    See Algorithms~\ref{alg:lsvi_krasnoselskii_mfg} and~\ref{alg:lsvi_krasnoselskii_fullcov}.
    \begin{figure}[H]
        \centering
        \begin{minipage}{0.48\textwidth}
            \begin{algorithm}[H]
                \caption{LSVI-MF (mean-field Gaussian family)}\label{alg:lsvi_krasnoselskii_mfg}
                \begin{algorithmic}
                    \Require $(\mu_0, \sigma_0^2): \sigma_{0,i} > 0, \, i \in [1, d]$, $N \geq 1$
                    \State $(\hat{\mu}_0, \hat{\sigma}^2_0) \gets (\mu_0, \sigma_0^2)$
                    \State $\hat{\eta}_0 \gets (0, \hat{\mu}_0/\hat{\sigma}^2_0, -\frac{1}{2\hat{\sigma}^2_0})$
                    \While{not converged}
                        \State $Z_1, \dots, Z_N \sim \mathcal{N}(0, I)$
                        \State $\hat{\gamma}_{t+1} \gets \frac{1}{N} \sum_{i=1}^N t(Z_i)f(\hat{\mu}_t + \hat{\sigma}_t \otimes Z_i)$
                        \State Compute $\hat{\eta}'_{t+1}$ using~\eqref{eq:nicolas_scheme1_bis}
                        \State $\varepsilon_t \gets \texttt{stepsize}(\hat{\gamma}_{t+1}, \hat{\eta}'_{t+1}, \hat{\eta}_t, Z_{1:N})$
                        \State $\hat{\eta}_{t+1} \gets \varepsilon_t \hat{\eta}'_{t+1} + (1 - \varepsilon_t)\hat{\eta}_t$
                        \State $\hat{\mu}_{t+1} \gets -\frac{1}{2}\hat{\eta}^{(2),-1}_{t+1}\hat{\eta}_{t+1}^{(1)}$
                        \State $\hat{\sigma}^2_{t+1} \gets -\frac{1}{2}\hat{\eta}^{(2),-1}_{t+1}$
                        \vspace{0.1pt}
                    \EndWhile
                \end{algorithmic}
            \end{algorithm}
        \end{minipage}
        \hfill
        \begin{minipage}{0.48\textwidth}
            \begin{algorithm}[H]
                \caption{LSVI-FC (full-covariance Gaussian family)}\label{alg:lsvi_krasnoselskii_fullcov}
                \begin{algorithmic}
                    \Require $\mu_0, \Sigma_0 \succ 0$, $N \geq 1$
                    \State $(\hat{\mu}_0, \hat{\Sigma}_0) \gets (\mu_0, \Sigma_0)$
                    \State $\hat{\eta}_0 \gets (0, \hat{\Sigma}_0^{-1}\hat{\mu}_0, -\frac{1}{2}\vect{\hat{\Sigma}_0^{-1}})$
                    \While{not converged}
                        \State $\hat{C}_t \gets \text{Cholesky}(\hat{\Sigma}_t)$
                        \State $Z_1, \dots, Z_N \sim \mathcal{N}(0, I)$
                        \State $\hat{\gamma}_{t+1} \gets \frac{1}{N} \sum_{i=1}^N t(Z_i)f(\hat{\mu}_t + \hat{C}_t Z_i)$
                        \State Compute $\hat{\eta}'_{t+1}$ using~\eqref{eq:mappingfullcov}
                        \State $\varepsilon_t \gets \texttt{stepsize}(\hat{\gamma}_{t+1},\hat{\eta}'_{t+1}, \hat{\eta}_t, Z_{1:N})$
                        \State $\hat{\eta}_{t+1} \gets \varepsilon_t \hat{\eta}'_{t+1} + (1 - \varepsilon_t)\hat{\eta}_t$
                        \State $\hat{\mu}_{t+1} \gets -\frac{1}{2}\unvec(\hat{\eta}^{(2)}_{t+1})^{-1}\hat{\eta}_{t+1}^{(1)}$
                        \State $\hat{\Sigma}_{t+1} \gets  -\frac{1}{2}\unvec(\hat{\eta}^{(2)}_{t+1})^{-1}$
                    \EndWhile
                \end{algorithmic}
            \end{algorithm}
        \end{minipage}
    \end{figure}

    \section{Numerical experiments}\label{sec:experiments}
    We consider three examples: one where SGD may be used to minimise the
    KL objective, and two where it may not, because the
    reparameterisation trick is not possible: distributions $q$ in $\mathcal{Q}$ are discrete, $\log \pi$ is not differentiable,
    or because the log-derivative trick yields noisy estimates~\citep{Williams1992}.

    In the first example (logistic regression), we compare all three LSVI\footnote[1]{Python package: \url{https://github.com/ylefay/LSVI}} instances
    with other gradient-based KL minimisation procedures, including ADVI, NGD, and a
    gradient-free procedure for Gaussian mixtures.
    In the second and third examples (variable selection and Bayesian synthetic likelihood, BSL), since SGD is not available, we assess the
    approximation error of LSVI relative to the true posterior using exact Bayesian inference.

    \subsection{Logistic regression}\label{sec:Logistic regression}

    Given data $(x_i,y_i)\in \bbR^{d}\times \{-1, 1\}$, $i=1,\ldots, n$, the
    posterior distribution of a logistic regression model is:
    $    \pi(\beta) \propto p(\beta) \prod_{i=1}^{n} F(y_i
    x_i^{\top}\beta)$
    where $F(x) = 1 / (1 + e^{-x})$ and   $p(\beta)$ is a (typically Gaussian)
    prior over the parameter $\beta$. This type of posterior is often close to a
    Gaussian, and is a popular benchmark in Bayesian computation
    \citep{Chopin2017Leave}. See Appendix~\ref{app:logistic} for a summary of the considered
    datasets and the priors.

    Whenever applicable, we compare LSVI (Algorithms~\ref{alg:lsvi_krasnoselskii},~\ref{alg:lsvi_krasnoselskii_mfg},~\ref{alg:lsvi_krasnoselskii_fullcov})
    with NGD and ADVI. For NGD, the gradients are obtained via
    JAX autodifferentiation~\citep{jax2018github} and the FIM is estimated via Monte Carlo.
    For ADVI, we use the standard implementations given by pyMC3~\citep{Salvatier2016} and
    Blackjax~\citep{cabezas2024blackjaxcomposablebayesianinference} with default step size schedules (that is,
    a modification of Adam and RMSProp for pyMC, and comparable fixed step sizes for Blackjax).
    In addition, we provide a comparison of LSVI (Algorithm~\ref{alg:lsvi_krasnoselskii})
    in low dimension with the gradient-free iteration for Gaussian mixtures (GMMVI,~\citep[][]{pmlr-v80-arenz18a})
    which is a fair comparison since GMMVI and LSVI~\cref{alg:lsvi_krasnoselskii} have the same complexity in this case.
    In addition, we illustrate the compatibility of our proposed methods with
    subsampling procedures for large datasets~\citep{10.5555/2997189.2997285,JMLR:v14:hoffman13a}
    to reduce the cost of the log-likelihood evaluations.

    \Cref{fig:pima_short} summarises this comparison for the Pima dataset (full-covariance case).
    One sees that LSVI (\cref{alg:lsvi_krasnoselskii}) converges
    essentially in one step, LSVI-FC (\cref{alg:lsvi_krasnoselskii_fullcov})
    converges in less than 100 steps for linearly decreasing step sizes.
    For such a low-dimensional dataset ($d=9$), LSVI remains competitive with
    LSVI-FC since it converges faster, and the matrices it needs to invert are
    small. LSVI performs comparably to NGD and GMMVI, but is less noisy (with or without an adaptive schedule, Algorithm~\ref{alg:recipe_backracking_variance_control}).
    We consider larger and more challenging datasets as recommended by~\citep{Chopin2017Leave}.
    \Cref{fig:mnist_variable_selection} (left) does the same comparison for the MNIST dataset (mean-field covariance).
    In Appendix~\ref{app:extra_numerics},~\Cref{fig:sonar} does the same for the Sonar dataset (full-covariance) and~\Cref{fig:census} for the Census-Income dataset (mean-field covariance with subsampling).
    This time, inverting the FIM is too costly (e.g., $2015\times 2015$ for Sonar), so we only use the tailored schemes
    LSVI-MF and LSVI-FC.
    ~\cref{app:logistic} contains extra details and results for all datasets in~\cref{tab:logistic_data}, including runtimes and memory usage (\cref{table:perf}),
    average cost time per iteration with respect to $N$ (\cref{fig:mnist_runtime,fig:sonar}), loss vs iterations and
    classification performance (\cref{fig:mnist_miss}), details on the considered schedules for the step sizes (\cref{tab:schedules}).
    \begin{figure}[ht]
        \begin{subfigure}[ht]{0.48\columnwidth} %
            \begin{center}
                \centerline{
                    \includegraphics[height=4.75cm]{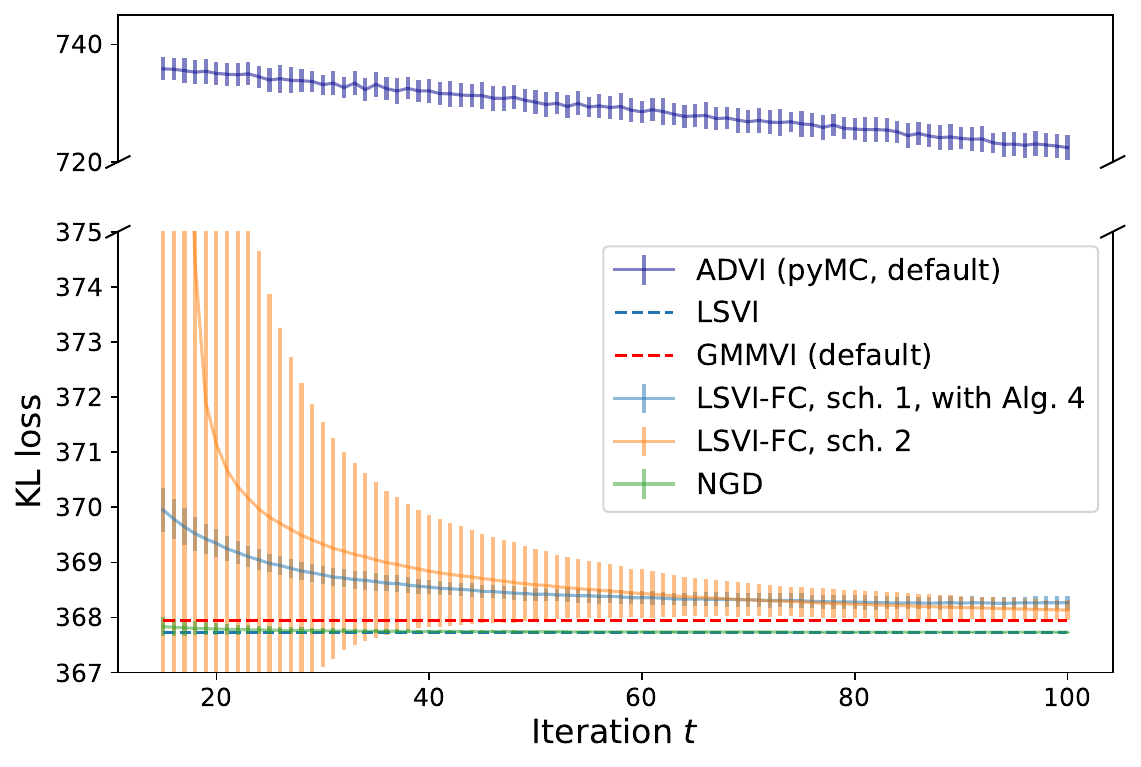}
                }
            \end{center}
        \end{subfigure}
        \hfill %
        \begin{subfigure}[ht]{0.48\columnwidth} %
            \begin{center}
                \centerline{
                    \includegraphics[height=4.75cm]{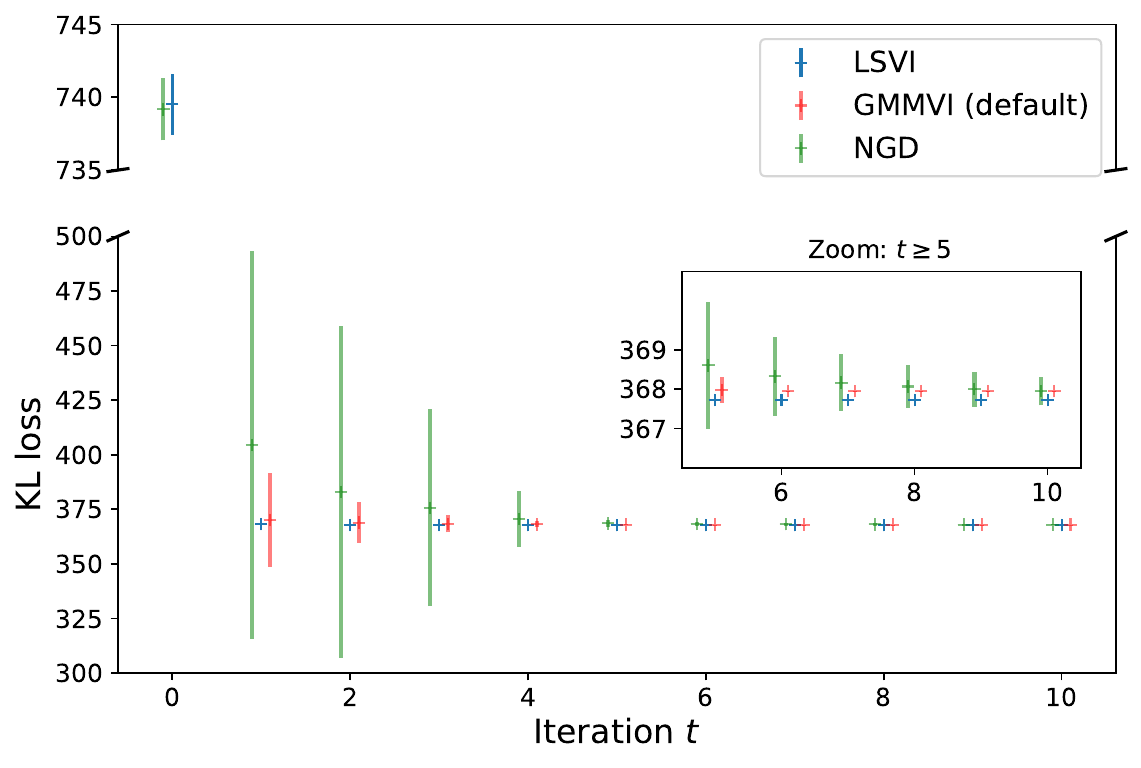}
                }
            \end{center}
        \end{subfigure}

        \caption{
            Logistic regression, Pima data, full-covariance
            approximation. KL divergence (up to an unknown constant) between the variational approximation and the
            posterior, as a function of the number of iterations. Left: truncated from iteration $t\geq 20$ for better readability. Right: focus on GMMVI, LSVI and NGD. Mean over $100$ repetitions and one standard deviation interval (\texttt{jax.numpy.std}).}
        \label{fig:pima_short}
    \end{figure}

    \begin{figure}[ht]
        \begin{subfigure}[ht]{0.48\columnwidth} %
            \begin{center}
                \centerline{
                    \includegraphics[height=4.75cm]{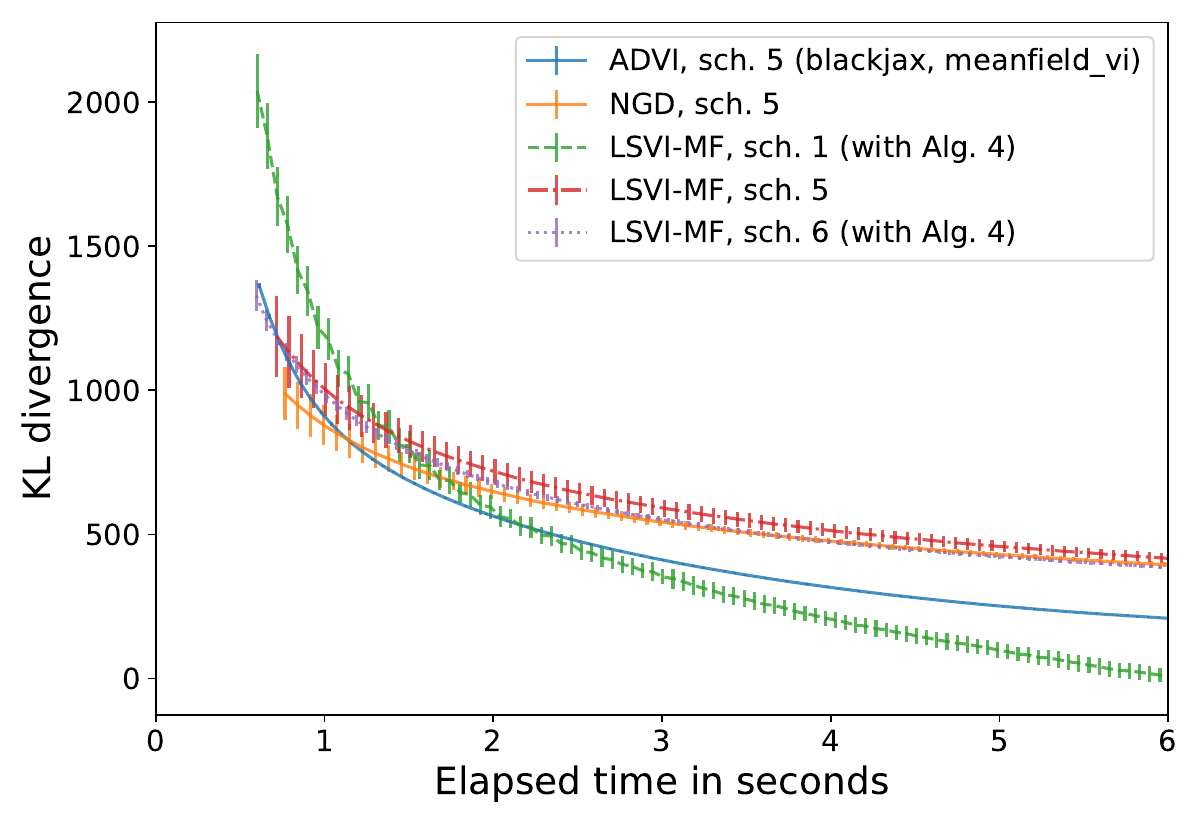}
                }
            \end{center}
        \end{subfigure}
        \hfill %
        \begin{subfigure}[ht]{0.48\columnwidth} %
            \begin{center}
                \centerline{
                    \includegraphics[height=4.8cm]{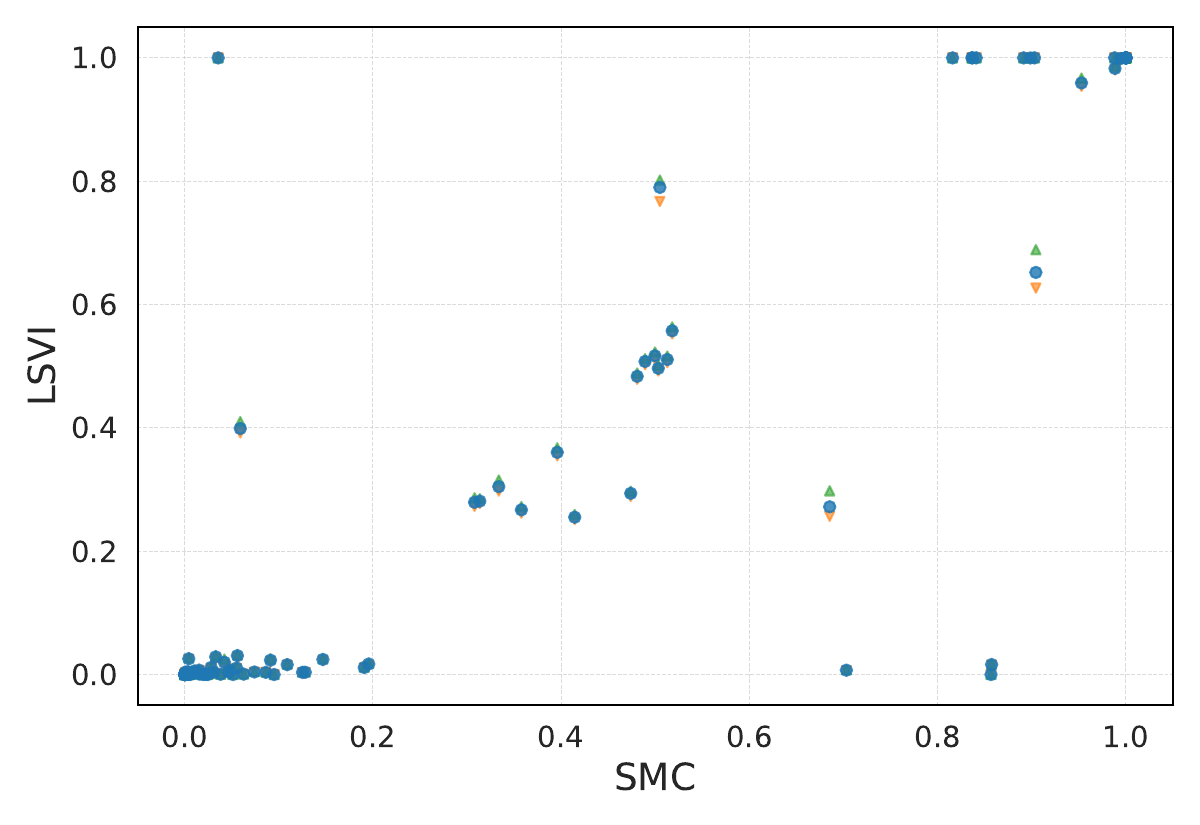}
                }
            \end{center}
        \end{subfigure}

        \caption{Left: Logistic regression, MNIST, mean-field
        approximation. KL divergence between the variational approximation and the
        posterior, as a function of the elapsed time. Truncated from iteration $t\geq 10$.
        Mean over $100$ repetitions and one standard deviation interval. Right: Variable selection example, posterior marginal probabilities
            $\pi(\gamma_i=1|\mathcal{D})$: LSVI vs SMC. (LSVI: 100 repetitions, the min-max intervals are reported with arrows, SMC: 3 repetitions).}
        \label{fig:mnist_variable_selection}
    \end{figure}

    \subsection{Variable selection}
    \label{subsec:selection}

    Given a dataset $\mathcal{D}=(x_i, y_i)_{i=1,\dots,n}$,
    $x_i\in\bbR^d$, $y_i\in \bbR$,
    the variable selection task in Bayesian linear regression may be modelled as
    $y_i = x_i^\top \diag(\gamma) \beta + \sigma \varepsilon_i, \varepsilon_i\sim N(0, 1)$,
    where $\gamma \in \{0, 1\}^d$ is a vector of inclusion variables,
    which is assigned a prior distribution that is a product of Bernoulli$(p)$; e.g.,
    $p=1/2$.
    If  $(\beta, \sigma^2)$ is assigned a conjugate prior,  the marginal posterior
    distribution $\pi(\gamma|\mathcal{D})$ (with $\beta$, $\sigma^2$ integrated
    out) admits a closed-form expression, the support of which is $\{0, 1\}^d$. It
    is therefore natural to set $\mathcal{Q}$ to the family of Bernoulli products,
    i.e.  $q(\gamma) = \prod_{i=1}^{d} q_i^{\gamma_i} (1 - q_i)^{1 - \gamma_i}$ with
    $q_i\in [0, 1]$ for $i=1,\dots, d$. This family is discrete, which precludes a
    reparametrisation trick, and the application of ADVI.

    \Cref{fig:mnist_variable_selection} (right) compares the posterior inclusion
    probabilities, i.e. $\pi(\gamma_i=1|\mathcal{D})$ approximated either through LSVI
    (\cref{alg:lsvi_krasnoselskii}), or the Sequential Monte Carlo (SMC) sampler of
    \citet{Schäfer2013}, for the concrete dataset ($d=92$).
    This dataset is challenging as it generates strong
    posterior correlations between the $\gamma_i$. Despite this,
    LSVI gives a reasonable approximation of the true posterior.
    To the best of our knowledge, this is the first time variational inference is
    implemented for variable selection using the Bernoulli product family.
    See \cref{app:extra_vs} for extra numerical results and
    more details on the prior, the data, and the implementation.

    \subsection{Bayesian synthetic likelihood}\label{subsec:bsl}

    BSL is a popular way to perform likelihood-free
    inference, that is, inference on a parametric model which is described only
    through a simulator: one is able to sample $Y\sim P_\theta$, but not to compute
    the likelihood $p(y|\theta)$; see~\citet{bsl_review} for a review.

    BSL requires to specify $s(y)$, a low-dimensional summary of the data and
    assumes that $s(y)\sim N\left(b(\theta), \Sigma(\theta)\right)$, leading to
    posterior density  $\pi(\theta) \propto p(\theta) N\left(s(y); b(\theta),
                                                           \Sigma(\theta)\right)$, where $p(\theta)$ is the prior. Since functions $b$
    and $\Sigma$ are unknown, they are replaced by empirical moments
    $\hat{b}(\theta)$, $\hat{\Sigma}(\theta)$, computed from
    simulated data. This makes BSL, and in particular its Markov Chain Monte Carlo (MCMC) implementations,
    particularly CPU-intensive, as the data simulator must be run many times.
    Furthermore, each evaluation of $\pi$ is corrupted with
    noise, making it impossible to differentiate $\log \pi$. Note that, in general,
    the data simulator is too complex to implement some form of reparametrisation
    trick, or the application of automatic differentiation procedures.

    We consider the toad's displacement example from~\citep{MARCHAND201763}, which
    has been considered in various BSL papers~\citep{bsl_review,AnNottDrovandi}.
    The model is parameterised by $\theta=(\alpha, \gamma, p_0)\in \bbR^+\times \bbR^+ \times [0, 1]$.
    See \cref{app:extra_bsl} for more details on the model.
    We implement both LSVI-MF and LSVI-FC (both implemented as Algorithm~\ref{alg:lsvi_krasnoselskii}).
    For the former, we use a family of truncated Gaussian distributions, while for
    the latter, we re-parametrise the model in terms of $\xi=f(\theta)$, where $f$
    is one-to-one transform between $\Theta$ and $\bbR^d$.
    The top panel of \Cref{fig:bsl}  shows that both LSVI algorithms converge quickly.
    The bottom panel shows that the full-covariance LSVI approximation matches
    the posterior obtained via MCMC, at a fraction of the CPU cost, see~\cref{table:perf}.
    Again, we refer to \cref{app:extra_bsl} for more details on the implementation of
    either LSVI or MCMC.
    \begin{figure}[ht]
        \begin{minipage}{0.48\columnwidth} %
            \centering
            \includegraphics[height=4.5cm]{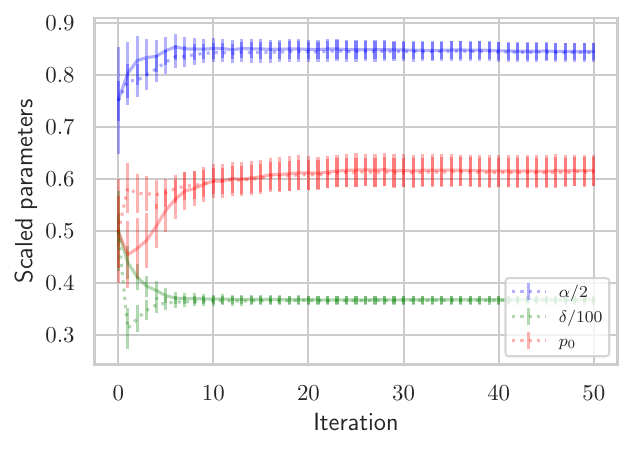}
        \end{minipage}
        \hfill %
        \begin{minipage}{0.48\columnwidth} %
            \centering
            \includegraphics[height=4.5cm]{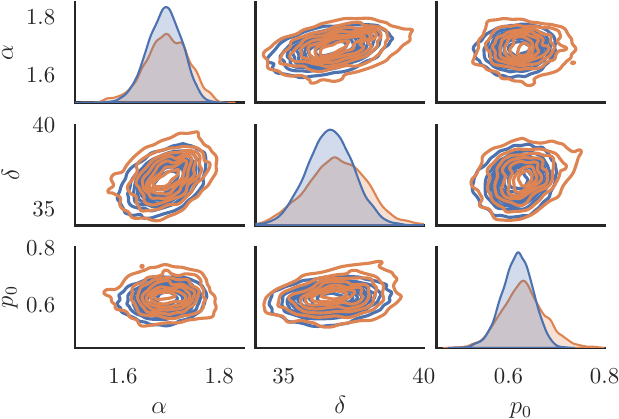}
        \end{minipage}

        \caption{
            Left: Variational approximations of each coordinate of $\theta$ with one standard
            deviation interval, normalised. Truncated Gaussian: solid line. Full
            covariance Gaussian: dashed line.
            Right: Full-covariance Gaussian variational approximation (blue),
            MCMC approximation (orange).
        }
        \label{fig:bsl}
    \end{figure}

    \section{Limitations}
    \label{sec:limitations}
    The current approach is limited to exponential families; mixture of exponential families may be tackled by adapting the expectation-maximisation approach of~\citet{pmlr-v80-arenz18a}, or by building on existing applications of NGD VI methods to mixtures of exponential families~\citep{pmlr-v97-lin19b,10.1214/13-BA858}.
    For Gaussian approximations, if the posterior contains directions that are strongly non-Gaussian, then conditional-Gaussian strategies like integrated nested Laplace approximations may be applied~\citep{rue_laplace09}.
    In discrete exponential families, independence can be lifted by considering tree-structured dependencies, which are quite flexible, see, e.g.,~\citet[][]{wainwright2008graphical}.

    \begin{ack}
        The first author gratefully acknowledges partial support from the Magnus Ehrnrooth foundation.
        The authors thank Sam Power, Mohammad Emtiyaz Khan and anonymous reviewers for insightful remarks on a preliminary version.
    \end{ack}

    \bibliography{main}
    \bibliographystyle{unsrtnat}

    \newpage
    \appendix
    \addtocontents{toc}{\protect\setcounter{tocdepth}{2}}
    \tableofcontents

    \section{Notations}
    \label{appendix:notations}
    For any vector $u\in \bbR^p$, we denote by $u^{-1}\in\bbR^p$ the componentwise inverse of $u$, and $u\otimes u'$ the componentwise product of $u$ with $u'$.
    For any matrix $U\in \bbR^{p\times q}$, we denote by $\vect(U)$ the $p\times q$ vector obtained by vertically stacking the columns of $U$, and by $\unvec$ the inverse operation satisfying $\unvec(\vect(U)) = U$,
    while $U\otimes U'$ denotes the Kronecker product between matrices.
    For any square matrix $U\in \bbR^{p\times p}$, let $\diag(U)$ be the $p$ vector composed of the diagonal components of $U$ and let $\norm{U}$ be the spectral norm of $U$.

    For any set $A$, $\mathcal{U}(A)$ denotes the uniform distribution over $A$.
    $N(\mu,\Sigma)$ denotes the Gaussian distribution with mean $\mu$ and covariance matrix $\Sigma$, and $N(\mu,\sigma^2)$ with $\sigma^2 = (\sigma_1^2,\ldots, \sigma_d^2)^{\top}$ denotes the Gaussian distribution with mean $\mu$ and diagonal covariance matrix $\diag(\sigma^2)$.

    The $O$ is the usual big-$O$ notation, i.e., $A_n = O(B_n)$ for some sequences $A_n$, $B_n$, let it be reals, vectors or matrices, if there exists a constant $C>0$ such that for $N$ large enough and all $n\geq N$, $\norm{A_n}\leq C\norm{B_n}$.
    If $\norm{A_n}\geq C\norm{B_n}$, we write $A_n=\Omega(B_n)$.
    We write $A_n = \OO_P(1)$ for a sequence of random variables $(A_n)$ such
    that, for any $\varepsilon > 0$,  there exists a constant $B>0$ such that
    $P(\norm{A_n} > B) \leq \varepsilon$ for $n$ large enough.

    For any definite positive matrix $\Sigma$, we denote by $C = \textup{Chol}(\Sigma)$ the unique lower triangular matrix such that $C C^{\top} = \Sigma$.

    \section{Adaptive schedule algorithm}
    \label{appendix:algorithm_eps}
    \begin{algorithm}[H]
        \caption{Variance control and backtracking strategy}
        \label{alg:recipe_backracking_variance_control}
        \begin{algorithmic}[1]
            \Require $\varepsilon'>0$, $\eta\in \mathcal{V}$, $\eta'\in \bbR^m$, $N\geq 1$, $X_1, \dots, X_N \stackrel{\mathrm{i.i.d.}}{\sim} q_{\eta}$, $u>0$
            \State $\varepsilon \gets \varepsilon'$
            \While{$\varepsilon \eta'+(1-\varepsilon)\eta \notin \mathcal{V}$}
                \State $\varepsilon \gets \varepsilon/2$
            \EndWhile
            \State $\eta \gets \varepsilon\eta'+(1-\varepsilon)\eta$
            \State $\hat{m} \gets N^{-1}\sum_{i=1}^{N} f(X_i)-\eta^{\top}s(X_i)$
            \State $\hat{v}^2 \gets N^{-1}\sum_{i=1}^{N}(f(X_i) - \eta^{\top}s(X_i)-\hat{m})^2$
            \If{$\hat{v}\geq u$}
                \State $\varepsilon \gets \textup{min}(\varepsilon, u/\hat{v})$
            \EndIf \\
            \Return $\varepsilon$
        \end{algorithmic}
    \end{algorithm}

    \section{Extra details on numerical experiments}\label{app:extra_numerics}

    \subsection{Runtime analysis}\label{app:runtime}
    All the experiments were conducted using Python 3.13, jax 0.5 with GPU support, Cuda 12.5, and using float64.
    The hardware specifications are CPU AMD EPYC 7702 64-Core Processor and GPU NVIDIA A100-PCIE-40GB, except for SONAR, Census and MNIST datasets where EPYC 7713 and NVIDIA A100-PCIE-80GB were used.
    See~\cref{table:perf}.
    \begin{table}[ht]
        \caption{For all conducted experiments, runtimes and max memory usage, across $5$ repetitions. $T$ is the number of iterations and $N$ the number of samples whenever applicable.
        LR = Logistic regression, BSL = Bayesian Synthetic likelihood, MF Gaussian
            = mean-field Gaussian, Gaussian = full-covariance Gaussian.}
        \resizebox{\columnwidth}{!}{%

            \begin{tabular}{l*{3}{c}r}
                Experiment & \multicolumn{3}{c}{Runtime (seconds)} & max resident set size (memory usage) \\
                & mean (std)          & min     & max     & \,\, (gigabytes) \\
                \hline
                BSL Gaussian, Alg.~\ref{alg:lsvi_krasnoselskii}, $(N, T)=(100, 50)$ (JAX)                                                & $72.9$ ($\pm 2.8$)  & $71.5$  & $77.8$  & $1.07$\\
                BSL Truncated MF Gaussian, Alg.~\ref{alg:lsvi_krasnoselskii}, $(100, 50)$ (JAX)                                          & $137.5$ ($\pm 0.6$) & $137.3$ & $138.7$ & $1.05$\\
                BSL MCMC, Blackjax (JAX)                                                                                                 & $268.1$ ($\pm 3.4$) & $266.5$ & $274.3$ & $1.16$           \\
                \hline
                Variable Selection, Alg.~\ref{alg:lsvi_krasnoselskii} sch. $3$, $(5\times 10^4, 25)$                          & $60.8$ ($\pm 0.3$)  & $60.3$  & $61.1$ & $0.42$\\
                Variable Selection, SMC                                                                                                  & $290.7$ ($\pm 1.7$) & $284.1$ & $298$   & $0.45$           \\
                \hline
                LR Gaussian, \emph{PIMA}, Alg.~\ref{alg:lsvi_krasnoselskii} sch. $3$, $(10^4, 10)$ (JAX)                     & $1.6$ ($\pm 1.4$) & $1.0$ & $4.1$ & $1.28$\\
                // NGD, ($10^4, 10$) (JAX)                                                                                               & $2.2$ ($\pm 1.5$)   & $1.51$  & $4.9$   & $3.29$           \\
                // Alg.~\ref{alg:lsvi_krasnoselskii_fullcov} sch. $1$, $(10^5, 100)$, (JAX)                                              & $4.3$ ($\pm 1.7$) & $3.5$                  & $7.3$       & $1.28$                     \\
                // Alg.~\ref{alg:lsvi_krasnoselskii_fullcov} sch. $2$, $(10^5, 100)$, (JAX)                                              & $3.4$ ($\pm 0.1$) & $3.3$ & $3.6$ & $1.28$\\
                // PyMC ADVI, ($T=10^4$)                                                                                                 & $5.9$ ($\pm 0.3$)   & $5.5$   & $6.4$   & $0.84$           \\
                // GMMVI, ($10^4, 10$) (TensorFlow)                                                                                      & $4.8$ ($\pm 3.8$)   & $3.0$   & $11.6$  & $4.57$           \\
                \hline
                LR Gaussian, \emph{SONAR}, Alg.~\ref{alg:lsvi_krasnoselskii_fullcov} sch. $1$, $(10^5, 100)$,  (JAX) & $5.1$ ($\pm 0.3$) & $5.0$ & $5.6$ & $3.11$ \\
                // Alg.~\ref{alg:lsvi_krasnoselskii_fullcov} sch. $2$, $(10^5, 100)$, (JAX)                                              & $5.1$ ($\pm 1.3$) & $4.4$ & $7.4$ & $3.11$\\
                // PyMC ADVI, ($10^4$)                                                                                                   & $9.6$ ($\pm 1.3$)   & $7.4$   & $10.5$  & $3.21$           \\
                \hline
                LR MF Gaussian, \emph{MNIST}, Alg.~\ref{alg:lsvi_krasnoselskii_mfg} sch. $1$, $(10^4, 500)$, (JAX) & $19.1$ ($\pm 1.0$) & $18.7$& $21.8$ & $2.36$\\
                // Alg.~\ref{alg:lsvi_krasnoselskii_mfg} sch. $5$, $(10^4, 500)$, (JAX)                                                  & $10.4$ ($\pm 0.1$) & $10.3$ & $10.6$ & $2.36$\\
                // Alg.~\ref{alg:lsvi_krasnoselskii_mfg} sch. $6$, $(10^4, 500)$, (JAX)                                                  & $18.9$ ($\pm 1.1$) & $22.1$ & $18.5$ & $2.36$\\
                // Blackjax ADVI, $(10^4, 500)$, (JAX)                                                                                   & $19.5$ ($\pm 1.0$)  & $18.9$  & $22.2$  & $3.34$           \\
                // NGD (MF), sch. $5$, $(10^4, 500)$ (JAX)                                                                               & $25.3$ ($\pm 2.2$)  & $29.2$  & $24.1$  & $4.38$           \\
                \hline
                LR MF Gaussian, subsampling, \emph{Census}, Alg~\ref{alg:lsvi_krasnoselskii_mfg}, sch. $5$, $(10^4, 10^4, P=10^3)$ (JAX) & $12.9$ ($\pm 0.2$)& $12.8$ & $13.3$ & $3.30$ \\
                // Alg~\ref{alg:lsvi_krasnoselskii_mfg}, sch. $6$, $(10^4, 10^4, 10^3)$ (JAX)                                            & $13.0$ ($\pm 0.03$)                    & $12.9$                & $13.0$& $3.30$\\
                // NGD (MF), sch. $5$, $(10^4, 10^4, 10^3)$ (JAX)                                                                        & $70.0$ ($\pm 1.7$)  & $69.0$  & $73.1$                & $3.33$     \\
            \end{tabular}
        }

        \label{table:perf}
    \end{table}

    \subsection{Logistic regression}\label{app:logistic}

    \paragraph{Data}
    The Sonar (CC BY 4.0 License) and the Census Income (CC BY 4.0 License) datasets are available in the UCI repository while the Pima dataset (CC0: Public Domain License) is in the example datasets of Python package particles (License MIT v0.4,~\citep[][Ch. 1]{Chopin2020SMC})
    and MNIST (CC BY-SA 3.0 License) is available at~\url{https://github.com/pjreddie/mnist-csv-png}.
    We use the following standard \citep[e.g.,][]{Chopin2017Leave}  pre-processing strategy for Pima,
    Sonar and Census-Income datasets: we add an intercept, and we rescale the covariates so
    that non-binary predictors are centred with standard deviation $0.5$, and the
    binary predictors are centred $0$ and range $1$. For the third dataset (MNIST
    dataset), we restrict ourselves to the binary classification
    problem by selecting pictures labelled $0$ or $8$. The gray-scale features
    which range between $0$ and $255$ are normalised to be between $0$ and $1$. No
    intercept is added. For the Census Income dataset, the categorical variables are
    mapped using one-hot encoding.

    \begin{table}[ht]

        \caption{Logistic regression example: summary of datasets and approximation
        families, in parentheses the batch-size}
        \centering
        \begin{tabular}{|c|c|c|c|}
            \hline
            Dataset              & Gaussian family & $d$ & $n$           \\ \hline
            Pima                 & full-covariance & 9   & 768           \\ \hline
            Sonar                & full-covariance & 62  & 128           \\ \hline
            Census (subsampling) & mean-field      & 48  & 49 000 (1000) \\ \hline
            MNIST                & mean-field      & 784 & 11,774        \\ \hline
        \end{tabular}
        \label{tab:logistic_data}
    \end{table}

    \paragraph{Prior}
    For all datasets except MNIST, the prior $\pi(\beta)$ is a zero-mean Gaussian
    distribution with diagonal covariance matrix, and the covariances are set to
    $25$ for all the other covariates, except for the intercept, for which it is
    set to $400$. For the MNIST dataset, the prior is a Gaussian
    distribution with zero-mean and covariance matrix $25I_d$.

    \paragraph{Initialisations, schedules and number of samples}
    The initialisation distributions for all datasets except MNIST are standard normal distributions.
    The initialisation for the MNIST dataset is $N(0, e^{-2}I_d)$.
    The learning schedules $(\varepsilon_t)$ are obtained via Algorithm~\ref{alg:recipe_backracking_variance_control} with specific inputs $(u^2,\varepsilon_t)$ summarised in~\cref{tab:schedules} along with the number of samples $N$.

    \begin{table}[ht]
        \caption{Logistic regression setup. Left: Inputs to Algorithm~\ref{alg:recipe_backracking_variance_control} by dataset. Right: Schedule index reference.}
        \centering
        \begin{minipage}{0.7\textwidth}
            \centering
            \resizebox{\textwidth}{!}{%
                \begin{tabular}{|c|c|c|c|}
                    \hline
                    Dataset & Algorithm                                 & Schedule input $(u^2,\varepsilon_t)$          & Samples $N$ \\ \hline
                    Pima    & Alg.~\ref{alg:lsvi_krasnoselskii}         & $(\infty,1)$                                  & $10^4$      \\ \hline
                    Pima    & NGD                                       & $(\infty, 1/(t+1))$                           & $10^4$      \\ \hline
                    Pima    & Alg.~\ref{alg:lsvi_krasnoselskii_fullcov} & $(10,1)$, $(\infty, 1/(t+1))$                 & $10^5$      \\ \hline
                    Sonar   & Alg.~\ref{alg:lsvi_krasnoselskii_fullcov} & $(10,1)$, $(\infty, 1/(t+1))$                 & $10^5$      \\ \hline
                    MNIST   & Alg.~\ref{alg:lsvi_krasnoselskii_mfg}     & $(10,1)$, $(\infty, 10^{-3})$, $(10,10^{-3})$ & $10^4$ \\ \hline
                    MNIST   & Blackjax (meanfield\_vi), NGD (MF)        & $(\infty, 10^{-3})$                           & $10^4$      \\ \hline
                    Census  & Alg.~\ref{alg:lsvi_krasnoselskii_mfg}     & $(10,10^{-3})$, $(\infty, 10^{-3})$ & $10^4$ \\ \hline
                    Census  & NGD (MF)                                  & $(\infty, 10^{-3})$                           & $10^4$      \\ \hline
                \end{tabular}
            }
        \end{minipage}%
        \hfill
        \begin{minipage}{0.28\textwidth}
            \centering
            \resizebox{\textwidth}{!}{%
                \begin{tabular}{|c|c|}
                    \hline
                    \# & Schedule input $(u^2,\varepsilon_t)$ \\ \hline
                    1  & $(10, 1)$                            \\ \hline
                    2  & $(\infty, 1/(t+1))$                  \\ \hline
                    3  & $(\infty, 1)$                        \\ \hline
                    4  & $(1, 1)$                             \\ \hline
                    5  & $(\infty, 10^{-3})$                  \\ \hline
                    6  & $(10, 10^{-3})$                      \\ \hline
                \end{tabular}
            }
        \end{minipage}
        \label{tab:schedules}
    \end{table}

    \begin{figure}[ht]
        \begin{subfigure}[ht]{0.48\columnwidth}
            \begin{center}
                \centerline{
                    \includegraphics[height=4.75cm]{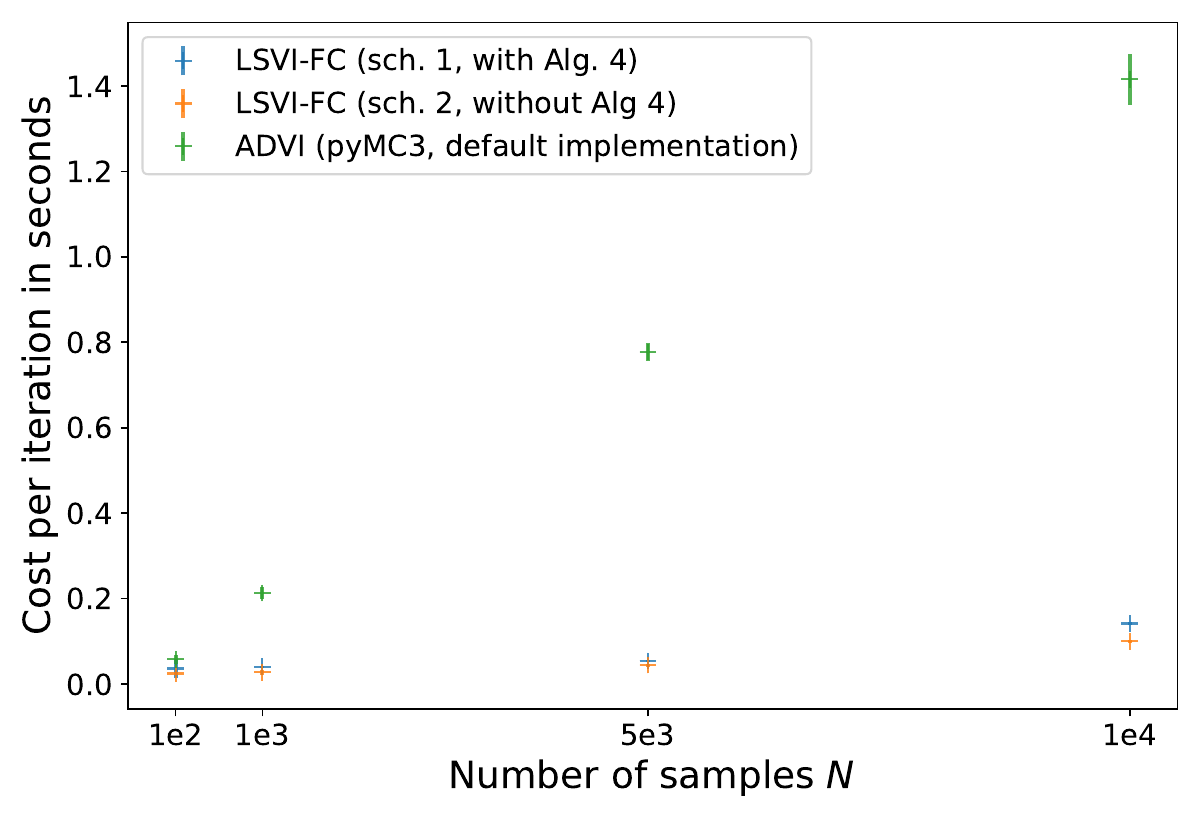}
                }
            \end{center}
        \end{subfigure}
        \begin{subfigure}[ht]{0.48\columnwidth}
            \begin{center}
                \centerline{
                    \includegraphics[height=4.75cm]{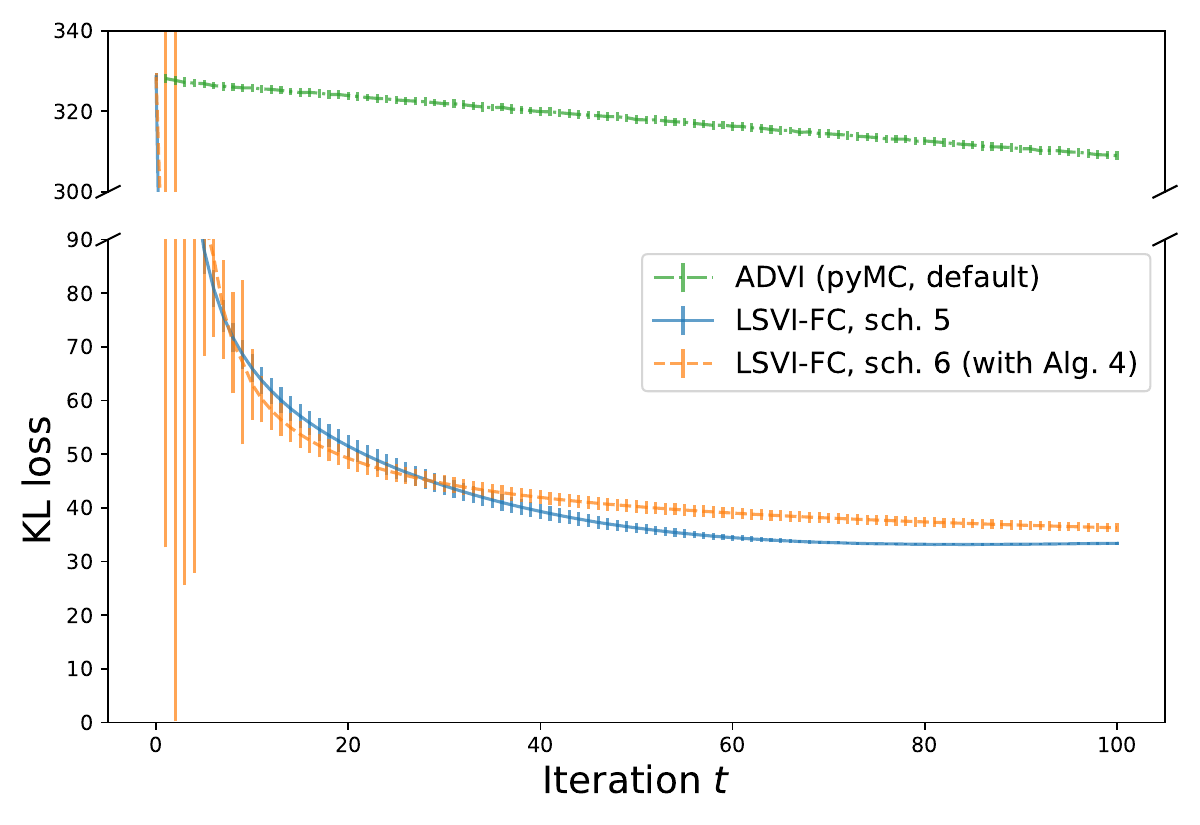}
                }
            \end{center}
        \end{subfigure}

        \caption{Logistic regression posterior, Sonar data, full-covariance approximation, LSVI-FC and ADVI implementations. Left: average cost per iteration in seconds as a function of the number of samples $N$, mean over $5$ repetitions with $2$ std interval.
        Right: KL divergence (up to an unknown constant) between current Gaussian variational approximation and the
        posterior, as a function of $t$, mean over $100$ repetitions with one standard deviation interval.}
        \label{fig:sonar}
    \end{figure}

    \paragraph{MNIST}
    The PyMC3 (License Apache 2.0 v. 5.22,~\citep[][]{Salvatier2016}) implementation fails in this context, and we resort to the stochastic gradient descent (SGD) implementation in Blackjax (License Apache 2.0 v1.2.5,~\citep[][]{cabezas2024blackjaxcomposablebayesianinference}) of the mean-field ADVI Algorithm.
    For SGD, we set the learning rate to $0.001$ and the number of samples for the Monte Carlo gradient estimates to $10^4$.
    See Figure~\ref{fig:mnist_runtime} for the average cost per iteration in seconds, and the same plot as Figure~\ref{fig:mnist_variable_selection} with respect to the iterations.

    \begin{figure}[ht]
        \begin{subfigure}[ht]{0.48\columnwidth}
            \centerline{
                \includegraphics[height=4.75cm]{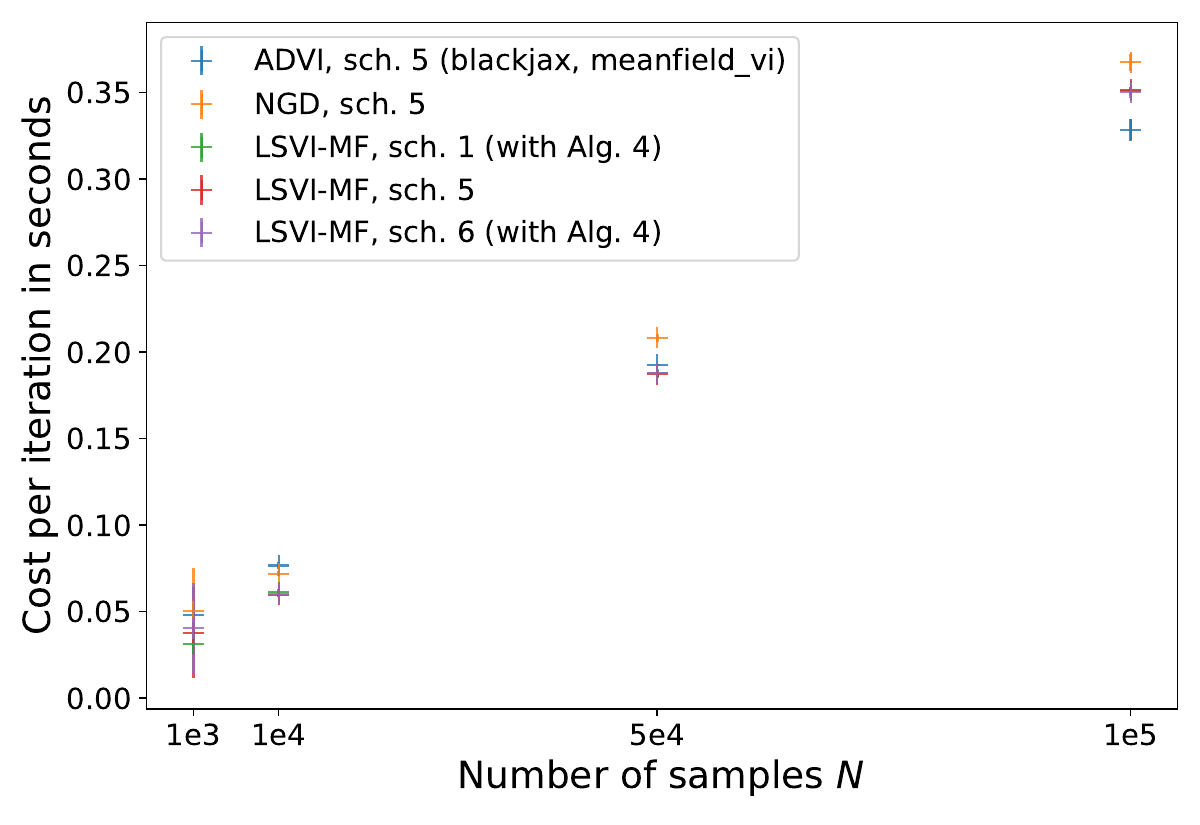}
            }
        \end{subfigure}
        \begin{subfigure}[ht]{0.48\columnwidth}
            \centerline{
                \includegraphics[height=4.75cm]{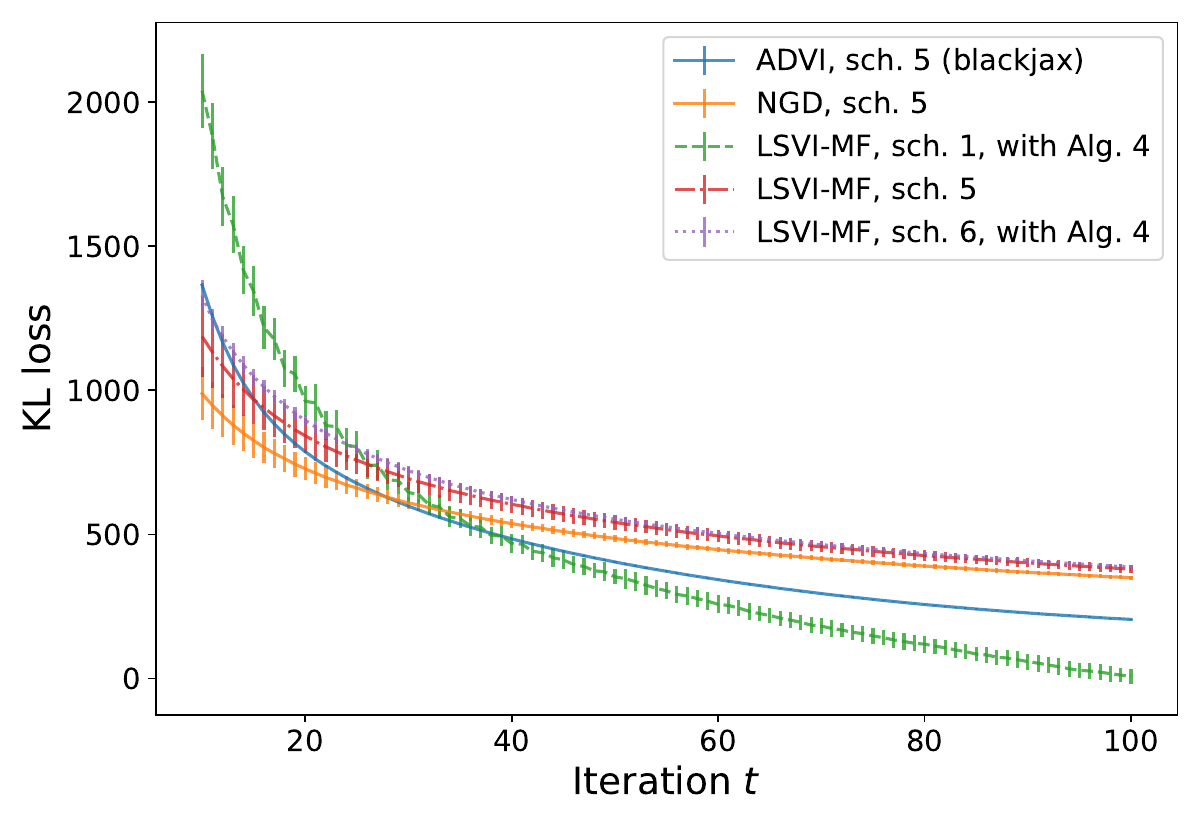}
            }
        \end{subfigure}

        \caption{
            Logistic regression posterior, MNIST data, diagonal covariance
            approximation, LSVI-MF, NGD (JAX) and Blackjax (meanfield\_vi) implementations. Left: average cost per iteration in seconds as a function of the number of samples $N$, mean over $5$ repetitions with $2$ std interval.
            Right: KL divergence (up to an unknown constant) between current Gaussian variational approximation and the
            posterior, as a function of $t$, mean over $100$ repetitions with one standard deviation interval.
        }
        \label{fig:mnist_runtime}
    \end{figure}

    In addition, we provide misclassification rate for the logistic regression model using the mean (of the Gaussian approximation) as the regression parameter, see~\cref{fig:mnist_miss}.

    \begin{figure}[ht]
        \begin{subfigure}[ht]{0.48\columnwidth}
            \centerline{
                \includegraphics[height=5cm]{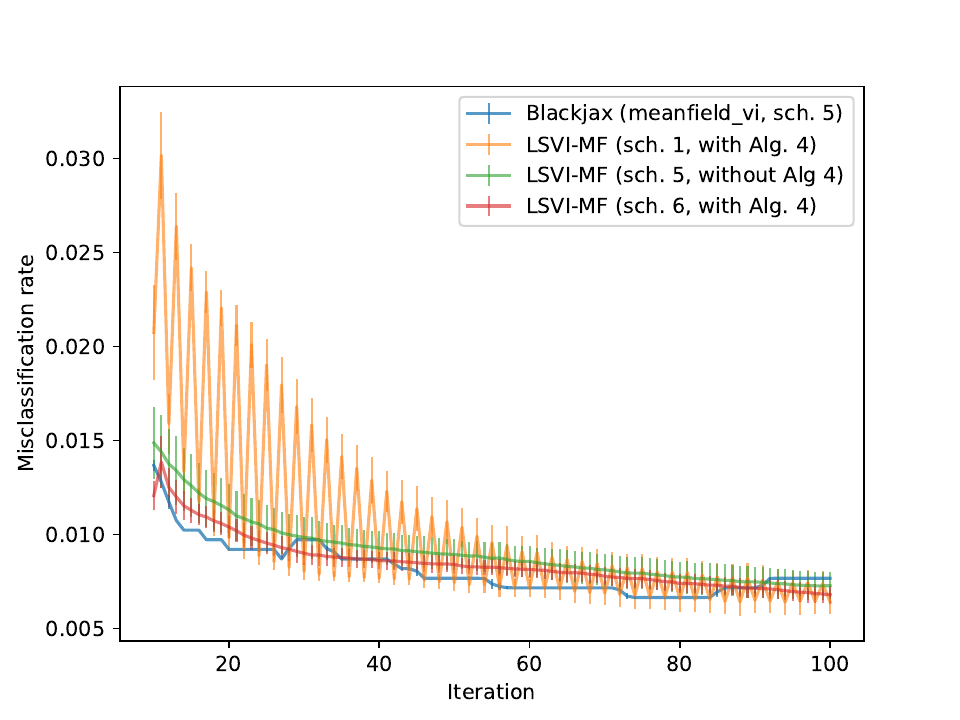}
            }
        \end{subfigure}
        \begin{subfigure}[ht]{0.48\columnwidth}
            \centerline{
                \includegraphics[height=5cm]{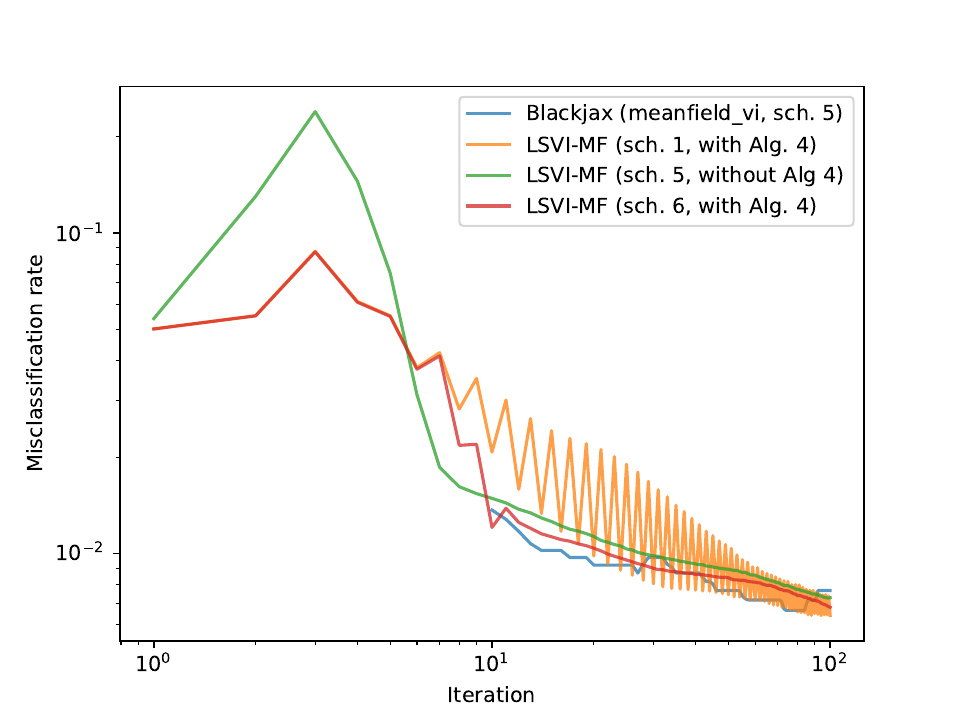}
            }
        \end{subfigure}

        \caption{
            Logistic regression posterior, MNIST data, diagonal covariance
            approximation, LSVI-MF and Blackjax (meanfield\_vi) implementations. Left: Misclassification rate as a function of the iterations, mean over $100$ repetitions with $1$ standard deviation.
            Right: same in log-log axis.
        }
        \label{fig:mnist_miss}
    \end{figure}

    \paragraph{Subsampling (Census dataset)}
    At each iteration $t$, a new batch is sampled uniformly with replacement from the dataset:
    \begin{equation}
        \hat{f}(\beta) = \log \hat{\pi}(\beta) = \log p(\beta) + \frac{n}{P}\sum_{i=1}^{P} \log p(y_{U_i} | x_{U_i}, \beta) + \frac{n}{P} \sum_{i=1}^{P} \log p(x_{U_i}),
    \end{equation}
    where $U_1,\ldots, U_P \sim \mathcal{U}(\{1, \ldots, n\})$.
    A new batch is drawn at each iteration. The batch size is $P = 10^3$.
    We also use $\hat{f}$ for evaluating the KL loss. See~\cref{fig:census}.

    \begin{figure}[ht]
        \centerline{
            \includegraphics[height=7cm]{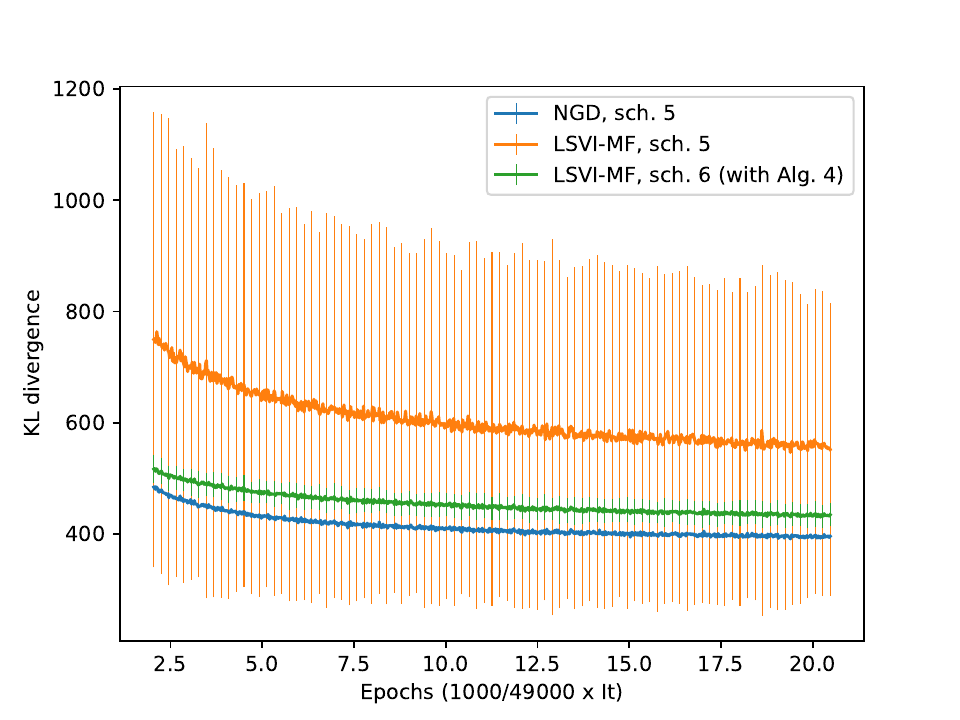}
        }

        \caption{
            Logistic regression posterior. KL loss for the Census-Income dataset (mean-field, with subsampling), mean over $100$ repetitions with $1$ standard deviation.
        }
        \label{fig:census}
    \end{figure}

    \subsection{Variable selection}\label{app:extra_vs}

    \paragraph{Dataset}
    The \emph{Concrete Compressive Strength} dataset~\citep{YEH19981797} is made of
    1030 observations and 8 initial predictors denoted by \textup{C, W, CA, FA,
        BLAST, FASH, PLAST}, and \textup{A}. We enrich the dataset by adding predictors
    computed from the existing predictors. $5$ new predictors, \textup{LG\_C,
        LG\_W, LG\_CA, LG\_FA, LG\_A}, where \textup{LG\_X} stands for the logarithm of
    the corresponding feature \textup{X}. The cross-product of the predictors is
    also added, resulting in $78$ new predictors. Finally, we add an intercept. The
    total of possible predictors is $d=92$.

    \paragraph{Prior}

    The hierarchical prior on $\beta, \sigma^2,\gamma$ is given by
    \begin{equation*}
        \pi(\beta\mid \sigma, \gamma, Z) = N\left(0, \sigma^2
        v^2\diag(\gamma)\right),\quad\pi(\sigma^2) = \textup{InvGamma}(w/2, \lambda
        w/2), \quad \pi(\gamma) = \mathcal{U}(\{0, 1\}^d).
    \end{equation*}

    We follow the recommendations of~\citep{mcculloch} by setting the
    hyperparameters to $w=4.0$, $\lambda=\hat{\sigma}_1^2$ and $v^2=10/\lambda$,
    where $\hat{\sigma}_1^2$ is the variance estimate of the residuals for the
    saturated linear model $\gamma=(1, \dots, 1)$.

    \paragraph{Close-form expression for $\pi(\gamma|\mathcal{D})$}

    For a model $\gamma\in \{0, 1\}^d$, let $Z_{\gamma} = [Z_i]_{i/\gamma_i = 1}$
    be the selected covariates and let $b_{\gamma} = Z_{\gamma}^{\top}y$. Consider
    the Cholesky decomposition $C_{\gamma, v}C_{\gamma, v}^{\top} =
    Z_{\gamma}^{\top}Z_{\gamma}+v^{-2}I_{\norm{\gamma}_1}$, and define the
    least squares estimate for the residuals based on the model given by $\gamma$,
    $\hat{\sigma^2}_{\gamma, v} = \frac{1}{n}(y^{\top}y - (C_{\gamma,
    v}^{-1}b_{\gamma})^{\top}(C_{\gamma, v}^{-1}b_{\gamma}))$. Then, the
    log-posterior for $\gamma$ up to the log-partition constant is given by
    \begin{equation*}
        \label{eq:posterior_selection} \log \pi(\gamma\mid \mathcal{D}) = -
        \sum_{i=1}^{\norm{\gamma}_1} \log c_{i,i}^{(\gamma, v)}  - \norm{\gamma}_1
        \log(v) - \frac{w+n}{2}\log(w\lambda /n +\hat{\sigma^2}_{\gamma, v}).
    \end{equation*}

    \paragraph{SMC, extra numerical results}

    As a benchmark, we compute the posterior marginal probabilities of inclusion
    using a waste-free variant of the tempering SMC algorithm of \cite{Schäfer2013}
    with chain length $P=10^4$ and $N=10^5$ particles.

    Given any probability vector $p\in [0, 1]^d$, we plot the histogram of the
    variable $\log(\pi^*(\gamma)/q(\gamma\mid  p))$ with $\gamma \sim q(\cdot \mid
    p)$ (Bernoulli product). The pendant for the SMC discrete measure
    is obtained by replacing $q$ with the SMC empirical measure $\hat{\pi}^*$. In
    Figure~\ref{fig:concrete_scores}  we plot the histograms when $\gamma$ is
    distributed according to the SMC empirical distribution $\hat{\pi}^*$, and
    when $\gamma$ is distributed according to three different mean-field Bernoulli
    distributions  $\gamma\sim q(\cdot\mid p)$: i) $p = \left(1, \frac{1}{2},\dots,
                                                            \frac{1}{2}\right)$,  i.e., the intercept is always included and the other
    coordinates have $0.5$ probability to be included,  ii) the LSVI estimates, and
    iii) the marginal posterior probabilities estimated via SMC.

    \begin{figure}[ht]
        \begin{center}
            \centerline{\includegraphics[height=7cm]{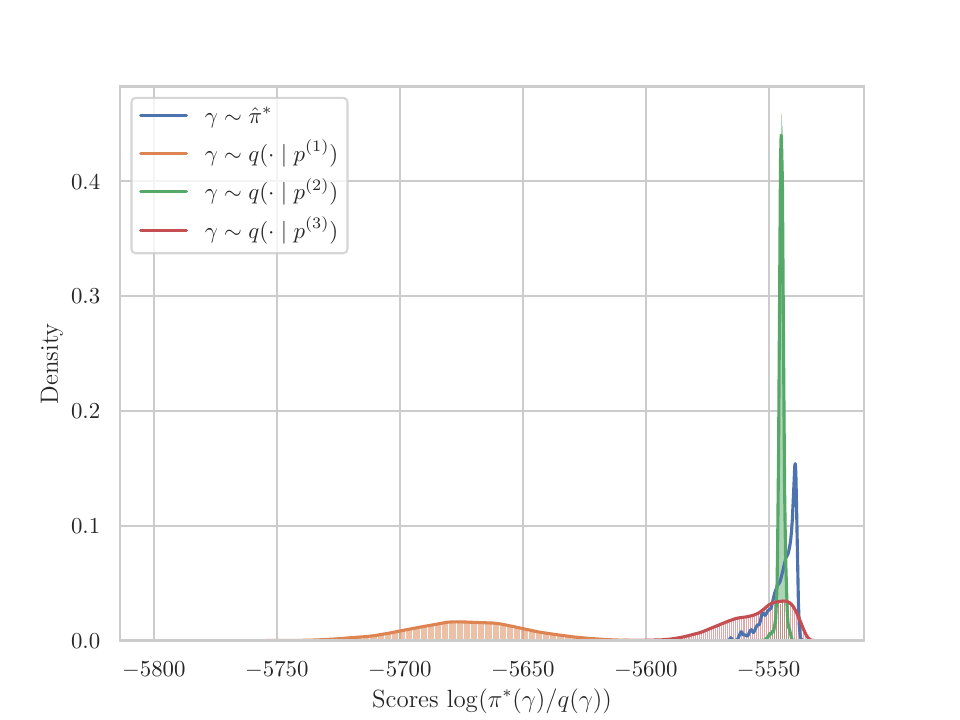}}
        \end{center}
        \caption{Variable selection example: distribution of scores $\log
        \pi^*(\gamma)/q(\gamma)$ when $\gamma\sim q=\hat{\pi}^*$, when
            $\gamma\sim q=q(\cdot\mid  p^{(i)})$ with $p^{(i)}$ given either by i),
            ii) or iii)).}
        \label{fig:concrete_scores}
    \end{figure}

    \subsection{BSL and toads displacement model}\label{app:extra_bsl}

    \paragraph{Model}

    The model assumes that $M$ toads move along a one-dimensional axis during $D$
    days. For any day $1\leq t\leq D$, the toad labelled by $1\leq i\leq M$, has
    observed position $y_{i,t}$. During the night of day $t+1$, the toad moves
    according to an overnight displacement, $\delta y_{i,t}$ which is assumed to be
    a Lévy-alpha stable distribution with stability parameter $\alpha $ and scale
    parameter $\delta$. With probability $p_0$, the toad takes refuge at $y_{i,t} +
    \delta y_{i,t}$. With probability $1-p_0$, the toad moves back to one of the
    previously explored sites $y_{i,t'}$ with $t'$ chosen uniformly in $1,\dots,
    t$. Finally, for any day $1\leq t<D$ the observed position is
    \begin{equation}
        y_{i,t+1} = B_{i,t}(y_{i,t} + \delta y_{i,t}) + (1-B_{i,t}) y_{i,t'},
    \end{equation}
    with $B_{i,t}\sim \textup{Ber}(p_0)$, $t'\sim \mathcal{U}\{1,\dots, t\}$ and
    $\delta y_{i,t}\sim \textup{Lévy-alpha}(\alpha, \delta)$, all variables being
    mutually independent. The initial position $y_{i}^{(1)}$ is set to $\delta
    y_{i}^{(0)}\sim \textup{Lévy-alpha}(\alpha, \delta)$. The model is parametrised
    by $\theta = (\alpha, \delta, p_0)\in [1,2]\times[0, 100]\times [0,
        0.9]\coloneq \Theta$. Simulating from the previous model yields the observed
    data $Y = (y_{i,t})_{1\leq t \leq D, 1\leq i \leq M}$.

    \paragraph{Summary statistic}
    The summary statistic is the concatenation of $4$ sets of statistics of size
    $12$ resulting in a total statistic of dimension $48$. Each subset is computed
    from the displacement information of lag $l$ for $l\in \{1, 2, 4, 8\}$, denoted
    by $Y_l = (\abs{y_{i,t}-y_{i,t+l}})_{1\leq t\leq D-l, 1\leq i\leq M}$. If the
    displacement from day $t$ to day $t+l$ of the toad $i$, $Y_l^{(i,
    t)}=\abs{y_{i,t}-y_{i,t+l}}$ is less than $10$, it is assumed the toad has not
    moved. The first statistic is the number of pairs $(i,t)$ such that $Y_l^{(i,
    t)}\leq 10$. We then compute the median displacement and the log difference
    between adjacent $p$-quantiles with $p=0, 0.1, \ldots, 1$ for all the
    displacements greater than $10$.

    \paragraph{Truncated Gaussian distributions approximation}
    The dataset $Y$ is generated with $(M, D)=(66, 63)$ and underlying $\theta^* =
    (1.7, 35, 0.6)$. The mean and covariance estimates are obtained with
    $P=100$ samples for each evaluation of the synthetic likelihood.
    We follow the methodology of~\citep{bsl_review} and use~\citep{warton}
    shrinkage covariance estimate given by $\hat{\Sigma} = \hat{D}^{1/2}(\gamma \hat{C}+(1-\gamma)I)\hat{D}^{1/2}$
    where $\hat{C}$ is the estimated correlation matrix and $\gamma=0.5$ is
    the regularization parameter. The prior distribution is the uniform distribution
    over $\Theta$. The variational family is the set of truncated Gaussian
    distributions over $\Theta$ with diagonal covariances. The initial distribution
    has mean $\mu = (1.5, 50, 0.5)$ and diagonal covariances $\sigma^2 = (0.05, 10,
    0.01)$. We run Algorithm~\ref{alg:lsvi_krasnoselskii} with $N=100$ samples and
    $T=50$ iterations, the step sizes are obtained by
    Algorithm~\ref{alg:recipe_backracking_variance_control} with $u=1$ and linearly
    decreasing step sizes.

    \paragraph{Full-covariance Gaussian distributions on transformed parameters}

    To constrain the parameters $\theta$, we perform inference on the transformed
    parameters $g(\theta) = \textup{logit}(g_i(\theta_i))$ with $g_i(\theta_i) =
    (\theta_i - a_i)/b_i$, with $a_i$, $b_i$ such that $g_i$ scales $\theta_{i}$ to
    $[0, 1]$. The prior distribution on the unconstrained parameters $\theta'$ is
    $\mathds{1}_{\Theta}\circ g^{-1}(\theta') \times \abs{\nabla g^{-1}
        (\theta')}$. The variational family is the set of full-covariance Gaussian
    distributions. The initial distribution for $\theta'$ has mean $\mu' = (0, 0,
    0)$ and covariance matrix $\Sigma' = \diag(0.1, 0.1, 0.1)$. The benchmark is
    obtained via MCMC with random walk step $N(0, 0.1I_3)$, the acceptance
    rate over the chain of length $10^4$ is roughly $31\%$, excluding
    the first $10^3$ states.

    \section{Proofs}
    \label{appendix:proofs}

    \subsection{First order condition and critical points of the uKL objective}
    \label{proof:first_order_condition_proof}
    \begin{proof}[Proof of Proposition~\ref{prop:first_order_optimality}]
        Injecting $\pi=\exp\left(f\right)$ and $q_\eta=\exp\left(\eta^{\top}s\right)$ into the objective function, we obtain
        \begin{equation}
            \label{eq:kl_minimization}
            \uKL(q_{\eta}\mid \pi) = \int ( \eta^{\top}s - f) q_{\eta}
            + \int \pi - \int \exp (\eta^{\top}s) .
        \end{equation}
        Using~\eqref{eq:kl_minimization}, we have
        \begin{equation}
            \label{eq:first_order_optimality_2}
            \nabla_{\eta}\uKL = \int ss^{\top}q_{\eta} \eta
            - \int s f q_{\eta}  .
        \end{equation}
        Writing the first-order optimality condition for the following minimisation problem
        \begin{equation}
            \label{eq:objective_1_kl_min}
            \begin{split}
                \eta^{*} \in \argmin_{\eta\in \mathcal{V}} \uKL(q_{\eta}\mid \pi)
            \end{split}
        \end{equation}
        and applying~\eqref{eq:first_order_optimality_2}, then dividing by $Z_{\eta}<\infty$, yield $\bbE_{\eta}[ss^{\top}]\eta= \bbE_{\eta}[f s]$.
        Let $s = (1, \bar{s}^{\top})^{\top}$ be some fixed statistic with first component $1$.
        Assume that $\eta=(\eta^{(0)}, \bar{\eta}^{\top})^{\top}\in \mathcal{V}$ is a critical point, i.e., $\nabla_{\eta} \uKL(q_\eta \mid \pi) = 0$.
        We have
        \begin{equation}
            \label{eq:gradfirstcomponent}
            \partial_{\eta^{(0)}} \uKL(q_{\eta}\mid \pi) = \eta^{\top}\int s q_{\eta} - \int f q_{\eta}.
        \end{equation}
        Injecting $\eta^{\top}s = \eta^{(0)} + \bar{\eta}^{\top}\bar{s}$ into~\eqref{eq:gradfirstcomponent}, setting $\nabla_{\eta^{(0)}} \uKL(q_{\eta}\mid \pi) = 0$ and normalising by $Z_{\eta}$ yields
        \begin{equation}
            \label{eq:eta0}
            \eta^{(0)} = \bbE_{\eta}\left[f-\bar{\eta}^{\top}\bar{s}\right],
        \end{equation}
        from the definition of the KL divergence, we deduce
        \begin{equation}
            \eta^{(0)} = -\KL(\bar{q}_{\bar{\eta}}\mid \bar{\pi})
            + \log \left(Z(\pi) / \int_{\setX} \exp\left( \bar{\eta}^\top
                                                       \bar{s}\right)\right).
        \end{equation}
    \end{proof}

    \begin{proof}[Proof of Proposition~\ref{prop:min_ukl_kl}]
        \label{proof:min_ukl_kl}
        We have
        \[\KL(\bar{q}_{\bar{\eta}}\mid \bar{\pi})
        = Z^{-1}_{\bar{\eta}}\int e^{\bar{\eta}^{\top}\bar{s}}(\bar{\eta}^{\top}\bar{s}-f)
        - \log Z_{\bar{\eta}} + \log Z_{\pi}.
        \]
        Computing the gradient of the KL requires computing the gradients of $Z_{\eta}$, $\log Z_{\eta}$, and $Z_{\eta}^{-1}$.
        We have $\nabla_{\bar{\eta}} Z_{\bar{\eta}} = \int \bar{s} e^{\bar{\eta}^{\top}\bar{s}}$, and $\nabla_{\bar{\eta}} \log(Z_{\bar{\eta}}) = Z_{\bar{\eta}}^{-1}\nabla_{\bar{\eta}}Z_{\bar{\eta}} =  \bbE_{\eta} [\bar{s}]$.
        Similarly, $\nabla_{\bar{\eta}} Z_{\bar{\eta}}^{-1} = -Z_{\bar{\eta}}^{-2}\nabla_{\bar{\eta}}Z_{\bar{\eta}} = -Z_{\bar{\eta}}^{-1} \bbE_{\eta}[\bar{s}]$.
        Then, using the previous equalities, the gradient of the KL is
        \begin{equation}
            \label{eq:gradkl}
            \begin{split}
                \nabla_{\bar{\eta}} \KL(\bar{q}_{\bar{\eta}}\mid \bar{\pi}) &= \nabla_{\bar{\eta}} \left(Z_{\bar{\eta}}^{-1}\int e^{\bar{\eta}^{\top}\bar{s}} (\bar{\eta}^{\top}\bar{s} - f) \right) - \nabla_{\bar{\eta}} \log(Z_{\bar{\eta}})\\
                &= \nabla_{\bar{\eta}}Z_{\bar{\eta}}^{-1}\times \int e^{\bar{\eta}^{\top}\bar{s}}(\bar{\eta}^{\top}\bar{s}-f) + Z_{\bar{\eta}}^{-1}\nabla_{\bar{\eta}}\left( \int e^{\bar{\eta}^{\top}\bar{s}}(\bar{\eta}^{\top}\bar{s}-f)\right) - \nabla_{\bar{\eta}} \log(Z_{\bar{\eta}})\\
                &=  -  \bbE_{\eta}[\bar{s}]  \bbE_{\eta}[\bar{\eta}^{\top}\bar{s}- f] +  \bbE_{\eta}[\bar{s}\bar{s}^{\top}\bar{\eta} - \bar{s}f + \bar s] -  \bbE_{\eta}[\bar s] \\
                & = -  \bbE_{\eta}[\bar{s}]  \bbE_{\eta}\left[\bar{\eta}^{\top}\bar{s}- f\right]
                +  \bbE_{\eta}\left[\bar{s}\bar{s}^{\top}\bar{\eta}- \bar{s}f\right].
            \end{split}
        \end{equation}
        Now, let us compute the gradient of the $\uKL$ objective with respect to $\eta=(\eta^{(0)}, \bar{\eta}^{\top})^{\top}$.
        Using
        \begin{equation}
            s s^{\top} = \begin{pmatrix}
                             1 & \bar{s}^{\top}\\ \bar{s} & \bar{s}\bar{s}^{\top}
            \end{pmatrix}
        \end{equation}
        to expand~\eqref{eq:first_order_optimality_2} yields
        \begin{equation}
            \begin{split}
                \label{eq:ukl_grad0}
                \partial_{\eta^{(0)}} \uKL(q_{\eta}\mid \pi) &= \int (\eta^{(0)} + \bar{s}^{\top}\bar{\eta})q_{\eta} - \int f q_{\eta},\\
                \nabla_{\bar{\eta}} \uKL(q_{\eta}\mid \pi) &= \int (\bar{s}\eta^{(0)} + \bar{s}\bar{s}^{\top}\bar{\eta}) q_{\eta} - \int \bar{s} f q_{\eta}.
            \end{split}
        \end{equation}
        Assume that $\nabla_{\eta}\uKL(q_{\eta}\mid \pi) = 0$, then from~\eqref{eq:ukl_grad0}, we obtain $ \bbE_{\eta} [\bar{s}\eta^{(0)} + \bar{s}\bar{s}^{\top}\bar{\eta}] - \bbE[\bar{s}f] = 0$.
        Reinjecting the previous equality into the gradient of the KL~\eqref{eq:gradkl} yields
        \begin{equation}
            \label{eq:gradklafterukl0}
            \nabla_{\bar{\eta}}\KL(\bar{q}_{\bar{\eta}}\mid \bar{\pi}) = -  \bbE_{\eta}[\bar{s}] \bbE_{\eta}[\bar{s}^{\top}\bar{\eta}- f] -  \bbE_{\eta}[\bar{s}]\eta^{(0)}.
        \end{equation}
        Injecting the expression for $\eta^{(0)}$~\eqref{eq:eta0} into~\eqref{eq:gradklafterukl0} yields $\nabla_{\bar{\eta}}\KL(\bar{q}_{\bar{\eta}}\mid\bar{\pi}) = 0$.
        Conversely, the previous computations show that if $\nabla_{\bar{\eta}}\KL(\bar{q}_{\bar{\eta}}\mid \bar{\pi})=0$ and $\partial_{\eta^{(0)}}\uKL(q_{\eta}\mid \pi) = 0$, then $\nabla_{\eta}\uKL(q_{\eta}\mid \pi)=0$.
    \end{proof}

    \subsection{The exact LSVI is a natural gradient descent}
    \begin{proof}[Proof of Proposition~\ref{prop:gradient_descent}]
        Assumption~\ref{ass:minimality_and_regularity} ensures that for any $\eta\in \mathcal{V}$, $F_{\eta}$ is invertible (minimality assumption), and ensures the differentiability of all the involved functions (regularity).
        Let us denote by $\nabla_\eta l $ the Jacobian of $\eta\mapsto l(\eta)$.
        Let $\eta\in \mathcal{V}$, using~\eqref{eq:first_order_optimality_2}, we have
        \begin{equation}
            \label{eq:gradient_kl_proof}
            \nabla_{\eta} l(\eta) = Z_{\eta}(F_{\eta}\eta - z_{\eta}),
        \end{equation}
        where $Z_{\eta}=Z(q_\eta)$ is the normalisation constant of $q_{\eta}$.
        Let $(\eta_t)$ be the sequence obtained via natural gradient descent given by~\eqref{eq:first_iteration_gd}.
        Then, by~\eqref{eq:first_iteration_gd} and~\eqref{eq:gradient_kl_proof}, we have
        \begin{equation}
            \begin{split}
                \eta_{t+1} &= \eta_t - \frac{\varepsilon_t}{Z_{\eta_{t}}} F_{\eta_{t}}^{-1}\nabla_{\eta}l(\eta_t)\\
                &= \eta_{t}- \varepsilon_t F_{\eta_{t}}^{-1}   (F_{\eta_{t}}\eta_{t} - z_{\eta_{t}})\\
                &= \left(1-\varepsilon_t \right) \eta_{t} + \varepsilon_t F_{\eta_{t}}^{-1}z_{\eta_{t}}\\
                &=\left(1-\varepsilon_t\right) \eta_{t} + \varepsilon_t \phi(\eta_{t}).
            \end{split}
        \end{equation}
        Thus, the LSVI iteration with learning schedule $(\varepsilon_t)$ given by~\eqref{eq:krasnoselskii-mann} is the natural gradient descent $(\eta_{t})$ with learning schedule $(\varepsilon_t/Z_{\eta_t})$ given by~\eqref{eq:first_iteration_gd}.
        Let us now prove~\eqref{eq:second_iteration_gd}.
        We have
        \begin{equation}
            \label{eq:nablaetaomega}
            \begin{split}
                \nabla_{\eta} \omega &=  \nabla^2_{\eta}Z\\
                &=  \int s s^{\top} q_{\eta}\\
                &=Z_{\eta}F_{\eta}.
            \end{split}
        \end{equation}
        By the chain rule and~\eqref{eq:nablaetaomega}, the Jacobian of $\eta \mapsto l(\eta)$ is
        \begin{equation}
            \label{eq:nablaetal}
            \begin{split}
                \nabla_{\eta} l  &=  \nabla_{\eta} \omega \times \nabla_{\omega}  l \\
                &= Z_{\eta} F_{\eta} \times \nabla_{\omega}  l .
            \end{split}
        \end{equation}
        Finally, injecting~\eqref{eq:nablaetal} into~\eqref{eq:first_iteration_gd} yields
        \begin{equation}
            \eta_{t+1} = \eta_{t} -  \varepsilon_t \nabla_{\omega} l(\omega(\eta_{t})),
        \end{equation}
        which is~\eqref{eq:second_iteration_gd}.
        This shows the first equivalence.
        Let $\omega_0\in \mathcal{W}$, and define $(\omega_t)$ as given by~\eqref{eq:mirror_descent}.
        The first order condition on the minimisation problem~\eqref{eq:mirror_descent} yields
        \begin{equation}
            \label{eq:mirror_descent1}
            \nabla Z^\star(\omega_{t+1}) = \nabla_{\omega}Z^\star(\omega_t)
            - \varepsilon_t \nabla_{\omega}l(\omega_t),
        \end{equation}
        but $\nabla Z^*(\omega_t) = \eta(\omega_t) = \eta_t$, thus~\eqref{eq:mirror_descent1} is exactly~\eqref{eq:second_iteration_gd}.
    \end{proof}

    \subsection{The exact LSVI mapping for mean-field Gaussian distributions}
    \label{proof:mfg}
    \label{appendix:mf}
    Let $s(x)\coloneq (1, x, x^2)^{\top}$ where $x = (x_1, \dots, x_d)$
    and $x^2 = (x_1^2,\dots, x_d^2)$.
    The set of admissible natural parameters is given by
    $\mathcal{V} = \bbR \times \bbR^d \times (\bbR^{-}\backslash\{0\})^d \subset
    \bbR^{m}$, $m=2d+1$. Let $\eta = (\eta^{(0)}, \eta^{(1),\top}, \eta^{(2),\top})^{\top}\in \mathcal{V}$.
    The natural mapping from $\eta$ to $(\mu, \sigma^2)$ is given by $T(\eta) \coloneq \left(-\frac{1}{2}\eta^{(1)}\otimes \eta^{(2), -1}, -\frac{1}{2}\eta^{(2), -1}\right)$ where $\otimes$ is the componentwise product and $\eta^{(2), -1}$ is the componentwise inverse of $\eta^{(2)}$.

    \begin{lemma}[Reparametrisation of the regression in the mean-field case]
        \label{lemma:reparametrization_mfg}
        Let $X\sim N(\mu,\sigma^2)$ be a mean-field Gaussian distribution with $\mu, \sigma\in \bbR^d$, and $\sigma_i > 0$ for all $i\in\{1,\ldots, d\}$.
        Let $\eta=(\eta^{(0)},\eta^{(1),\top},\eta^{(2),\top}) \in \mathcal{V}$, $\eta^{(0)}\in \bbR$, $\eta^{(1)}\in \bbR^d$, $\eta^{(2)}\in \bbR^d$ be the natural parameter associated with $X$ for the statistic $s : x\in \bbR^d \mapsto (1, x, x^2)^{\top}\in \bbR^{1+2d}$.
        Let $t$ be given by~\eqref{eq:statistic_mfg}.
        For any $z\in\bbR^d$, let $x(z)=\mu +\sigma\otimes z$, if $Z\sim N(0, I)$, then $x(Z)\sim N(\mu, \sigma^2)$.
        Let $\gamma = (\gamma^{(0)}, \gamma^{(1),\top}, \gamma^{(2), \top})^{\top}\in \bbR^{2d+1}$, $\gamma^{(0)}\in\bbR$, $\gamma^{(1)}\in\bbR^d$, $\gamma^{(2)}\in \bbR^d$
        be defined componentwise by
        \begin{equation}
            \label{eq:reparam_mfg}
            \gamma = \begin{pmatrix}
                         \eta^{(0)} + \eta^{(1),\top}\mu + \eta^{(2),\top}(\mu_j^2 + \sigma_j^2)_j \\
                         \eta^{(1)}\otimes \sigma + 2 \eta^{(2)}\otimes \mu\otimes \sigma          \\
                         \sqrt{2}\eta^{(2)} \otimes \sigma^2                                   \\
            \end{pmatrix}.
        \end{equation}
        Then, for any $z\in\mathbb{R}^d$
        \begin{equation}
            \label{eq:equation_to_besatisfied_reparametrization_mfg}
            \gamma^{\top}t(z) = \eta^{\top}s(x(z)).
        \end{equation}

        \begin{proof}
            Let us identify $\gamma$ such that~\eqref{eq:equation_to_besatisfied_reparametrization_mfg} is satisfied.
            Suppose that for all $z\in\mathbb{R}^d$, we have~\eqref{eq:equation_to_besatisfied_reparametrization_mfg}, then
            \begin{equation}
                \begin{split}
                    \eta^{\top}s(x) &= \eta^{(0)} + \eta^{(1),\top} x + \eta^{(2),\top}x^2\\
                    &= \eta^{(0)}+\eta^{(1),\top}\mu + \eta^{(1),\top}(\sigma\otimes z) + \eta^{(2), \top}(\mu_j^2)_j \\
                    &\qquad + 2\eta^{(2),\top}(\mu_j \sigma_j z_j)_j + \eta^{(2),\top}(\sigma_j^2 z_j^2)_j\\
                    &=\underbrace{ \eta^{(0)} + \eta^{(1),\top}\mu}_{\textup{terms in group $1$}} +\underbrace{(\eta^{(1)} \otimes \sigma)^{\top} z}_{\textup{term in group $2$}} + \underbrace{\eta^{(2),\top}(\mu_j^2)_j}_{\textup{term in group $1$}} \\
                    &\qquad + \underbrace{2 (\eta^{(2)}\otimes \sigma\otimes \mu)^{\top}z}_{\textup{term in group $2$}} + \underbrace{(\eta^{(2)}\otimes \sigma^2)^{\top}1_d}_{\textup{term in group $1$}} \\
                    &\qquad + \underbrace{(\eta^{(2)}\otimes\sigma^2)^{\top}(z_j^2-1)_j}_{\textup{term in group $3$}}\\
                    &=\gamma^{\top}t(z).
                \end{split}
            \end{equation}
            By identifying the factors in front of $1$ (group $1$), the $z_j$'s (group $2$), and the $z_j^2$'s (group $3$), we obtain~\eqref{eq:reparam_mfg}.
            By injecting~\eqref{eq:reparam_mfg} into~\eqref{eq:equation_to_besatisfied_reparametrization_mfg}, the equality is satisfied.
        \end{proof}
    \end{lemma}

    \begin{theorem}[LSVI mapping $\phi$ for the mean-field Gaussian distributions]
        \label{th:mfg}
        Let $X\sim N(\mu, \sigma^2)$, and $\eta\in \mathcal{V}$ be the corresponding natural parameter and let $t$ be given by
        \begin{equation}
            \label{eq:statistic_mfg}
            t(z) \coloneq \left(1, z^{\top}, \frac{z_1^2-1}{\sqrt{2}}, \dots, \frac{z_d^2 -1}{\sqrt{2}}\right)^{\top}.
        \end{equation}
        Then, the LSVI mapping $\beta \coloneq \phi(\eta)$ is defined recursively bottom to top by
        \begin{equation}
            \label{eq:nicolas_scheme1_bis}
            \beta=\begin{pmatrix}
                      \gamma^{(0)} - \beta^{(1),\top}\mu - \beta^{(2),\top}(\mu^2+\sigma^2) \\
                      \gamma^{(1)}\otimes \sigma^{-1} - 2\beta^{(2)}\otimes\mu              \\
                      \gamma^{(2)}\otimes \left(\sqrt{2}\sigma^2\right)^{-1}
            \end{pmatrix}
        \end{equation}
        and $\gamma \coloneq \bbE\left[t(Z)f(\mu+\sigma \otimes Z)\right]$,
        with subcomponents $\gamma=(\gamma^{(0)}, \gamma^{(1),\top}, \gamma^{(2),\top})^{\top}$, $\gamma^{(0)}\in\bbR$, $\gamma^{(1)}, \gamma^{(2)}\in \bbR^d$.
        In addition, if $f$ admits second-order derivatives such that $\bbE_{X}
        [f]<\infty$, $\norm{\bbE_{X}[\nabla f]}< \infty$, and $0\prec
        -\bbE_{X}\left[\textup{Diag}(\nabla^2 f)\right]$, then $\phi(\eta)$ defines a
        Gaussian distribution with mean and variance given by
        \begin{align}
            \label{eq:canonical_parameters_mfg_gaussian}
            (\mu',\sigma'^{2})= \left(\mu - \left(\bbE\left[ \diag(\nabla^2 f)(X)\right]\right)^{-1} \otimes
                                    \bbE\left[ \nabla f (X)\right], -\left(\bbE\left[ \diag(\nabla^2 f)(X)\right]\right)^{-1}\right).
        \end{align}
        \begin{proof}
            We know that $\phi(\eta)$ realises the minimum of the OLS objective~\eqref{eq:LSVI_as_OLS}, i.e.,
            \begin{equation}
                \label{eq:ols2}
                \begin{split}
                    \phi(\eta) &= \argmin_{\beta\in \bbR^m} \bbE_{X\sim N(\mu, \sigma^2)}\left[\left(\beta^{\top}s(X) - f(X)\right)^2\right].
                \end{split}
            \end{equation}
            Using Lemma~\ref{lemma:reparametrization_mfg}, we can rewrite the regression objective with covariates given by $s$ into a regression with covariates given by $t$.
            Using the notations of Lemma~\ref{lemma:reparametrization_mfg}, we let $\gamma$ be given such that $\gamma^{\top}t(z) = \beta^{\top}s(x(z))$ for all $z\in\bbR^d$, and where $\beta=\phi(\eta)$ is the unique minimizer of the OLS objective~\eqref{eq:ols2}.
            Then,
            \begin{equation}
                \label{eq:nicolas_scheme0}
                \begin{split}
                    \gamma &= \argmin_{\gamma\in \bbR^m} \bbE_{Z}\left[\left(\gamma^{\top}t(Z) - f(\mu + \sigma \otimes Z)\right)^2\right]\\
                    &= \left(\bbE_{Z}\left[ tt^{\top}(Z)\right]\right)^{-1} \bbE_{Z}\left[t(Z)f(\mu+\sigma \otimes Z)\right]\\
                    &= \bbE_{Z}\left[t(Z)f(\mu+\sigma \otimes Z)\right],
                \end{split}
            \end{equation}
            since $\bbE_{Z} \left[tt^{\top}(Z)\right] = I_m$.
            Inverting the relation~\eqref{eq:reparam_mfg} given by Lemma~\ref{lemma:reparametrization_mfg} between $\gamma$ and $\beta$, which is possible since all the $\sigma_i$'s are strictly positive, we obtain
            \begin{equation}
                \begin{split}
                    \label{eq:nicolas_scheme1}
                    \beta &= \begin{pmatrix}
                                 \gamma^{(0)} - \beta^{(1),\top}\mu - \beta^{(2),\top}(\mu^2+\sigma^2) \\
                                 \gamma^{(1)}\otimes \sigma^{-1} - 2\beta^{(2)}\otimes\mu              \\
                                 \gamma^{(2)}\otimes \left(\sqrt{2}\sigma^2\right)^{-1}                       \\
                    \end{pmatrix}.
                \end{split}
            \end{equation}
            But $\beta = \phi(\eta)$, this proves the first statement~\eqref{eq:nicolas_scheme1_bis} of Theorem~\ref{th:mfg}.
            For the second statement, assume that $f$ admits second-order derivatives.
            Using Stein's Lemma and~\eqref{eq:nicolas_scheme0}, we obtain,
            \begin{equation}
                \begin{split}
                    \label{eq:it_gamma}
                    \gamma  &= \bbE_{Z\sim N(0, I_d)} [t(Z)f(\mu+\sigma\otimes Z)]\\
                    &= \begin{pmatrix}
                           \bbE_{Z}\left(f(\mu+\sigma\otimes Z)\right)                                                    \\
                           \sigma\otimes \bbE_{Z}(\nabla f)(\mu+\sigma \otimes Z)                                         \\
                           \frac{1}{\sqrt{2}}\left(\sigma^2 \otimes \bbE_{Z}\diag(\nabla^2 f)(\mu+\sigma\otimes Z)\right)\\
                    \end{pmatrix}.
                \end{split}
            \end{equation}
            Injecting~\eqref{eq:it_gamma} into~\eqref{eq:nicolas_scheme1_bis}, we obtain for
            \begin{equation}
                \label{eq:it_eta}
                \phi(\eta) = \begin{pmatrix}
                                 \phi(\eta)_0                                                 \\
                                 \bbE_{X}(\nabla f)(X) - \mu\otimes \bbE \diag(\nabla^2 f)(X) \\
                                 \frac{1}{2}\bbE \diag(\nabla^2 f) (X)                        \\
                \end{pmatrix}.
            \end{equation}
            Using the natural mapping $T(\eta) = \left(-\frac{1}{2}\eta^{(1)} \eta^{(2), -1}, -\frac{1}{2}\eta^{(2), -1}\right)$, we obtain~\eqref{eq:canonical_parameters_mfg_gaussian}.
        \end{proof}
    \end{theorem}

    \subsection{The exact LSVI mapping for Gaussian distributions}
    \begin{lemma}[Reparametrisation of the regression in the full-covariance case]
        \label{lemma:reparametrization_fullcov}
        Let $X\sim N(\mu,\Sigma)$ be a Gaussian distribution with $\mu\in \bbR^d$, and $\Sigma\succ 0$.
        Let $\eta \in \mathcal{V}$ be the natural parameter associated with $X$ for the statistic $s : x\in \bbR^d \mapsto (1, X, (\vect{XX^{\top}})^{\top})^{\top}\in \bbR^{1+d+d^2}$.
        Let $t$ be given by~\eqref{eq:statistic_full_gaussian_case}.
        For any $z\in\bbR^d$, let $x(z)=\mu +C z$ with $C\in\mathbb{R}^{d\times d}$ such that $CC^{\top} =  \Sigma$.
        If $Z\sim N(0, I_d)$, then $x(Z)\sim N(\mu, \Sigma)$, and for any $z\in\mathbb{R}^d$
        \begin{equation}
            \label{eq:equation_to_besatisfied_reparametrization_fullcov}
            \gamma^{\top}t(z) = \eta^{\top}s(x(z)),
        \end{equation}
        with $\gamma = (\gamma^{(0)}, \gamma^{(1),\top}, \gamma^{(2),\top})^{\top}\in \bbR^{1+d+d(d+1)/2}$.
        Furthermore, the components of $\gamma$, $\gamma^{(0)}\in\bbR$, $\gamma^{(1)}\in\bbR^d$, $\gamma^{(2)}\in \bbR^{d(d+1)/2}$ are given by
        \begin{equation}
            \label{eq:reparam_fullcov}
            \gamma = \begin{pmatrix}
                         \eta^{(0)} + \eta^{(1),\top}\mu +\eta^{(2),\top}\vect{\mu\mu^{\top}} + \sum_{i=1}^{d}\Gamma_{i,i}\\
                         C^{\top}\eta^{(1)} + 2\left(\mu\otimes C\right)^{\top}\eta^{(2)}                                  \\
                         \gamma^{(2)}
            \end{pmatrix},
        \end{equation}
        where
        \begin{equation}
            \Gamma = \unvec{\left((C\otimes C)^{\top}\eta^{(2)}\right)},
        \end{equation}
        and
        \begin{align}
            \gamma^{(2)} = \left(
                               \sqrt{2}\Gamma_{1, 1}, 2\Gamma_{1, 2}, \dots, 2\Gamma_{1, d},\sqrt{2}\Gamma_{2, 2}, 2\Gamma_{2, 3},\dots, 2\Gamma_{2, d}, \dots, \sqrt{2}\Gamma_{d, d}\right)^{\top}.
        \end{align}
        \begin{proof}
            The proof is similar to the proof of Lemma~\ref{lemma:reparametrization_mfg}.
            Let us rewrite the regression with respect to $Z$.
            Let $\eta = (\eta^{(0)}, \eta^{(1),\top}, \eta^{(2),\top})^{\top}\in \bbR^{1+d+d^2}$ with $\eta^{(0)}\in \bbR$, $\eta^{(1)}\in \bbR^d$, $\eta^{(2)}\in \bbR^{d^2}$.
            Let $X = \mu + C Z$ with $C$ such that $CC^{\top} = \Sigma$.
            Rewriting the linear regression on $s(X)$ with $s(Z)$, we have
            \begin{equation}
                \label{eq:dev_regXZ}
                \begin{split}
                    \eta^{\top} s(X) &= \eta^{(0)} + \eta^{(1),\top}\mu + \eta^{(1),\top}C Z + \eta^{(2),\top}\vect{\mu\mu^{\top}} \\
                    &\quad + \eta^{(2),\top}\vect{\mu Z^{\top}C^{\top}}+\eta^{(2),\top}\vect{CZ\mu^{\top}} \\
                    &\quad + \eta^{(2),\top}\vect{CZZ^{\top}C^{\top}}\\
                    &= \hat{\gamma}^{\top}s(Z),
                \end{split}
            \end{equation}
            where $\hat{\gamma} = (\hat{\gamma}^{(0)}, \hat{\gamma}^{(1),\top},\hat{\gamma}^{(2),\top})^{\top}\in\bbR^{1+d+d^2}$ are left to be identified.
            By identifying the quadratic terms in~\eqref{eq:dev_regXZ}, we have for $\hat{\gamma}^{(2)}$
            \begin{equation}
                \eta^{(2),\top}\vect{CZZ^{\top}C^{\top}} = \eta^{(2),\top}(C\otimes C)\vect{ZZ^{\top}} = \hat{\gamma}^{(2),\top}\vect{ZZ^{\top}},
            \end{equation}
            where we used $\vect{ABC} = (C^{\top}\otimes A)\vect{B}$.
            Thus,
            \begin{equation}
                \eta^{(2)} = (C\otimes C)^{-\top} \hat{\gamma}^{(2)} =(C^{-1}\otimes C^{-1})^{\top}\hat{\gamma}^{(2)},
            \end{equation}
            where we used $(A\otimes B)^{-1} = A^{-1}\otimes B^{-1}$.
            For $\hat{\gamma}^{(1)}$, expanding the linear term in~\eqref{eq:dev_regXZ}, we have
            \begin{equation}
                \begin{split}
                    \eta^{(1),\top}C z + \eta^{(2),\top}\vect{Cz\mu^{\top}}+\eta^{(2),\top}\vect\mu z^{\top}C^{\top}&=\eta^{(1),\top}Cz + 2\eta^{(2),\top}\vect{Cz\mu^{\top}}\\
                    &=\left(\eta^{(1),\top}C + 2\eta^{(2),\top}\left(\mu\otimes C\right)\right)z\\
                    &=\hat{\gamma}^{(1),\top}z,
                \end{split}
            \end{equation}
            i.e.,
            \begin{equation}
                \hat{\gamma}^{(1)} = C^{\top}\eta^{(1)} + 2\left(\mu\otimes C\right)^{\top}\eta^{(2)}.
            \end{equation}
            Regrouping all the constants in~\eqref{eq:dev_regXZ}, we obtain for $\hat{\gamma}^{(0)}$,
            \begin{equation}
                \begin{split}
                    \hat{\gamma}^{(0)} = \eta^{(0)}+\eta^{(1),\top}\mu + \eta^{(2),\top}\vect{\mu\mu^{\top}}.
                \end{split}
            \end{equation}
            Now, we want to rewrite the regression on $s(Z)$ in terms of $t(Z)$ where
            \begin{align}
                t(Z) =
                \left(1,
                Z^{\top}                       ,
                    \frac{Z_1^2-1}{\sqrt{2}} ,
                Z_1 Z_2                  ,
                    \dots                   ,
                Z_1 Z_d                  ,
                    \frac{Z_2^2-1}{\sqrt{2}} ,
                Z_2 Z_3                  ,
                    \dots                   ,
                    \frac{Z_d^2-1}{\sqrt{2}}
                \right)^{\top},
            \end{align}
            which satisfies $\bbE_{Z}\left[tt^{\top}\right] = I_{m'}$ with $m'=d+d(d+1)/2+1$.
            We do that in two steps, let
            \begin{align}
                t_1(Z) =
                \left(1,
                Z^{\top}                       ,
                    \frac{Z_1^2-1}{\sqrt{2}} ,
                Z_1 Z_2                  ,
                    \dots                   ,
                Z_1 Z_d                  ,
                Z_1 Z_2,
                    \frac{Z_2^2-1}{\sqrt{2}} ,
                Z_2 Z_3                  ,
                    \dots                   ,
                    \frac{Z_d^2-1}{\sqrt{2}}
                \right)^{\top}.
            \end{align}
            Let $\tilde{\gamma}=(\tilde{\gamma}^{(0)}, \hat{\gamma}^{(1),\top},\tilde{\gamma}^{(2),\top})^{\top} \in \bbR^{1+d+d^2}$ be such that
            \begin{equation}
                \label{eq:equation_to_satisfy_reparametrization}
                \tilde{\gamma}^{\top}t_1(Z) = \hat{\gamma}^{\top}s(Z),
            \end{equation}
            i.e., keeping only the constant terms and the terms quadratic in $Z$,
            \begin{equation}
                \label{eq:eq_to_satisfy_reparam2}
                \begin{split}
                    &\hat{\gamma}^{(0)} + \sum_{j=1}^{d^2} \hat{\gamma}_{j}^{(2)}(\vect{ZZ^{\top}})_j= \tilde{\gamma}^{(0)} + \sum_{k=0}^{d-1}\tilde{\gamma}^{(2)}_{ 1+(d+1)k}\left\{\frac{Z_{k+1}^2-1}{\sqrt{2}}\right\} + \sum_{j\neq 1+(d+1)k} \tilde{\gamma}_{2,j} (\vect{ZZ^{\top}})_j.
                \end{split}
            \end{equation}
            We set, for any $k\geq 0$, $\tilde{\gamma}^{(2)}_{1+(d+1)k} = \hat{\gamma}^{(2)}_{1+(d+1)k}\sqrt{2}$, and $\tilde{\gamma}^{(0)} = \hat{\gamma}^{(0)} + \sum_{k=0}^{d-1}\hat{\gamma}^{(2)}_{ (1+(d+1)k)}$.
            Then,~\eqref{eq:equation_to_satisfy_reparametrization} and~\eqref{eq:eq_to_satisfy_reparam2} are satisfied.
            To go from $t_1$ to $t$, we need to get rid of the coordinates $t(Z)_{k} = Z_i Z_j$ for some $i>j$, i.e., $k\in [dp+1, (d+1)p]$ for some integer $p$.
            Let $\Gamma = \unvec(\hat{\gamma}^{(2)})$, and let $\gamma^{(2)}\in \bbR^{d(d+1)/2}$ be defined by
            \begin{align}
                \gamma^{(2)} = \left(
                                   \sqrt{2}\Gamma_{1, 1}, 2\Gamma_{1, 2}, \dots, 2\Gamma_{1, d}, \sqrt{2}\Gamma_{2, 2}, 2\Gamma_{2, 3},\dots, 2\Gamma_{2, d}, \dots, \sqrt{2}\Gamma_{d, d}
                \right)^{\top}.
            \end{align}
            Then $\gamma = [\tilde{\gamma}^{(0)}, \hat{\gamma}^{(1),\top}, \gamma^{(2),\top}]^{\top}\in\bbR^{m'}$ satisfies
            \begin{equation}
                \gamma^{\top}t(Z) = \tilde{\gamma}^{\top}t_1(Z) =\hat{\gamma}^{\top}s(Z) =\eta^{\top}s(X).
            \end{equation}
            All the previous computations give the expression of $\gamma$ as a function of $\eta$.
        \end{proof}
    \end{lemma}
    We now turn to prove Theorem~\ref{th:full_cov_gaussian} using the previous Lemma.
    \begin{proof}[Proof of Theorem~\ref{th:full_cov_gaussian}]
        As in the proof of Theorem~\ref{th:mfg}, the least squares regression on $s(X)$ can be rewritten in terms of $t(Z)$.
        Then, by applying Lemma~\ref{lemma:reparametrization_fullcov}, we can map the regressor $\gamma$ with respect to $t$, to the regressor with respect to $s$, given $\beta=\phi(\eta)$.
        Since $\bbE_{Z}[tt^{\top}] = I$, the OLS simplifies to $\gamma = \bbE_{Z}[t(Z)f(\mu+CZ)]$.
        By Lemma~\ref{lemma:reparametrization_fullcov}, the mapping from $\gamma$ to $\beta$ is given by
        \begin{equation}
            \label{eq:inversemappingfullcov}
            \beta = \begin{pmatrix}
                        \gamma^{(0)} - \sum_{i=1}^{d}\Gamma_{i,i} - \beta^{(1),\top}\mu - \beta^{(2),\top}\vect{\mu\mu^{\top}} \\
                        C^{-\top}\gamma^{(1)}-2\unvec{\beta^{(2)}}\mu                                           \\
                        \vect{\left(C^{-\top}\Gamma C^{-1}\right)}
            \end{pmatrix},
        \end{equation}
        where
        \begin{equation}
            \Gamma = \begin{pmatrix}
                         \gamma_{1}^{(2)}/\sqrt{2} & \gamma_{ 2}^{(2)}/2          & \dots  & \dots  & \gamma_{ d}^{(2)}/2               \\
                         \gamma_{2}^{(2)}/2        & \gamma_{ d+1}^{(2)}/\sqrt{2} & \dots  & \dots  & \gamma_{ 2d-1}^{(2)}/2            \\
                         \vdots                 &                          & \ddots &        & \vdots                        \\
                         \vdots                 &                          &        & \ddots & \vdots                        \\
                         \gamma_{d}^{(2)}/2        & \dots                    & \dots  & \dots  & \gamma_{d(d+1)/2}^{(2)}/\sqrt{2}
            \end{pmatrix},
        \end{equation}
        or componentwise $\Gamma_{i,i} = \gamma^{(2)}_{1+1/2(2d+2-i)(i-1)}/\sqrt{2}$, $\Gamma_{i,i+k}=\gamma^{(2)}_{1+1/2(2d+2-i)(i-1)+k}/2$ for $1\leq i\leq d$ and $1\leq k\leq d-i$, and $\Gamma_{i,j} = \Gamma_{j,i}$ for $j<i$.
        Regarding the complexity, the computation of the Cholesky matrix $C$ and its inverse requires $\OO(d^3)$ operations; consequently, the computation of $\gamma$ can also be performed in $\OO(d^3)$ operations.
        Using~\eqref{eq:inversemappingfullcov} to map $\gamma$ to $\eta$, involves computing $\vect{\left(C^{-\top}\Gamma C^{-1}\right)}$ and $C^{-\top}\gamma^{(1)}$, both of which can be performed in $\OO(d^3)$.
    \end{proof}

    \subsection{Concentration bounds for the Fisher matrix in the compact case}
    We now prove a Lemma to control the bias induced by inverting the estimated FIM, conditioned on the event that the estimated FIM is well-conditioned, which happens with high-probability given the number of samples $N$ is large enough.
    We first prove a version of the Lemma when $s$ is bounded (Lemma~\ref{lemma:inverse_g_compactA}), and then tackle the case where $s$ is unbounded but with
    bounded-moments (Lemma~\ref{lemma:polynomial}).
    \begin{lemma}[Mean error bound for the inverse of $\hat{F}$ when $s$ is uniformly bounded]
        \label{lemma:inverse_g_compactA}
        Let $\delta \in (0, 1)$, $N\geq B(8/3r+8B)r^{-2} \log(2m \delta^{-1}) $, $\omega\in \mathcal{W}_0$, and $\mathcal{A}(\omega) = [\norm{F_\omega-\hat{F}_\omega}\leq \tfrac{1}{2}\norm{F^{-1}_{\omega}}^{-1}]$.
        Then, under Assumptions~\ref{ass:minimality_and_regularity},~\ref{ass:bounded_iterates}, and $\norm{s}_2^2\leq B$, $\mathcal{A}(\omega)$ occurs with probability at least $1-\delta$.
        Furthermore,
        \begin{equation}
            \label{eq:upperboundconditionalexpectationA}
            \norm[\bigg]{\bbE \left[\hat{F}^{-1}_{\omega}-F^{-1}_{\omega} \middle| \mathcal{A}(\omega)\right]} = \OO(N^{-1}),
        \end{equation}
        where the constant in the big-$\OO$ term can be chosen independently of $\omega$.
        \begin{proof}[Proof of Lemma~\ref{lemma:inverse_g_compactA} (exponential tail bound)]
            \label{proof:exponential_tail_bound}
            Fix $\omega\in\mathcal{W}_0$.
            For the sake of notation, we drop the subscript in $\omega$ but indicate the dependency in $N$.
            For any $N\geq 1$, let $\hat{F}_N = N^{-1}\sum_{i=1}^N ss^{\top}(X_i)$ with $X_1, \ldots, X_N\overset{\textup{i.i.d}}{\sim} q $.
            Conditionally on $\mathcal{A}^{(N)}=[\norm{\hat{F}_N-F}\leq \tfrac{1}{2}\norm{F^{-1}}^{-1}]$, $\hat{F}_N=F(I-(I-F^{-1}\hat{F}_N))$ is invertible because $F$ is invertible thanks to Assumption~\ref{ass:minimality_and_regularity} and $\norm{I - F^{-1}\hat{F}_N} \leq \norm{F^{-1}}\norm{F - \hat{F}_N} \leq 1/2$.
            Using the Neumann series, we have
            \begin{equation}
                \begin{split}
                    \label{eq:lemma7firsteq}
                    \hat{F}^{-1}-F^{-1} = &\left(I-F^{-1}\hat{F} + \OO(\norm{I-F^{-1}\hat{F}}^2)\right)F^{-1}.
                \end{split}
            \end{equation}
            Since $\hat{F}_N-F$ and $I-F^{-1}\hat{F}_N$ are averages of $N$ i.i.d random matrices with zero mean and uniformly bounded second moments thanks to~\eqref{eq:fourth_order_moment} and $r>0$, $\bbE[\norm{\hat{F}_N-F}^2]$ and $\bbE[\norm{I-F^{-1}\hat{F}_N}^2]$ are $\OO(N^{-1})$ (uniformly in $\omega$), and thus $P(\mathcal{A}^{(N),c}) = \OO(N^{-1})$ by Bienaymé-Tchebychev's inequality.
            Since $\bbE[I-F^{-1}\hat{F}_N]=0$, Cauchy-Schwarz inequality yields $\norm{\bbE[(I-F^{-1}\hat{F}_N)\mathds{1}_{\mathcal{A}^{(N)}}]} = \norm{\bbE[(I-F^{-1}\hat{F}_N)\mathds{1}_{\mathcal{A}^{(N),c}}]} \leq (\bbE[\norm{I-F^{-1}\hat{F}_N}^2])^{1/2}P(\mathcal{A}^{(N),c})^{1/2} = \OO(N^{-1})$.
            Thus, taking the expectation in~\eqref{eq:lemma7firsteq} yields $\norm{\bbE[\hat{F}^{-1}-F^{-1}\mid \mathcal{A}^{(N)}]} = \OO(N^{-1})$, again using $r>0$.

            By~\citep[][Theorem~1.6.2]{tropp2015} with uniform bound $\norm{(s
            s^{\top}(X_i)-F)/N}\leq 2B/N$ and variance
            $\norm{\sum_{i=1}^{N}\bbE[((ss^{\top}(X_i)-F)/N)^2]}\leq B\norm{F}/N$, and
            the definition of $r$, we have
            \begin{align}
                \begin{split}
                    P(\norm{\hat{F}_N-F}> \tfrac{1}{2}\norm{F^{-1}}^{-1}) &\leq P(\norm{\hat{F}_N-F}\geq r/2) \\
                    &\leq 2m\exp\left(-\frac{Nr^2}{B(8/3r+8\norm{F})}\right)\\
                    &\leq 2m\exp\left(-\frac{Nr^2}{B(8/3r+8B)}\right),
                \end{split}
            \end{align}
            where to go from the second to the third line, we use $\norm{F}\leq B$.
            Setting $N \geq B(8/3r + 8B)r^{-2} \log(2m\delta^{-1})$ yields $P(\norm{\hat{F}_N-F}> \tfrac{1}{2}\norm{F^{-1}}^{-1})\leq \delta$, i.e., $P(\mathcal{A}) \geq 1-\delta$.
            The bound is independent of $\omega$, and true for any $\omega \in \mathcal{W}_0$, finally yielding the result.
        \end{proof}
    \end{lemma}
    \begin{lemma}
        \label{lemma:polynomial}
        Under Assumptions~\ref{ass:minimality_and_regularity} and~\ref{ass:bounded_iterates}, for $\sqrt{N}\geq 2r^{-1}\delta^{-1/2} m^{1/2}(\sqrt{8e\log(m)} \mu_4 \nu + 8e\, m\log(m)\mu_4^2)$, $P(\mathcal{A}(\omega))\geq 1-\delta$ and~\eqref{eq:upperboundconditionalexpectationA} holds.
        \begin{proof}[Proof of Lemma~\ref{lemma:polynomial} (polynomial tail bound)]
            We follow the proof of~\ref{proof:exponential_tail_bound}.
            The second moments of $\hat{F}$ are still uniformly bounded thanks to~\eqref{eq:fourth_order_moment}, and the same argument as in~\ref{proof:exponential_tail_bound} is still valid.
            Thus, $\norm{\bbE[\hat{F}^{-1}-F^{-1}\mid\mathcal{A}]}=\OO(N^{-1})$, and the constant inside the big-$\OO$ notation is also independent on $\omega$ using the uniform bounds on the fourth-moment of $s$.
            By the definition of $r$, the Bienaymé-Tchebychev's inequality, and~\citep[][Theorem 3.1]{10.1093/imaiai/ias001}, applied to $s(X)\in\bbR^m$ ($m\geq 3$) with the all-ones mask $M$, for which the masked estimator is $\hat{F}_N$ itself, $\norm{M}_{1\to 2} = m^{1/2}$ and $\norm{M} = m$ (its proof does not use that the random vector is centred, the matrix $\Sigma$ there being the second-moment matrix, here $F$), we have
            \begin{align}
                \begin{split}
                    P(\norm{\hat{F}_N-F}> \tfrac{1}{2}\norm{F^{-1}}^{-1}) &\leq P(\norm{\hat{F}_N-F}\geq r/2) \\
                    &\leq \bbE{\norm{\hat{F}_N - F}^2} / (r/2)^2 \\
                    &\leq \left\{2r^{-1} \left\{\sqrt{8e\frac{\log(m)}{N}}\, m^{1/2}\mu_4 \nu + 8e \frac{\log(m)}{N}\, m \left(\bbE\left[\max_{1\leq i\leq N}\norm{s(X_i)}_\infty^4\right]\right)^{1/2}\right\}\right\}^2\\
                    &\leq \left\{2r^{-1} \left\{\sqrt{8e\frac{\log(m)}{N}}\, m^{1/2}\mu_4 \nu + 8e \frac{\log(m)}{N}\, m\, (Nm)^{1/2}\mu_4^2\right\}\right\}^2\\
                    &= \left\{2r^{-1}\left(\frac{m}{N}\right)^{1/2} \left\{\sqrt{8e\log(m)}\mu_4 \nu + 8e\, m\log(m)\mu_4^2\right\}\right\}^2,
                \end{split}
            \end{align}
            where to go from the third to the fourth line, we used
            \begin{equation}
                \bbE\left[\max_{1\leq i\leq N}\norm{s(X_i)}_\infty^4\right] \leq \sum_{i=1}^{N}\bbE\left[\norm{s(X_i)}_\infty^4\right] \leq \sum_{i=1}^{N}\sum_{j=1}^{m}\bbE\left[\abs{s(X_i)_j}^4\right] \leq N m \mu_4^4.
            \end{equation}
            Setting $\sqrt{N} \geq 2r^{-1}\delta^{-1/2} m^{1/2}(\sqrt{8e\log(m)} \mu_4 \nu + 8e\, m\log(m)\mu_4^2)$ yields $P(\norm{\hat{F}_N-F}> \tfrac{1}{2}\norm{F^{-1}}^{-1})\leq \delta$, i.e., $P(\mathcal{A}) \geq 1-\delta$.
            The bound is independent of $\omega$, and true for any $\omega \in \mathcal{W}_0$.
        \end{proof}
    \end{lemma}

    \subsection{Convergence analysis of the stochastic LSVI algorithm}
    The proof of Theorem~\ref{th:main_theorem} relies on Lemmas~\ref{lemma:equivalence},~\ref{lemma:technicallemma},~\ref{lemma:controlbias}, and~\ref{lemma:bound_variance}.
    Lemma~\ref{lemma:equivalence} states the equivalence between stochastic mirror descent and stochastic natural gradient descent, the proof is very similar to the non-stochastic case (Proposition~\ref{prop:gradient_descent}).
    Lemma~\ref{lemma:technicallemma} gives the general convergence rate for stochastic mirror descent with the presence of an additional bias under the assumption the bias has bounded variance. This is a generalisation of~\citet[][Th. 4.5]{Hanzely2021}.
    Both Lemmas~\ref{lemma:controlbias},~\ref{lemma:bound_variance} are required to handle the two first moments of the bias induced by inverting the FIM estimate. This analysis requires conditioning on the event that the estimated FIMs are well-conditioned, which happens with high-probability (Lemma~\ref{lemma:polynomial}).
    Theorem~\ref{th:main_theorem} follows by successively applying Lemma~\ref{lemma:equivalence} and Lemma~\ref{lemma:technicallemma}, the latter requires Lemmas~\ref{lemma:controlbias} and~\ref{lemma:bound_variance}.
    \begin{lemma}[Equivalence between stochastic mirror descent and stochastic natural gradient descent]
        \label{lemma:equivalence}
        Define the stochastic gradient $\hat{\nabla}_{\omega}l$ by
        \begin{equation}
            \label{eq:def_stochastic_gradient}
            \hat{\nabla}_{\omega}l : \omega\mapsto \eta(\omega)-\hat{F}_{\omega}^{-1}\hat{z}_{\omega},
        \end{equation}
        given that $\hat{F}_{\omega}$ is invertible.
        Then~\eqref{eq:km_lsvi_stochastic} is equivalent to
        \begin{equation}
            \label{eq:lsvi}
            \hat{\eta}_{t+1} = \hat{\eta}_{t}-\varepsilon_t \hat{\nabla}_{\omega}l(\hat{\omega}_{t}),
        \end{equation}
        where $\hat{\omega}_{t} = \omega(\hat{\eta}_{t})$.
        Furthermore, the previous dynamic is equivalent to
        \begin{equation}
            \label{eq:argmin_estimate}
            \hat{\omega}_{t+1} = \argmin_{\omega\in \mathcal{W}} \left\{  \hat{\nabla}^{\top}_\omega l(\hat{\omega}_{t}) \omega + \frac{1}{\varepsilon_t} D_{Z^*}(\omega, \hat{\omega}_{t})\right\},
        \end{equation}
        with $\hat{\eta}_{t+1}=\eta(\hat{\omega}_{t+1})$.

        \begin{proof}
            The first equivalence follows from the same computations as in Proposition~\ref{prop:gradient_descent}.
            Let us show that iteration~\eqref{eq:km_lsvi_stochastic} can be recovered as the dual in the natural parameter space of a stochastic mirror descent, i.e., that $\hat{\eta}_{t+1} = \eta(\hat{\omega}_{t+1})$ with $(\hat{\omega}_t)$ given by~\eqref{eq:argmin_estimate} recovers~\eqref{eq:km_lsvi_stochastic}.
            The first order condition on~\eqref{eq:argmin_estimate} gives
            \begin{equation}
                \label{eq:first_order_condition}
                \nabla Z^*(\hat{\omega}_{t+1}) = \nabla Z^*(\hat{\omega}_t)-\varepsilon_t \hat{\nabla}_{\omega}l(\hat{\omega}_t).
            \end{equation}
            However, since $\nabla Z^*(\hat{\omega}_{t+1}) = \eta(\hat{\omega}_{t+1}) = \hat{\eta}_{t+1}$, the desired equivalence between the two dynamics follows.
        \end{proof}
    \end{lemma}

    \begin{lemma}[General convergence for biased stochastic mirror descent]
        \label{lemma:technicallemma}
        Let us define the bias $B_t$ of the stochastic gradient at iteration $t$ by
        \begin{equation}
            B_t = \bbE[\hat{\nabla}_{\omega}l(\hat{\omega}_t)-\nabla_{\omega}l(\hat{\omega}_t)\mid  \hat{\omega}_t],
        \end{equation}
        given that $\hat{F}_{\hat{\omega}_t}$ is invertible, and let us denote by $m(\hat{\omega}_t) \coloneq \omega_{t+1, *}$ the exact mirror-descent iterate starting from $\hat{\omega}_t$, i.e.,
        \begin{equation}
            \label{eq:exact_md_starting_point_sto}
            \omega_{t+1, *} = \argmin_{\omega\in\mathcal{W}} \left\{ \nabla^{\top}_{\omega}l(\hat{\omega}_t) \omega + \varepsilon_t^{-1}D_{Z^*}(\omega,\hat{\omega}_t)\right\}.
        \end{equation}
        Assume there exists $\sigma^2 > 0$ (to be specified later) such that
        for any $t\geq 0$,
        \begin{equation}
            \label{eq:bounded_noiseass}
            \bbE\left[ (\hat{\nabla}_{\omega}l(\hat{\omega}_t)-\bbE[\hat{\nabla}_{\omega}l(\hat{\omega}_t)\mid \hat{\omega}_t])^{\top} (\omega_{t+1, *}-\hat{\omega}_{t+1})\middle| \hat{\omega}_t\right] \leq \sigma^2 \varepsilon_t.
        \end{equation}
        Let $\varepsilon_t \leq \frac{1}{L}\wedge \frac{1}{\mu}$ for all $t\geq 0$, let $c_t =c_{t-1} \varepsilon_{t-1}^{-1} (\varepsilon_t^{-1}-\mu)^{-1} $ for $t\geq 1$, and let $c_0 = 1$.
        Let $C_k = \sum_{t=1}^{k}c_{t-1}$ for $k\geq 1$.
        Then, under Assumptions~\ref{ass:minimality_and_regularity},~\ref{ass:convexsmooth}, and the additional bounded-noise assumption~\eqref{eq:bounded_noiseass},
        \begin{align}
            \label{eq:res}
            \begin{split}
                \frac{1}{C_k}\sum_{t=1}^{k} c_{t-1}\bbE[l(\hat{\omega}_t)-l(\omega^*)] \leq \frac{(\varepsilon_0^{-1}-\mu)\uKL(q_{\omega^{*}}\mid q_{\omega_0})}{C_k} &+\sigma^2\sum_{t=0}^{k-1} \frac{c_t\varepsilon_t}{C_k} \\
                &+ \sum_{t=0}^{k-1} \frac{c_t}{C_k}\bbE[ B_t^{\top}   (\omega^*-\hat{\omega}_{t+1}) ].
            \end{split}
        \end{align}
        \begin{proof}
            Assumption~\ref{ass:minimality_and_regularity} allows us to define~\eqref{eq:exact_md_starting_point_sto}.
            Under Assumption~\ref{ass:convexsmooth} and the boundedness of the gradient estimate~\eqref{eq:bounded_noiseass}, we can derive a slightly modified version of the descent lemma~\citep[][Lemma 4.3]{Hanzely2021} which accounts for the presence of the bias.
            Next line follows from the calculations done in the proof of~\citet[Lemma 4.3]{Hanzely2021}:
            \begin{equation}
                \label{eq:descent_applied1}
                \begin{split}
                    \bbE[l(\hat{\omega}_{t+1})-l(\omega^*)\mid \hat{\omega}_t] &\leq \left(\frac{1}{\varepsilon_t} - \mu\right)D_{Z^{*}}(\omega^*, \hat{\omega}_t) - \frac{1}{\varepsilon_t} \bbE[D_{Z^{*}}(\omega^*, \hat{\omega}_{t+1})\mid \hat{\omega}_t]  \\
                    &\quad + \varepsilon_t \sigma^2 - \left(\frac{1}{\varepsilon_t}-L\right)\bbE[D_{Z^{*}}(\hat{\omega}_{t+1}, \hat{\omega}_t)\mid \hat{\omega}_t] + B_t^{\top}  (\omega^*-\omega_{t+1, *}) \\
                    &\quad - \bbE[ B_t^{\top} (\hat{\omega}_{t+1}-\omega_{t+1,*}) \mid  \hat{\omega}_t],
                \end{split}
            \end{equation}
            where $\omega^* = \argmin_{\omega\in\mathcal{W}} l(\omega)$.
            Since $\varepsilon_t^{-1} \geq L$, the fourth term is negative.
            Therefore,~\eqref{eq:descent_applied1} becomes
            \begin{equation}
                \label{eq:descent_applied2}
                \begin{split}
                    \bbE[l(\hat{\omega}_{t+1})-l(\omega^*)\mid \hat{\omega}_t] &\leq \left(\frac{1}{\varepsilon_t} - \mu\right)D_{Z^{*}}(\omega^*, \hat{\omega}_t)- \frac{1}{\varepsilon_t} \bbE[D_{Z^{*}}(\omega^*, \hat{\omega}_{t+1})\mid \hat{\omega}_t]\\
                    &\quad + \varepsilon_t \sigma^2 + B_t^{\top}  (\omega^*-\omega_{t+1, *}) - \bbE[ B_t^{\top}( \hat{\omega}_{t+1}-\omega_{t+1,*}) \mid \hat{\omega}_t].
                \end{split}
            \end{equation}
            Taking the expectation of~\eqref{eq:descent_applied2} gives
            \begin{equation}
                \label{eq:descent_lemma}
                \begin{split}
                    \bbE[l(\hat{\omega}_{t+1})-l(\omega^*)] &\leq \left(\frac{1}{\varepsilon_t} - \mu\right)\bbE[D_{Z^{*}}(\omega^*, \hat{\omega}_t)] - \frac{1}{\varepsilon_t} \bbE[D_{Z^{*}}(\omega^*, \hat{\omega}_{t+1})] \\
                    & \quad + \varepsilon_t \sigma^2 + \bbE[ B_t^{\top} (\omega^*-\hat{\omega}_{t+1})].
                \end{split}
            \end{equation}
            Let $c_t = c_{t-1}\varepsilon_{t-1}^{-1}(\varepsilon_t^{-1}-\mu)^{-1} $ for $t\geq 1$, and let $c_0 = 1$.
            Let $k\geq 1$ and define $C_k = \sum_{t=1}^{k}c_{t-1}$.
            Since $\varepsilon_t \leq \frac{1}{\mu}$, we have $c_t\geq 0$.
            Multiply by $c_t\geq 0$~\eqref{eq:descent_lemma} and sum for $t\in [0, k-1]$, then divide by $C_k$,
            \begin{equation}
                \label{eq:before_res}
                \begin{split}
                    \sum_{t=1}^{k} \frac{c_{t-1}}{C_k}\bbE[l(\hat{\omega}_t)-l(\omega^*)] &\leq \frac{(\varepsilon_0^{-1}-\mu)D_{Z^*}(\omega^{*}, \omega_0)}{C_k} +\sigma^2\sum_{t=0}^{k-1} \frac{c_t\varepsilon_t}{C_k} + \sum_{t=0}^{k-1} \frac{c_t}{C_k}\bbE[ B_t^{\top}   (\omega^*-\hat{\omega}_{t+1} )].
                \end{split}
            \end{equation}
            We essentially recover~\citet[Th. 4.5]{Hanzely2021}, but with the additional bias terms.
            Finally,~\eqref{eq:res} follows from~\eqref{eq:before_res} and $D_{Z^*}(\omega^{*}, \omega_0) = \uKL(q_{\omega^*} \mid q_{\omega_0})$.
        \end{proof}
    \end{lemma}
    \begin{remark}
        The previous lemma requires a boundedness assumption on the gradient estimate given by~\eqref{eq:bounded_noiseass}.
        This assumption is typically required for proving such descent lemmas, see~\citep{wu2024understandingstochasticnaturalgradient, Hanzely2021, pmlr-v162-scaman22a, batardière2024importancesamplingbasedgradientmethod, doi:10.1137/16M1080173}.
        In particular, this assumption is implied, by Cauchy-Schwarz inequality, if the gradient estimate $\hat{\nabla}_{\omega}l$ mapping given by~\eqref{eq:def_stochastic_gradient} has bounded bias and variance, or directly if the gradient estimate is unbiased.
        We will prove it is satisfied under Assumptions~\ref{ass:bounded_iterates} and~\ref{ass:moment_two_f} in Lemma~\ref{lemma:bound_variance}.
    \end{remark}
    \begin{lemma}[Controlling the bias terms in~\eqref{eq:res}]
        \label{lemma:controlbias}
        Under Assumptions~\ref{ass:minimality_and_regularity},~\ref{ass:bounded_iterates} and~\ref{ass:moment_two_f}, for any $t\geq 0$
        \begin{equation}
            \begin{split}
                \label{eq:desired_resultlemma9}
                \bbE[ B_t^{\top}( \omega^*-\hat{\omega}_{t+1})] \leq \OO(N^{-1}) \times m_2 m^{1/2}\mu_4 \times \left(\bbE[\norm{ \omega^*-\hat{\omega}_{t+1}}^2]\right)^{1/2}.
            \end{split}
        \end{equation}
        \begin{proof}
            By Cauchy Schwarz inequality, for any $t\geq 0$,
            \begin{align}
                \label{eq:application_cauchyschwarz}
                \bbE[ B_t^{\top} (\omega^*-\hat{\omega}_{t+1})] \leq \left(\bbE[\norm{B_t}^2]\bbE[\norm{ \omega^*-\hat{\omega}_{t+1}}^2]\right)^{1/2}.
            \end{align}
            Conditionally on $\mathcal{A}(\hat{\omega}_t)$, $\hat{F}_{\hat{\omega}_t}$ is invertible, and by Lemma~\ref{lemma:polynomial}, there exists $C> 0$ such that for $N$ large enough, $N \norm{\bbE[F^{-1}_{\hat{\omega}_t}-\hat{F}^{-1}_{\hat{\omega}_t}\mid \mathcal{A}(\hat{\omega}_t), \hat{\omega}_t]}\leq C$.
            Consequently,
            \begin{equation}
                \begin{split}
                    \label{eq:boundnormbiassquared}
                    N^2\norm{B_t}^2 &= N^2\norm{\bbE[F^{-1}_{\hat{\omega}_t}z_{\hat{\omega}_t}-\hat{F}^{-1}_{\hat{\omega}_t} \hat{z}_{\hat{\omega}_t}\mid \mathcal{A}(\hat{\omega}_t), \hat{\omega}_t]}^2\\
                    &=N^2 \norm{\bbE[F^{-1}_{\hat{\omega}_t}-\hat{F}^{-1}_{\hat{\omega}_t} \mid \mathcal{A}(\hat{\omega}_t), \hat{\omega}_t] z_{\hat{\omega}_t}}^2 \\
                    &\leq N^2 \norm{\bbE[F^{-1}_{\hat{\omega}_t}-\hat{F}^{-1}_{\hat{\omega}_t}\mid \mathcal{A}(\hat{\omega}_t), \hat{\omega}_t]}^2 \norm{z_{\hat{\omega}_t}}^2\\
                    &\leq C^2 \norm{z_{\hat{\omega}_t}}^2\\
                    &\leq C^2 m_2^2 m  \mu_4^2,
                \end{split}
            \end{equation}
            where we used to go from the first to the second line, the independency of $\hat{z}_{\hat{\omega}_t}$ with $\hat{F}_{\hat{\omega}_t}$ conditioned on $\hat{\omega}_t$ and $\bbE[\hat{z}_{\hat{\omega}_t}\mid \hat{\omega}_t] = z_{\hat{\omega}_t}$, and to go from the third to the fourth line, we use Lemma~\ref{lemma:polynomial}, requiring Assumptions~\ref{ass:minimality_and_regularity} and~\ref{ass:bounded_iterates}, and to go from the fourth to the fifth line, we use bounds on the moment of $s^2$ and $f^2$ given by~\eqref{eq:fourth_order_moment} and Assumption~\ref{ass:moment_two_f}:
            \begin{equation}
                \begin{split}
                    \norm{z_\omega} &\leq (\bbE\norm{s}^2)^{1/2} (\bbE f^2)^{1/2}\\
                    &\leq \left(m\sup_{\omega\in \mathcal{W}_0} \max_{1\leq i\leq m} \bbE\abs{s}_i^2\right)^{1/2} \times m_2\\
                    &\leq m^{1/2} \mu_4 m_2.
                \end{split}
            \end{equation}
            Taking the expectation of~\eqref{eq:boundnormbiassquared} and plugging it into~\eqref{eq:application_cauchyschwarz} yields~\eqref{eq:desired_resultlemma9}.
        \end{proof}
    \end{lemma}

    \begin{lemma}[High-probability uniform bound for the variance of the gradient estimate]
        \label{lemma:bound_variance}
        Let $\varepsilon > 0$.
        For any $\omega \in \mathcal{W}_0$, let $M(\omega) = \argmin_{\omega'\in \mathcal{W}}\{\nabla_\omega^{\top}l(\omega)\omega' + \varepsilon^{-1} D_{Z^*}(\omega', \omega)\}$ be the exact mirror-descent iterate starting from $\omega$ with step size $\varepsilon$.
        Similarly, let $\hat{M}(\omega)$ be the mirror-descent using gradient estimate $\hat{\nabla}_\omega l$ given by~\eqref{eq:def_stochastic_gradient}.
        Let $\sigma^2(\omega)$ be defined by
        \begin{equation}
            \label{eq:boundonsigmasq}
            \sigma^2(\omega) = \varepsilon^{-1}\bbE[(\hat{\nabla}_{\omega}l(\omega)-\bbE[\hat{\nabla}_\omega l(\omega)])^{\top}(M(\omega)-\hat{M}(\omega))],
        \end{equation}
        where all the expectations are taken conditioned on $\mathcal{A}(\omega)$.
        Under Assumptions~~\ref{ass:minimality_and_regularity},~\ref{ass:bounded_iterates}, and~\ref{ass:moment_two_f},
        there exists some constant $C>0$, such that for $N$ large enough (see Lemma~\ref{lemma:polynomial}), and any $\omega\in \mathcal{W}_0$ such that $M(\omega), \hat{M}(\omega)\in\mathcal{W}_0$,
        \begin{equation}
            \sigma^2(\omega) \leq N^{-1}C.
        \end{equation}
        \begin{proof}
            To obtain the uniform bound on~\eqref{eq:boundonsigmasq}, we independently bound both terms in the scalar product.
            Let $\omega\in\mathcal{W}_0$, and let $\eta$ be the corresponding natural parameter, $\eta = \eta(\omega)$.
            Conditionally on $\mathcal{A}(\omega)$, by the computations done in the
            proof of Lemmas~\ref{lemma:inverse_g_compactA}
            and~\ref{lemma:polynomial} there exists a constant $C>0$ such that for
            $N$ large enough, $N\bbE[\norm{\hat{F}^{-1}_\omega - F^{-1}_\omega}^2]
            \leq C$, i.e., the mean-square error is $\OO(1/N)$, and thus the variance $\bbE[\norm{\hat{F}_\omega^{-1}-\bbE[\hat{F}_{\omega}^{-1}]}^2]$ is $\OO(1/N)$.

            Using the definition of the stochastic gradient $\hat{\nabla}l$ given by~\eqref{eq:def_stochastic_gradient}, and the previous bound, the first term is bounded by:
            \begin{equation}
                \label{eq:boud1}
                \begin{split}
                    N\bbE[\norm{\bbE[\hat{\nabla}_\omega l(\omega)] - \hat{\nabla}_\omega l }^2] &= N \bbE \left[\norm{(\bbE[ \hat{F}_\omega^{-1}] - \hat{F}^{-1}_\omega)z_{\omega}+\hat{F}_\omega^{-1}(z_\omega - \hat{z}_\omega)}^2 \right]\\
                    &\leq 2 N \left(\bbE[\norm{(\bbE [\hat{F}_\omega^{-1}]-\hat{F}_\omega^{-1})z_{\omega}}^2] + \bbE[\norm{\hat{F}_\omega^{-1}(z_{\omega}-\hat{z}_\omega)}^2] \right)\\
                    &\leq 2N \big( \norm{z_{\omega}}^2 \times \bbE[\norm{(\bbE[\hat{F}_\omega^{-1}]-\hat{F}_\omega^{-1})}^2] \\
                    &\quad + \bbE[\norm{\hat{F}_\omega^{-1}(z_{\omega}-\hat{z}_\omega)}^2]\big)\\
                    &\leq 2N \left(m\mu_4^2 m_2^2 \times C/N + \bbE\norm{\hat{F}_\omega^{-1}(z_{\omega}-\hat{z}_\omega)}^2 \right),
                \end{split}
            \end{equation}
            where we used $(a+b)^2\leq 2a^2 + 2b^2$ from the first to the second line, from the second to the third line the sub-multiplicativity of $\norm{\cdot}$, from the third to the fourth $\norm{z_\omega}^2 \leq m\mu_4^2 m_2^2$ and the $\OO(1/N)$ bound on the variance of $\hat{F}_\omega^{-1}$.
            Furthermore, we can bound the last term in~\eqref{eq:boud1} by
            \begin{equation}
                \label{eq:boud2}
                \begin{split}
                    N\bbE[\norm{\hat{F}_\omega^{-1}(z_{\omega}-\hat{z}_\omega)}^2] &\leq N\bbE[\norm{\hat{F}_\omega^{-1}}^2]\bbE[\norm{z_{\omega}-\hat{z}_\omega}^2]\\
                    &\leq 2N\{\bbE[\norm{\hat{F}_\omega^{-1}-F_\omega^{-1}}^2 + \bbE\norm{F_\omega^{-1}}^2]\}\bbE[\norm{z_\omega - \hat{z}_\omega}^2]\\
                    &\leq 2N ( r^{-2} + C/N ) \times \bbE[\norm{z_{\omega}-\hat{z}_\omega}^2]\\
                    &\leq 2N( r^{-2} + C/N)\times C'/N
                \end{split}
            \end{equation}
            where we used that $r^{-2} = \sup_{\omega} \norm{F_\omega^{-1}}^2$, and $
            \bbE[\norm{z_{\omega}-\hat{z}_\omega}^2] = \text{Tr}(\textup{Cov}(sf(X)))/N$
            with~\eqref{eq:fourth_order_moment} and Assumption~\ref{ass:moment_two_f} to state that there exists $C'>0$ independent on $\omega$ such that $\bbE[\norm{z_{\omega}-\hat{z}_\omega}^2]\leq C'/N$.
            Gathering~\eqref{eq:boud1} and~\eqref{eq:boud2} yields for the first term of the scalar product:
            \begin{equation}
                \label{eq:eqfirstterm}
                N\bbE[\norm{\bbE[\hat{\nabla}_\omega l(\omega)] - \hat{\nabla}_\omega l }^2] \leq 2Cm\mu_4^2 m_2^2 + 4C'(r^{-2}+C/N),
            \end{equation}
            which in turn can be bounded by some constant $C''>0$ independent of $\omega$.
            Let us tackle the second term inside the scalar product.
            By Proposition~\ref{prop:gradient_descent},
            \begin{equation}
                M(\omega) = \omega(\eta- \varepsilon\nabla_{\omega}l(\omega)),
            \end{equation}
            and similarly for $\hat{M}$,
            \begin{equation}
                \hat{M}(\omega) = \omega(\eta - \varepsilon \hat{\nabla}_{\omega}l(\omega)).
            \end{equation}

            Under Assumption~\ref{ass:minimality_and_regularity}, the mapping $\omega : \eta\in \mathcal{V}\mapsto \omega(\eta)$ is differentiable with $\nabla_\eta\omega= Z_\eta F_{\eta}=\int s s^{\top}q_{\eta}$, see~\eqref{eq:nablaetaomega}.
            Thus, $Z^*$ is twice differentiable with $\nabla^2 Z^*(\omega) = (Z_{\omega}F_{\omega})^{-1}$, whose smallest eigenvalue is at least $\Lambda^{-1}$ for any $\omega\in\mathcal{W}_0$, where $\Lambda\coloneq \sup_{\omega\in\mathcal{W}_0}Z_\omega\norm{F_\omega}<\infty$ since $\mathcal{W}_0$ is compact.
            Since $\mathcal{W}_0$ is convex and contains $M(\omega)$ and $\hat{M}(\omega)$, and since $\nabla Z^*(M(\omega))-\nabla Z^*(\hat{M}(\omega)) = \varepsilon(\hat{\nabla}_\omega l(\omega)- \nabla_{\omega}l(\omega))$,
            \begin{equation}
                \label{eq:eqsecondterm}
                \Lambda^{-1}\norm{M(\omega)-\hat{M}(\omega)}^2 \leq \varepsilon (\hat{\nabla}_\omega l(\omega)- \nabla_{\omega}l(\omega))^{\top}(M(\omega)-\hat{M}(\omega)) \leq \varepsilon \norm{\hat{\nabla}_\omega l(\omega)- \nabla_{\omega}l(\omega)}\norm{M(\omega)-\hat{M}(\omega)}.
            \end{equation}
            Consequently,
            \begin{equation}
                \begin{split}
                    \label{eq:secondtterm}
                    \sqrt{N}(\bbE[\norm{M(\omega)-\hat{M}(\omega)}^2])^{1/2} &\leq \sqrt{N}\varepsilon \Lambda (\bbE[\norm{\hat{\nabla}_\omega l(\omega)-\nabla_{\omega}l(\omega)}^2])^{1/2}\\
                    &\leq \varepsilon \Lambda C^{(3)},
                \end{split}
            \end{equation}
            for some constant $C^{(3)}>0$, where we used~\eqref{eq:eqfirstterm} and Lemma~\ref{lemma:polynomial} to bound the variance and the bias of $\hat{\nabla}_\omega l(\omega)$.
            By Cauchy Schwarz inequality,~\eqref{eq:eqfirstterm} and~\eqref{eq:secondtterm}, for $N$ large enough,
            \begin{equation}
                \label{eq:bounds2}
                \begin{split}
                    N\sigma^2(\omega)&\leq \varepsilon^{-1}N (\bbE[\norm{\bbE[\hat{\nabla}_\omega l(\omega)] - \hat{\nabla}_{\omega}l(\omega)}^2] \bbE[\norm{M(\omega)-\hat{M}(\omega)}^2])^{1/2}\\
                    &\leq \Lambda C^{(3)}\sqrt{ C''}.
                \end{split}
            \end{equation}
            This concludes the proof.
        \end{proof}
    \end{lemma}

    With previous Lemmas~\ref{lemma:equivalence},~\ref{lemma:technicallemma},~\ref{lemma:controlbias}, and~\ref{lemma:bound_variance} in hand, we can prove the main result.
    \begin{proof}[Proof of Theorem~\ref{th:main_theorem}]

        Define $\bar{\omega}_k$ as given in the theorem.
        By convexity of $l$,
        \begin{equation}
            \label{eq:convexityl}
            l(\bar{\omega}_k)-l(\omega^*)\leq \frac{1}{C_k}\sum_{t=1}^{k}c_{t-1} (l(\hat{\omega}_{t})-l(\omega^*)).
        \end{equation}
        Let $\tilde{\bbE}$ denote the expectation of Lemmas~\ref{lemma:technicallemma},~\ref{lemma:controlbias} and~\ref{lemma:bound_variance}, i.e., when the samples of each iteration $t$ are drawn conditionally on $\mathcal{A}(\hat{\omega}_t)$.
        Then $\bbE[h\mathds{1}_{\mathcal{A}_k}]\leq \tilde{\bbE}[h]$ for any non-negative function $h$.
        Combining Lemma~\ref{lemma:technicallemma} with Lemma~\ref{lemma:controlbias} to control the bias terms, we find that the expectation of the RHS in~\eqref{eq:convexityl} is upper bounded by
        \begin{equation}
            \label{eq:biginequality}
            \begin{split}
                \bbE[(l(\bar{\omega}_k)-l(\omega^*))\mathds{1}_{\mathcal{A}_k}] &\leq \frac{(\varepsilon_0^{-1}-\mu)\uKL(q_{\omega^*}\mid q_{\omega_0})}{C_k} \\
                &\quad +\sigma^2\sum_{t=0}^{k-1} \frac{c_t\varepsilon_t}{C_k} + \OO(N^{-1})\times S_{k,N},
            \end{split}
        \end{equation}
        where $S_{k,N}\coloneq \frac{ m_2 m^{1/2}\mu_4}{C_k}\sum_{t=0}^{k-1}c_{t} (\tilde{\bbE}[\norm{\omega^*-\hat{\omega}_{t+1}}^2])^{1/2}$, where the big-$\OO$ term is independent of $k$ since it is independent of $\omega_0=\hat{\omega}_0, \hat{\omega}_1, \ldots \hat{\omega}_k$,
        and $\sigma^2$ satisfying the assumption of Lemma~\ref{lemma:technicallemma}, for all $t\leq k$:
        \begin{equation}
            \tilde{\bbE}\left[(\hat{\nabla}_\omega l(\hat{\omega}_t)-\tilde{\bbE}[\hat{\nabla}_\omega l(\hat{\omega}_t)\mid \hat{\omega}_t])^{\top}(\omega_{t+1,*}-\hat{\omega}_{t+1})\middle| \hat{\omega}_t\right] \leq \sigma^2 \varepsilon_t.
        \end{equation}
        By Lemma~\ref{lemma:bound_variance} and Assumption~\ref{ass:bounded_iterates}, we can set
        \begin{equation}
            \sigma^2 = N^{-1}C,
        \end{equation}
        for some constant $C$ independent on $N$ and the sequence $\hat{\omega}_0, \ldots, \hat{\omega}_{k-1}$.

        Let us tackle the terms which depend both upon $N$ and $k$ via the sequence $\hat{\omega}_0, \hat{\omega}_1, \ldots, \hat{\omega}_{k}$.
        By Assumption~\ref{ass:bounded_iterates}, for any $0\leq t\leq k$, $\norm{\omega^* - \hat{\omega}_t}\leq \rho\coloneq \sup_{\omega\in\mathcal{W}_0}\norm{\omega^*-\omega}<\infty$.
        Using $\sum_{t=0}^{k-1}c_t = C_k$, and bounding uniformly the summands of $S_{k,N}$ yields
        \begin{equation}
            \label{eq:limacontroler}
            \begin{split}
                S_{k,N} \leq m_2 m^{1/2}\mu_4  \rho.
            \end{split}
        \end{equation}
        Finally, plugging~\eqref{eq:limacontroler} into~\eqref{eq:biginequality} yields the uniform bound over $k$:
        \begin{equation}
            \label{eq:eq_theorem1}
            \bbE[(l(\bar{\omega}_k)-l(\omega^*))\mathds{1}_{\mathcal{A}_k}] \leq \frac{(\varepsilon_0^{-1}-\mu)\uKL(q_{\omega^*}\mid q_{\omega_0})}{C_k} + \OO(N^{-1})\sum_{t=0}^{k-1}\frac{c_{t}\varepsilon_t}{C_k} + \OO(N^{-1}).
        \end{equation}
        All the constants in the big-$\OO$ terms can be chosen independently on the sequence of $\hat{\omega}$.

        Using Lemma~\ref{lemma:polynomial} with $\delta/(k+1)$ and a union bound, we have $P(\cap_{t=0}^k\mathcal{A}(\hat{\omega}_t)) \geq 1-\delta$ for the chosen $N$.

        Finally, let us prove the explicit convergence rates for linearly decreasing stepsizes $\varepsilon_t = (L+\alpha t)^{-1}$, $t\geq 0$.
        Similarly to~\citet[Lemma 4.8][]{Hanzely2021}, we distinguish three cases depending on $\alpha$ compared to $\mu$.
        If $\alpha>\mu$, then $C_k = \Omega(k^{\mu/\alpha})$ and $\sum_{t=0}^{k-1} c_t\varepsilon_t = \OO(1)$ which yields $\OO(k^{-\mu/\alpha})+\OO(N^{-1})$ for the RHS of~\eqref{eq:eq_theorem}.
        If $\alpha=\mu$, then $C_k = \Omega(k)$ and $\sum_{t=0}^{k-1}c_t\varepsilon_t = O(\log(k))$.
        If $\alpha<\mu$, then $C_k = \Omega(k^{\mu/\alpha})$ and $\sum_{t=0}^{k-1}c_t \varepsilon_t=\OO(k^{\mu/\alpha-1})$.
    \end{proof}

    \newpage

    \section*{NeurIPS Paper Checklist}

    \begin{enumerate}

        \item {\bf Claims}
        \item[] Question: Do the main claims made in the abstract and introduction accurately reflect the paper's contributions and scope?
        \item[] Answer: \answerYes{} %
        \item[] Justification: The main contributions of the paper are as follows:
        (i) KL minimisation within exponential families can be performed via successive linear regressions under tempered variational approximations. This approach is equivalent to natural gradient descent (NGD) and mirror descent (MD) but avoids the need for explicit gradient-based procedures. This is detailed in Section~\ref{sec:lsvi} and Proposition~\ref{prop:gradient_descent}, with relevant prior work cited.
        (ii) In the Gaussian variational family, exact LSVI can be tailored to eliminate the need to invert the Fisher information matrix. The resulting procedures have computational complexity $\OO(d^3)$ in the full-covariance case and $\OO(d)$ in the mean-field case, as shown in Section~\ref{sec:gaussian}, specifically Theorems~\ref{th:full_cov_gaussian} and~\ref{th:mfg}.
        (iii) Under standard optimization assumptions, LSVI converges at explicit rates, established in Theorem~\ref{th:main_theorem} (Section~\ref{sec:stochastic}).
        (iv) Empirical results demonstrate that LSVI achieves performance comparable to state-of-the-art variational inference methods and remains effective for non-differentiable target densities, as shown in Section~\ref{sec:experiments} and Appendix~\ref{app:extra_numerics}.

        \item {\bf Limitations}
        \item[] Question: Does the paper discuss the limitations of the work performed by the authors?
        \item[] Answer: \answerYes{} %
        \item[] Justification: (i) The approach is currently restricted to exponential families. Extending LSVI to more general variational families, such as mixtures of exponential families, is a promising direction for future work.
        (ii) Convergence guarantees for the stochastic versions of LSVI, namely, LSVI-MF and LSVI-FC, are not yet established. (ii) A comprehensive theoretical analysis of LSVI's convergence when the target distribution $\pi$ is replaced by an unbiased estimator $\hat{\pi}$ (e.g., via subsampling) remains an open problem. (iii) A comprehensive study of the constants involved in the convergence rates, in particular with respect to the smallest singular value of the FIM $r$, the latent dimension $d$ and the dimension of the statistic $m$ is left for future work. We believe most of the proofs can be adapted but the analysis would be more involved. See Section~\ref{sec:limitations}.

        \item {\bf Theory assumptions and proofs}
        \item[] Question: For each theoretical result, does the paper provide the full set of assumptions and a complete (and correct) proof?
        \item[] Answer: \answerYes{} %
        \item[] Justification: All assumptions are provided in the core manuscript, see~\cref{ass:minimality_and_regularity},~\cref{ass:convexsmooth},~\cref{ass:bounded_iterates} and~\cref{ass:moment_two_f}. Furthermore, the assumptions are discussed in the core of the paper along with references mentioning existing similar assumptions in the VI literature. The proofs are deferred to the supplementary materials and are divided in several comprehensive steps, including lemmas in order to make the proof-reading procedure easier.

        \item {\bf Experimental result reproducibility}
        \item[] Question: Does the paper fully disclose all the information needed to reproduce the main experimental results of the paper to the extent that it affects the main claims and/or conclusions of the paper (regardless of whether the code and data are provided or not)?
        \item[] Answer: \answerYes{} %
        \item[] Justification: Section~\ref{sec:experiments} and Appendix~\ref{app:extra_numerics} along with the pseudo-code Algorithms given in Section~\ref{sec:stochastic} are sufficient to reproduce all the experimental results. In particular, all input parameters are provided in~\cref{tab:schedules}.

        \item {\bf Open access to data and code}
        \item[] Question: Does the paper provide open access to the data and code, with sufficient instructions to faithfully reproduce the main experimental results, as described in supplemental material?
        \item[] Answer: \answerYes{} %
        \item[] Justification: The paper provides a Python (JAX) package that includes all discussed Algorithms (LSVI~\cref{alg:lsvi_krasnoselskii}, MF-LSVI~\cref{alg:lsvi_krasnoselskii_mfg}, FC-LSVI~\cref{alg:lsvi_krasnoselskii_fullcov}, Variance control for the stepsizes~\cref{alg:recipe_backracking_variance_control}, NGD with details provided in~\cref{app:extra_numerics}) as well as scripts to reproduce all the listed experiments. The package is explicitly divided into two parts, \texttt{variational} contains the generic implementations while \texttt{experiments} contains three sub-folders for the three distinct variational problems (logistic regression, variable selection and Bayesian synthetic likelihood). Full-pipeline for the experiments is provided (download and pre-processing of the datasets, inference procedures and post-processing scripts).

        \item {\bf Experimental setting/details}
        \item[] Question: Does the paper specify all the training and test details (e.g., data splits, hyperparameters, how they were chosen, type of optimizer, etc.) necessary to understand the results?
        \item[] Answer: \answerYes{} %
        \item[] Justification: The paper provides all the necessary details to reproduce the experiments, including the hyperparameters (the number of samples $N$, the number of iterations $T$, the initialisation distributions and the schedules) which are given in~\cref{app:extra_numerics}. Different schedules have been considered to demonstrate robustness of the proposed methods while the number of samples is set to obtain reasonable numerical stability.

        \item {\bf Experiment statistical significance}
        \item[] Question: Does the paper report error bars suitably and correctly defined or other appropriate information about the statistical significance of the experiments?
        \item[] Answer: \answerYes{} %
        \item[] Justification: All experiments were conducted using multiple trials as indicated in the figure labels and the Appendix~\ref{app:extra_numerics}. For the logistic regression examples, one standard-deviation confidence intervals are provided over $100$ independent realisations. For the variable selection problem, the means and the min-max intervals for the posterior marginal probabilities obtained via LSVI over $100$ independent realisations. No statistical assumption is made for uncertainty measurement.
        \item[] Guidelines:
        \begin{itemize}
            \item The answer NA means that the paper does not include experiments.
            \item The authors should answer "Yes" if the results are accompanied by error bars, confidence intervals, or statistical significance tests, at least for the experiments that support the main claims of the paper.
            \item The factors of variability that the error bars are capturing should be clearly stated (for example, train/test split, initialization, random drawing of some parameter, or overall run with given experimental conditions).
            \item The method for calculating the error bars should be explained (closed-form formula, call to a library function, bootstrap, etc.)
            \item The assumptions made should be given (e.g., Normally distributed errors).
            \item It should be clear whether the error bar is the standard deviation or the standard error of the mean.
            \item It is OK to report 1-sigma error bars, but one should state it. The authors should preferably report a 2-sigma error bar than state that they have a 96\% CI, if the hypothesis of Normality of errors is not verified.
            \item For asymmetric distributions, the authors should be careful not to show in tables or figures symmetric error bars that would yield results that are out of range (e.g. negative error rates).
            \item If error bars are reported in tables or plots, The authors should explain in the text how they were calculated and reference the corresponding figures or tables in the text.
        \end{itemize}

        \item {\bf Experiments compute resources}
        \item[] Question: For each experiment, does the paper provide sufficient information on the computer resources (type of compute workers, memory, time of execution) needed to reproduce the experiments?
        \item[] Answer: \answerYes{} %
        \item[] Justification: The full hardware and software specifications are provided in Appendix~\ref{app:extra_numerics}, specifically in Table~\ref{table:perf}, along with Figures~\ref{fig:sonar} and~\ref{fig:mnist_runtime}, which report experiments runtime and memory usage. All performance statistics are computed using independent realisations for improved robustness. In addition, scripts for measuring the runtime and memory usage of the algorithms can be found in the package (\texttt{/experiments/\{...\}/time.py}). All experiments were successfully performed and are reported in Section~\ref{sec:experiments} and Appendix~\ref{app:extra_numerics}. There is the exception of the applicability of ADVI (PyMC3,~\citep[][]{Salvatier2016}) on the MNIST dataset, which is explicitly stated in Appendix~\ref{app:extra_numerics}. Instead, ADVI as provided by Blackjax~\citep{cabezas2024blackjaxcomposablebayesianinference} was used as a replacement to PyMC3.

        \item {\bf Code of ethics}
        \item[] Question: Does the research conducted in the paper conform, in every respect, with the NeurIPS Code of Ethics \url{https://neurips.cc/public/EthicsGuidelines}?
        \item[] Answer: \answerYes{} %
        \item[] Justification: We carefully read through the NeurIPS Code of Ethics, and we see no violation of any guideline.

        \item {\bf Broader impacts}
        \item[] Question: Does the paper discuss both potential positive societal impacts and negative societal impacts of the work performed?
        \item[] Answer: \answerNA{} %
        \item[] Justification: The paper focuses on VI methods, emphasizing theoretical analysis and algorithmic implementability. As such, the work is foundational in nature and does not directly pertain to real-world applications or deployments. Given its abstract and theoretical scope, it does not present identifiable positive or negative societal impacts, including concerns related to fairness, privacy, or misuse.

        \item {\bf Safeguards}
        \item[] Question: Does the paper describe safeguards that have been put in place for responsible release of data or models that have a high risk for misuse (e.g., pretrained language models, image generators, or scraped datasets)?
        \item[] Answer: \answerNA{} %
        \item[] Justification: We see no risk in the application of variational inference procedures.

        \item {\bf Licenses for existing assets}
        \item[] Question: Are the creators or original owners of assets (e.g., code, data, models), used in the paper, properly credited and are the license and terms of use explicitly mentioned and properly respected?
        \item[] Answer: \answerYes{} %
        \item[] Justification: All used Python packages are open-source, have permissive licenses, and are explicitly mentioned both in the manuscript and the code (\texttt{pyproject.toml} with complete dependency specifications). Specifically, Blackjax, PyMC, and JAX are mentioned in Section~\ref{sec:Motivation} and Section~\ref{sec:experiments}. Details on the datasets used, licenses and download links, are provided in Appendix~\ref{app:extra_numerics}.

        \item {\bf New assets}
        \item[] Question: Are new assets introduced in the paper well documented and is the documentation provided alongside the assets?
        \item[] Answer: \answerYes{} %
        \item[] Justification: The provided Python (JAX) package for LSVI is well documented and includes a README file with instructions for installation and usage. The license is also included in the package (Apache License 2.0). In addition, we provide usage examples and accompanying commentaries.

        \item {\bf Crowdsourcing and research with human subjects}
        \item[] Question: For crowdsourcing experiments and research with human subjects, does the paper include the full text of instructions given to participants and screenshots, if applicable, as well as details about compensation (if any)?
        \item[] Answer: \answerNA{} %
        \item[] Justification: No experiment involving human subjects were conducted.

        \item {\bf Institutional review board (IRB) approvals or equivalent for research with human subjects}
        \item[] Question: Does the paper describe potential risks incurred by study participants, whether such risks were disclosed to the subjects, and whether Institutional Review Board (IRB) approvals (or an equivalent approval/review based on the requirements of your country or institution) were obtained?
        \item[] Answer: \answerNA{} %
        \item[] Justification: No experiment involving human subjects were conducted.

        \item {\bf Declaration of LLM usage}
        \item[] Question: Does the paper describe the usage of LLMs if it is an important, original, or non-standard component of the core methods in this research? Note that if the LLM is used only for writing, editing, or formatting purposes and does not impact the core methodology, scientific rigorousness, or originality of the research, declaration is not required.
        \item[] Answer: \answerNA{} %
        \item[] Justification: There is no mention of LLMs in the manuscript, and no LLM was used.

    \end{enumerate}

\end{document}